\documentclass[UTF-8,reqno]{amsart}
\usepackage{enumerate}
\usepackage{mhequ}
\usepackage{amssymb,url,color, booktabs,nccmath}
\usepackage[left=2.7cm,right=2.7cm,top=3.5cm,bottom=3.5cm]{geometry}
\usepackage{mathrsfs}
\usepackage{enumitem,dsfont}
\usepackage{graphicx}
\usepackage{tikz}
\usepackage{tocvsec2}
\usetikzlibrary{shapes,snakes}
\usetikzlibrary{calc}
\usetikzlibrary{decorations.shapes}
\usepackage{amsmath,amsfonts, verbatim, amsthm,amssymb,bbm,mathtools,mathrsfs,authblk, esint}
\usepackage{bm}
\usepackage{extarrows}

\DeclareMathOperator\supp{supp}
\usepackage{color}
\usepackage[colorlinks=true]{hyperref}
\hypersetup{
	linkcolor=blue,          
	citecolor=red,        
	filecolor=blue,      
	urlcolor=cyan
}

\definecolor{darkergreen}{rgb}{0.0, 0.5, 0.0}
\definecolor{myblue}{RGB}{0,0,139} 

\usepackage{fancyhdr} 
\numberwithin{equation}{section}
\def\theequation{\arabic{section}.\arabic{equation}}
\newcommand{\be}{\begin{eqnarray}}
	\newcommand{\ee}{\end{eqnarray}}
\newcommand{\ce}{\begin{eqnarray*}}
	\newcommand{\de}{\end{eqnarray*}}
\newtheorem{theorem}{Theorem}[section]
\newtheorem{lem}[theorem]{Lemma}
\newtheorem{prop}[theorem]{Proposition}

\newtheorem{assumption}{Assumption}[section]
\newtheorem{definition}[theorem]{Definition}
\theoremstyle{definition}
\newtheorem{remark}[theorem]{Remark}

\newcommand{\RR}{\mathbb{R}} 
\newcommand{\fH}{\mathcal{H}} %
\newcommand{\EE}{\mathbb{E}} 
\newcommand{\PP}{\mathbb{P}} 
\newcommand{\fM}{\mathcal{M}} 
\newcommand{\fF}{\mathcal{F}} 
\newcommand{\1}{\mathbbm{1}} 
\newcommand{\fP}{\mathcal{P}} 
\newcommand{\fI}{\mathcal{I}}

\newcommand{\ZZ}{\mathbb{Z}} 
\newcommand{\HH}{\mathbb{H}} 
\newcommand{\WW}{\mathbb{W}}

\newcommand{\fG}{\mathcal{G}}

\newcommand{\NN}{\mathbb{N}}
\newcommand{\fN}{\mathcal{N}}

\newcommand{\cF}{\mathfrak{F}}
\newcommand{\cH}{\mathfrak{H}}

\newcommand{\fR}{\mathcal{R}}

\newcommand{\fV}{\mathcal{V}}
\newcommand{\fU}{\mathcal{U}}

\newcommand{\BB}{\mathbb{B}}
\newcommand{\CC}{\mathbb{C}}
\newcommand{\fT}{\mathcal{T}}
\newcommand{\me}{\mathfrak e}
\newcommand{\fg}{\mathfrak{g}}
\newcommand{\ud}{\operatorname{d}\! }
\newcommand{\Tr}{\textnormal{Tr}}
\newcommand{\cZ}{\mathscr{Z}}
\newcommand{\cK}{\mathfrak{K}}

\newcommand{\fm}{\mathbf m}
\newcommand{\ii}{\boldsymbol i}

\DeclareMathOperator{\Fix}{Fix}

\tikzset{
	dot/.style={circle,fill=black,inner sep=0pt, outer sep=0.7pt, minimum size=1mm},
	dot1/.style={circle,draw=black,inner sep=0pt, outer sep=0.7pt, minimum size=1mm},
	Phi1/.style={white!40!red,thick,snake=coil,segment amplitude=0.6pt, segment length=2pt},
	Phi2/.style={white!40!blue,thick,snake=coil,segment amplitude=0.6pt, segment length=2pt},
	Phi3/.style={white!40!purple,thick,snake=coil,segment amplitude=0.6pt, segment length=2pt},
	Z/.style={black!40!green,thick,snake=coil,segment amplitude=0.6pt, segment length=2pt},
	C/.style={thick,black!20!black},
	Cr/.style={thick, densely dashed,black!20!black},
	Cg/.style={thick, densely dashed,black!20!green},
	Co/.style={thick,black!20!orange},
	Cg1/.style={thick,black!20!green},
}

\def\Wick#1{\,\colon\!\! #1 \!\colon}
\newcommand{\authorfootnotes}{\renewcommand\thefootnote{\@fnsymbol\c@footnote}}%
\begin{document}
\title{$\fP(\Phi)_2$ Theory from many-body quantum Gibbs states}
\maketitle
\begin{center}
	Phan Th\`anh Nam\footnote{LMU Munich, Department of Mathematics, Theresienstrasse 39, 80333 Munich, Germany. Email: nam@math.lmu.de},
	Zhilin Yang \footnote{Academy of Mathematics and Systems Science, Chinese Academy of Sciences, Beijing 100190, China. E-mail: yangzhilin0112@163.com}, 
	Xiangchan Zhu\footnote{State Key Laboratory of Mathematical Sciences, Academy of Mathematics and Systems Science, Chinese Academy of Sciences, Beijing 100190, China. E-mail: zhuxiangchan@126.com}
\end{center}
\begin{abstract} 
	
We derive the $\fP(\Phi)_2$ measure on the two-dimensional unit torus as the rigorous limit of many-body quantum Gibbs states. In particular, the quantum model corresponding to the $\Phi^{2p}_2$ measure  is formulated in the grand-canonical ensemble with a general symmetric $p$-body interaction potential which is not required to have a factorized form. Unlike in the $\Phi^4_2$ case investigated by Fröhlich--Knowles--Schlein--Sohinger in \cite{FKSS25}, in the treatment of higher-order interactions, Wick renormalization generates a full hierarchy of lower-order interactions which we organize systematically using a graphical formalism.  One key ingredient of our analysis is a logarithmic stability estimate for both the nonlocal Hartree functional and the many-body Hamiltonian that is uniform in the interaction range, allowing us to use the positivity of the leading $p$-body interaction to control all lower-order terms generated by the Wick counterterms.

\end{abstract}

\tableofcontents
\section{Introduction}
Let $\Lambda=[-\tfrac12,\tfrac12)^2$ be the two-dimensional torus and let $\mu_0$ be the centered complex Gaussian field with covariance $h^{-1}=(1-\Delta)^{-1}$. For a radial polynomial
\[
\fP(z)=\sum_{r=1}^{p}a_r|z|^{2r},\quad a_p>0,
\]
the complex-valued $\fP(\Phi)_2$ measure is the probability measure
\begin{align*}
	\ud\mu_{\fP}(u)=\cZ_\fP^{-1}e^{-V_\fP(u)}\,\ud\mu_0(u),
	\quad
	V_\fP(u)=\sum_{r=1}^{p}a_r\int_\Lambda \Wick{|u(x)|^{2r}}\,\ud x .
\end{align*}
Here Wick ordering is necessary because $\mu_0$ is supported on distributions of negative regularity. The basic case is the monomial $\fP(z)=p^{-1}|z|^{2p}$, whose measure and partition function will be denoted by $\mu_p$ and $\cZ_p$, respectively. We concentrate on this case and the result for general radial polynomials follows similarly.

These measures originate in constructive quantum field theory; their two-dimensional construction goes
back to Nelson \cite{Nel66, Nel73}, with further developments summarized in \cite{Sim74, GJ87}. They also arise as invariant Gibbs measures for the defocusing nonlinear Schr\"odinger equation (NLS). In the cubic case $p=2$, Bourgain \cite{Bou96} used it to construct global solutions of the NLS for rough initial data. For higher powers, global solutions of the NLS preserving the Gibbs measure in law were constructed in \cite{OT18}, and almost-sure global well-posedness was subsequently established in \cite{DNY24}.

Our goal is to derive $\mu_\fP$ from grand-canonical many-body quantum Gibbs states. Fix $p\geq2$. We consider a system of interacting bosons on the torus $\Lambda$. 
Since we do not fix the particle number, the appropriate many-body Hilbert space is the bosonic Fock space
\begin{equation*}
	\mathfrak{F} = \mathbb{C} \oplus \bigoplus_{n=1}^\infty \cH^n, \quad \cH^n = L^2_{\text{sym}}(\Lambda^n),
\end{equation*}
where $L^2_{\text{sym}}(\Lambda^n)$ is the symmetric subspace of $L^2(\Lambda^n)$.
For a symmetric translation-invariant $l$-body kernel $K$, write $\mathbb V_l[K]$ for the operator on $\cF$ whose $n$-particle restriction is the ordered sum displayed below. We consider the quantum Hamiltonian
\begin{equation}
	\begin{aligned}
		\HH_\lambda
		={}&0\oplus\bigoplus_{n=1}^\infty\left(\lambda\sum_{i=1}^n(-\Delta_{x_i}-\vartheta)\right)
		+\frac{\lambda^p}{p}\,\mathbb V_p[v^\varepsilon]
		+\sum_{l=2}^{p-1}\lambda^l\,\mathbb V_l[R_l(v^\varepsilon)]
		+E_0\,.
	\end{aligned}\label{Hamiltonian1}
\end{equation}
Its restriction to the $n$-particle sector $\cH^n$, $n\geq1$, is
\begin{align*}
	H_{\lambda,n}
	={}&\lambda\sum_{i=1}^n(-\Delta_{x_i}-\vartheta)
	+\frac{\lambda^p}{p}\sum_{i_1,\ldots,i_p=1}^nv^\varepsilon(x_{i_2}-x_{i_1},\ldots,x_{i_p}-x_{i_1})\\
	&+\sum_{l=2}^{p-1}\lambda^l
	\sum_{i_1,\ldots,i_l=1}^n
	R_l(v^\varepsilon)
	(x_{i_2}-x_{i_1},\ldots,x_{i_l}-x_{i_1})
	+E_0\,.\nonumber
\end{align*}
while the vacuum sector equals $E_0$. The inverse temperature $\lambda \to 0$ plays the role of the semiclassical parameter, and $\varepsilon\to0$ is the interaction range. The kernel $v^\varepsilon$ is the periodic rescaling of a nonnegative compactly supported positive-type $p$-body profile $v$, symmetric under particle permutations, and $v^\varepsilon\rightharpoonup\delta_0^{\otimes(p-1)}$. 
No factorization or other prescribed form of $v$ is assumed; within this class, the limiting field theory is independent of its detailed shape. The kernels $R_l(v^\varepsilon)$, together with the chemical-potential counterterm $\vartheta$ and the scalar shift $E_0$, are determined by Wick renormalization and are specified in Section \ref{sec2}.

We consider the Gibbs state
$$\Gamma_\lambda=Z_\lambda^{-1}e^{-\HH_\lambda},\quad Z_\lambda=\Tr(e^{-\HH_\lambda}),$$
and compare it with the free state associated with 
\begin{align}
	\HH_0=0\oplus\bigoplus_{n=1}^\infty\left(\lambda\sum_{i=1}^n(-\Delta_{x_i}+1)\right) \label{def:H0}
\end{align}
and partition function $Z_0$.
Our main result states that, for $\lambda,\varepsilon\to0$ with $\varepsilon\geq\lambda^\eta$ and $\eta=\eta(p)>0$ sufficiently small,
\[
\log\frac{Z_\lambda}{Z_0}\longrightarrow\log \cZ_p,
\qquad
k!\lambda^k\Gamma_\lambda^{(k)}
\longrightarrow\int |u^{\otimes k}\rangle\langle u^{\otimes k}|\,\ud\mu_p(u)
\]
in Hilbert--Schmidt norm for every $k\geq1$, where $\Gamma_\lambda^{(k)}$ is the reduced $k$-body density matrix of $\Gamma_\lambda$. We also prove trace-class convergence of the relative one-body density matrix. Thus both the thermodynamics and the correlation functions of the quantum system converge to those of the $\Phi_2^{2p}$ field, and the same conclusions hold for every radial polynomial $\fP$.

The derivation of nonlinear Gibbs measures from quantum Gibbs states was initiated in one dimension in \cite{LNR15}. In higher dimensions, the analysis is more complicated since Wick renormalization is necessary, and the corresponding problem for nonlocal pair interactions in two and three dimensions was resolved by two different approaches in \cite{LNR21, FKSS23} (see also \cite{FKSS17} for an earlier work concerning a modified, non-equilibrium state in two and three dimensions). The local $\Phi^4_2$ measure was derived later in \cite{FKSS25} by refining the functional integral method in \cite{FKSS23} (see also \cite{NZZ25, JR25} for alternative solutions based on the variational method from \cite{LNR15, LNR21}, allowing a polynomial relation between the interaction range and the temperature).  The local $\Phi^4_3$ measure, which requires  renormalization beyond Wick ordering, was derived in \cite{NZZ25} by combining the variational method and paracontrolled calculus \cite{GIP15}. 

Further developments concerning the emergence of nonlinear Gibbs measures from quantum models include the derivation of the inhomogeneous $\Phi^4_2$ measure with an external trapping potential in \cite{CKRG26}, and the derivation of the homogeneous $\Phi^4_2$ measure with a singular Bessel interaction in \cite{NZZ26}. While all the previously mentioned works focus on the defocusing case, in one dimension it is possible to consider attractive interactions; see \cite{RS25} and \cite{LNZ26} for the derivations of the mass-critical focusing $\Phi^6_1$ measure from quantum systems with three-body interactions. 

In the present work, we focus on the derivation of local Gibbs measures from general $p$-body interactions for arbitrary $p\ge 2$ in two dimensions. While it is a natural extension of the pair interaction setting studied in \cite{FKSS25, NZZ25, JR25}, the analysis of the general defocusing $\fP(\Phi)_2$ theory requires a substantial development of existing techniques. In the study of the classical two-dimensional defocusing NLS with rough random data, the extension from the cubic equation considered in \cite{Bou96} to higher-order nonlinearities also required substantially new techniques; in particular a general almost-sure global well-posedness result for these higher powers was obtained only recently in \cite{DNY24} using random averaging operators. Although the quantum problem considered here is of a different nature, it is not surprising that the passage from pair interactions to general $p$-body interactions introduces new structural difficulties.

First, for $\Phi^{2p}_2$ theory with a general power $p\ge 3$, Wick renormalization produces a full hierarchy of effective interactions of all lower orders, with no common quadratic structure, making the treatment of the counterterms more challenging. More precisely, the standard Wick renormalization produces a hierarchy of kernels $f_l^\varepsilon$, $0\leq l\leq p$, associated with the different contraction patterns. After second quantization, each $m$-th order term contributes further to all interactions of orders $l\leq m$, yielding the counterterms $R_l(v^\varepsilon)$ in \eqref{Hamiltonian1}, as well as the corresponding chemical-potential counterterm and scalar shift. This combinatorial procedure cannot be reduced to a simple induction argument since contraction patterns of different orders are strongly coupled. We will use graphical notation for the contraction patterns, which allows us to identify the full renormalization hierarchy and its combinatorial coefficients in a systematic way.  

Second, while the nonlocal (Hartree) interaction corresponding to the $\Phi^4_2$ theory admits a lower bound  by completing a square, this technique is no longer available for our general interaction potentials, and we need new tools to give quantitative control of high-momentum correlations. In this respect,  one key ingredient of our analysis is a logarithmic stability estimate which holds for both the nonlocal Hartree functional and the many-body Hamiltonian. In this bound, we use the positivity of the leading $p$-body interaction to control all lower-order terms generated by the Wick counterterms, leading to a uniform control independent of the interaction range. This estimate provides the exponential integrability needed for the classical zero-range limit, which ensures the convergence of the classical Hartree measure to the local one by using a Nelson-type argument. 
 
Third, while the stability estimate mentioned above is very helpful for the free energy lower bound, the corresponding upper bound requires a separate lower bound for the projected interaction as the spectral projection does not commute with multiplication by the interaction kernels. This part also requires a detailed analysis of all lower-order terms from Wick ordering. 

Resolving these issues allows us to take the semiclassical and zero-range limits simultaneously in the polynomial regime $\varepsilon\ge\lambda^\eta$, resulting in the desired local $\Phi^{2p}_2$ theory. The limiting field theory is universal in the sense that it is independent of the detailed shape of the microscopic interaction profile.

\section{Setting and main result}\label{sec2}

In this section, we introduce the notation, formulate the quantum problem, and state the main results together with an outline of the proof.
We begin by recalling the basic notations for quantum states on bosonic Fock space. A quantum state $\Gamma$ is a non-negative self-adjoint operator on $\cF$ with
$\Tr(\Gamma)=1$.
The reduced $k$-body density matrix $\Gamma^{(k)}$ is defined by duality as
\begin{equation}
	\textnormal{Tr}_{\cH^k} \left[A_k \Gamma^{(k)}\right] = \textnormal{Tr}_{\mathfrak{F}} \left[\mathbb{A}_k \Gamma\right]
\end{equation}
for every self-adjoint operator $A_k$ on $\cH^k$. Here $\mathbb{A}_k$ denotes the second quantization of $A_k$, defined by
\begin{align*}
	\mathbb{A}_k=0\oplus \cdot \cdot \cdot \oplus 0\bigoplus_{n=k}^\infty \left( \sum_{1 \leq i_1 < ...< i_k \leq n} (A_k)_{i_1,...,i_k} \right) .
\end{align*}
When $k=1$, we denote the second quantization by $\ud\Gamma(A_1)$. 
If $\Gamma$ commutes with the number operator $\fN:=\ud\Gamma(\1_\cH)$, then it admits a diagonal decomposition of the form $\Gamma=\bigoplus_{n=0}^\infty \Gamma_n$, and reduced density matrices are equivalently given via partial traces as
\begin{align*}
	\Gamma^{(k)} = \sum_{n\geq k} { n \choose k} \Tr_{k+1 \to n} (\Gamma_n) .
\end{align*}
The Gaussian quantum state is defined by
\begin{align}\label{def:Gaussian state}
	\Gamma_0 = Z_0^{-1} e^{-\HH_0},\quad Z_0 = \Tr(e^{-\HH_0})\,,
\end{align}
where $\HH_0$ is given in \eqref{def:H0}.
We denote by $N_0=\Tr(\fN\Gamma_0)$ the expected particle number under $\Gamma_0$. 

We next specify the \(p\)-body interaction potential in the interacting Gibbs state. Throughout, we fix \(p\ge 2\), and assume that \(v\) satisfies the following condition:

\begin{assumption}\label{assum:w}
	Let $v\in C_c(\mathbb R^{2(p-1)})$ be an even probability density whose Fourier transform satisfies,
for some $\delta_0>0$,
\begin{align}\label{v-fourier}
		0\leq\widehat v(k)\lesssim\frac{1}{1+|k|^{2p-2+\delta_0}}\,.
	\end{align}
	Moreover, assume that for every permutation $\sigma$ of $\{1,\dots,p\}$ and $x_1,\cdots,x_p\in\RR^2$,
	\begin{align}\label{con:sym}
		v(x_2-x_1,\cdots,x_p-x_1) &= v(x_{\sigma(2)}-x_{\sigma(1)},\cdots,x_{\sigma(p)}-x_{\sigma(1)})\,.
	\end{align}
\end{assumption}

\begin{remark}
	
	Let us quickly comment on the role of these assumptions. Here the positivity and decay of its Fourier transform provide useful control of the effective kernels generated by Wick contractions. In particular, the fact that $v$ is a probability density and \eqref{v-fourier} imply that $v(0)>0$, and hence by continuity $v$ is strictly positive near the origin, which plays a crucial role in the interpolation estimates used to absorb the lower-order interactions into the leading $p$-body term. Moreover, the permutation symmetry in \eqref{con:sym} guarantees that the interaction is compatible with bosonic symmetry. Note that the form $v(x_2-x_1,...,x_p-x_1)$ is simply a convenient representation of a translation-invariant $p$-body interaction, since any translation-invariant kernel $K$ can be written in this form by setting $v(y_1,...,y_{p-1})=K(0,y_1,...,y_{p-1})$. 
	
\end{remark}
\noindent
For $\varepsilon>0$, we define the rescaled periodic function $v^\varepsilon$, which approximates $\delta_0^{\otimes(p-1)}$ as $\varepsilon\to0$:
\begin{align}
	v^\varepsilon(y):=\sum_{k\in\mathbb Z^{2(p-1)}}
	\widehat v(\varepsilon k)e^{2\pi\iota k\cdot y}
	=\varepsilon^{-2(p-1)}
	\sum_{m\in\mathbb Z^{2(p-1)}}v\bigl(\varepsilon^{-1}(y+m)\bigr),
	\quad y\in\Lambda^{p-1},\label{poisson}
\end{align}
where the second identity follows from the Poisson summation formula.
In particular, $v^\varepsilon\ge0$ and $\int_{\Lambda^{p-1}}v^\varepsilon=1$.
For later use, we introduce marginal operators for general symmetric kernels. 
For \(F\in\cH^{r-1}\) and \(1\le l< r\), we define the \(l\)-particle marginal \(\Pi_lF\in\cH^{l-1}\) by
\begin{align}
	(\Pi_l F)(y_{1:l-1})
	:= \int_{\Lambda^{r-l}} F(y_{1:l-1},z_l,\cdots,z_{r-1})\,\ud z_l\cdots\ud z_{r-1}.\label{def:marginal}
\end{align}
Here and below, if \(l=r\), we use the convention \(\Pi_rF=F\), and if
\(l=1\), \(\Pi_1F=\int_{\Lambda^{r-1}}F\ud x\).

We are now ready to define the interacting Gibbs state on \(\cF\):
\begin{align}\label{def:Gibbs state}
	\Gamma_{\lambda} = Z_{\lambda}^{-1} e^{-\HH_{\lambda}},\quad Z_{\lambda} = \Tr(e^{-\HH_{\lambda}})\,,
\end{align}
with the Hamiltonian operator
\begin{align}
	\HH_{\lambda}=E_0\oplus\bigoplus_{n=1}^\infty&\left(\lambda\sum_{i=1}^n(-\Delta_{x_i}-\vartheta)+\frac1p\lambda^p\sum_{i_{1:p}=1}^n v^\varepsilon(x_{i_2}-x_{i_1},\cdots,x_{i_p}-x_{i_1})\right.\nonumber\\
	&\left.+\sum_{l=2}^{p-1}\lambda^l\sum_{i_{1:l}=1}^n R_l(v^\varepsilon)(x_{i_2}-x_{i_1},\cdots,x_{i_l}-x_{i_1})+E_0\right)\,,\label{Hamiltonian}
\end{align}
where 
\begin{align}
	R_l(v^\varepsilon) =& \sum_{m=l}^p(-\lambda N_0)^{m-l}\binom ml\Pi_lf_m^\varepsilon\,,\quad 2\leq l\leq p-1\,,\label{def:R_l}
\end{align}
and
\begin{align}
	\vartheta =& -\sum_{l=1}^pl\,(-\lambda N_0)^{l-1}\int_{\Lambda^{l-1}}f_l^\varepsilon(x)\ud x-1\,,\quad
	E_0= f_0^\varepsilon + \sum_{l=1}^p(-\lambda N_0)^l\int_{\Lambda^{l-1}}f_l^\varepsilon(x)\ud x\,.\label{def:nu}
\end{align}
The effective kernels $f_l^\varepsilon$, $0\le l\le p$, are generated by the Wick renormalization, with
$f_p^\varepsilon=p^{-1}v^\varepsilon$. For $0\le l<p$, they are given by
\begin{align}
	f_l^\varepsilon(x_2-x_1,\cdots,x_l-x_1) :=& \sum_{\mathbf m\in \fM_{p,l}}C_\fm\int_{\Lambda^{p-l}} v^\varepsilon(x_2-x_1,\cdots,x_p-x_1)\prod_{1\leq a<b\leq p}G(x_a-x_b)^{m_{ab}}\ud x_{l+1:p}\,,\label{def:F_l}
\end{align}
where \begin{align}\label{def:M}
	\fM_{p,l}=\{\fm=(m_{ab})_{1\leq a<b\leq p}: m_{ab}\in\NN, \sum_{1\leq a<b\leq p}m_{ab}=p-l\},
\end{align}
and $G(x-y)=(1-\Delta)^{-1}(x,y)$ is the periodic Green function.
For $l=0,1$, the right-hand side of \eqref{def:F_l} is interpreted as a scalar. The constants $C_\fm$ are the combinatorial coefficients defined in
\eqref{C_g2}. We refer the reader to Remark \ref{rmk 4.1} for an explanation on the form of $f_l^\varepsilon$.
Under Assumption \ref{assum:w}, $\HH_\lambda$ is bounded from below and can therefore be defined as a self-adjoint operator by the Friedrichs extension; see Subsection~\ref{sec4.1} for details.
Consequently, the partition function $Z_\lambda$ is finite and $\Gamma_\lambda$ is a well-defined quantum state.

\begin{remark}\label{rmk2.2}
	The lower-order kernels \(R_l(v^\varepsilon)\) in \eqref{def:R_l} arise as effective \(l\)-body counterterms after two renormalization steps.
	First, Wick expansion of the nonlocal interaction produces kernels \(f_m^\varepsilon\), \(0\le m\le p-1\), such that the corresponding Hartree functional can be written as
	\begin{align}
		W_p^\varepsilon(u)
		=\sum_{m=0}^p\int f_m^\varepsilon(x_2-x_1,\dots,x_m-x_1)
		\prod_{j=1}^m \Wick{|u(x_j)|^2}\,\ud x_{1:m},\label{def:W^epsilon,p}
	\end{align}
	which is an approximation of $\frac1p\int_{\Lambda}\Wick{|u(x)|^{2p}}\ud x$, see Section \ref{sec3} for details.
	Second, upon passing to Fock space, one expands
	\[\Wick{\ud\Gamma(e_k)}=\ud\Gamma(e_k)-N_0\,\1_{k=0}.\]
	Hence an $l$-body term receives contributions from every $m$-th order kernel	with $m\ge l$: one chooses $l$ factors that remain as $\ud\Gamma(e_k)$, while the remaining $m-l$ factors are replaced by $N_0$. This gives the form of \(R_l(v^\varepsilon)\) in \eqref{def:R_l}.
	In particular, the cases $l=1$ and $l=0$ produce the formulas for the chemical potential $\vartheta$ and the energy shift $E_0$, respectively.
	For illustration, when $p=3$, the decomposition \eqref{def:W^epsilon,p} takes the form  
	\begin{align*}
		W_3^\varepsilon(u)
		={}& \frac13 \int_{\Lambda^3}
		v^\varepsilon(x_2-x_1,x_3-x_1)
		\prod_{j=1}^3 \Wick{|u(x_j)|^2}\, \ud x_{1:3} \\
		&+ \int_{\Lambda^2} f_2^\varepsilon(x_2-x_1)
		\Wick{|u(x_1)|^2}\Wick{|u(x_2)|^2}\, \ud x_{1:2} + f_1^\varepsilon \int_\Lambda \Wick{|u(x)|^2}\, \ud x
		+ f_0^\varepsilon.
	\end{align*}
	where the effective kernels are given explicitly by
	\begin{align*}
		f_2^\varepsilon(x_2-x_1)
		&= -\int_\Lambda
		v^\varepsilon(x_2-x_1,x_3-x_1)
		\bigl(G_{13}+G_{23}\bigr)\,\ud x_3, \\
		f_1^\varepsilon
		&= \int_{\Lambda^2}
		v^\varepsilon(x_2-x_1,x_3-x_1)
		\bigl(2G_{12}G_{13}-G_{23}^2\bigr)\,\ud x_{2:3}, \\
		f_0^\varepsilon
		&= \frac13\int_{\Lambda^3}
		v^\varepsilon(x_2-x_1,x_3-x_1)
		\Bigl[
		G_{12}^2(G_{13}+G_{23}) \\
		&\qquad +G_{13}^2(G_{12}+G_{23})
		+G_{23}^2(G_{12}+G_{13})
		-2G_{12}G_{13}G_{23}
		\Bigr] \,\ud x_{1:3}
	\end{align*}
	with $G_{ij}=G(x_i-x_j)$ for $1\leq i<j\leq 3$. These lower-order kernels are generated by Wick contractions of the original three-body interaction. The graphical derivation of their coefficients is explained in Remark \ref{rmk3.4}.
\end{remark}

Finally, let $\mu_p$ be the $\Phi_2^{2p}$ measure defined by
\begin{align}\label{def:Phi measure}
	\ud\mu_p(u) = \cZ_p^{-1}\exp\left(-\frac1p\int_{\Lambda}\Wick{|u(x)|^{2p}}\ud x\right)\ud\mu_0(u),\quad \cZ_p=\int \exp\left(-\frac1p\int_{\Lambda}\Wick{|u(x)|^{2p}}\ud x\right)\ud\mu_0(u)\,,
\end{align}
where $\Wick{\cdot}$ denotes Wick ordering with respect to the Gaussian measure $\mu_0$.
See, for example, \cite[Proposition 1.2]{OT18} for a rigorous construction.
With the above definitions in place, we are ready to state the main result.

\begin{theorem}\label{thm}
	For $p\geq2$, consider the Gibbs state $\Gamma_{\lambda}$ in \eqref{def:Gibbs state} with $v$ satisfying Assumption \ref{assum:w}. Then as $\lambda,\varepsilon\to0$ with $\varepsilon\geq\lambda^\eta$ for a sufficiently small constant $\eta>0$ depending on $p$, we have 
	
	(1) Convergence of the relative free energy:
	\begin{align}\label{thm-eq1}
		\left|\log\frac{Z_{\lambda}}{Z_0}-\log\cZ_p\right|\to0\,.
	\end{align}

	(2) Convergence of density matrices:
	for every $k\in\NN$, we have convergence in Hilbert-Schmidt norm,
	\begin{align}\label{thm-eq2}
		\Tr\left|k!\,\lambda^k\Gamma_{\lambda}^{(k)}-\int|u^{\otimes k}\rangle\langle u^{\otimes k}|\ud\mu_p(u)\right|^2\to0\,.
	\end{align}

	(3) Trace class convergence of the relative one-body density matrix:
	\begin{align}\label{thm-eq3}
		\Tr\left|\lambda(\Gamma_{\lambda}^{(1)}-\Gamma_0^{(1)})-\int|u\rangle\langle u|(\ud\mu_p(u)-\ud\mu_0(u))\right|\to 0 \,.
	\end{align}
\end{theorem}
\noindent
In two dimensions, the free covariance $h^{-1}=(1-\Delta)^{-1}$ is Hilbert--Schmidt but not trace class, and hence the distinction between the convergence topologies in \eqref{thm-eq2} and \eqref{thm-eq3} is natural. In particular, the convergence as in \eqref{thm-eq2} does not hold in the trace class norm, even for $k=1$, which is the reason we need to subtract the corresponding free density matrix in \eqref{thm-eq3}. See also  \cite{LNR21,FKSS23,FKSS25} for further explanations in this direction.

To the best of our knowledge, this is the first derivation of the defocusing $\Phi^{2p}_2$ measure of arbitrary polynomial degree, from grand-canonical many-body Gibbs states with genuinely higher-order interactions. In particular, this result covers a rather general setting where the interaction profile is not assumed to be  factorized, and the interaction range is allowed to tend to zero simultaneously with the inverse temperature. 
Similar to the program initiated by \cite{LNR15}, the quantum-to-classical convergence is established for not only the free energy but also for all correlation functions. See \cite{FKSS25} for earlier results in the case $p=2$, and \cite{OT18, DNY24} for related results concerning invariant Gibbs measures for the classical 2D NLS flow.  
The result in Theorem \ref{thm} can be extended to the full defocusing $\fP(\Phi)_2$ as follows. For a general polynomial \(\fP(t)=\sum_{r=1}^p a_r|t|^{2r}\), $t\in\CC$ with $a_p>0$ and $a_r\in\RR$, $1\leq r<p$, 
we define the renormalized polynomial functional by
\begin{align*}
	V_\fP(u):=\sum_{r=1}^p a_r\int_\Lambda \Wick{|u(x)|^{2r}}\,\ud x.
\end{align*}
The positivity of the leading coefficient and Nelson’s estimate imply that $e^{-V_\fP}\in L^1(\mu_0)$; for example, see the proof of \cite[Proposition 1.2]{OT18} for details. Hence the \(\fP(\Phi)_2\) measure
\begin{align*}
	\ud\mu_\fP(u) = \cZ_\fP^{-1}e^{-V_\fP(u)}\ud\mu_0(u),\quad \cZ_\fP =\int e^{-V_\fP(u)}\ud\mu_0(u),
\end{align*}
is well-defined.
For $2\le r\le p$, let $f_l^{\varepsilon,(r)}$,
$0\le l\le r$, be the kernels obtained from \eqref{def:F_l} with $p$ replaced by $r$ and with $r$-body profile $\Pi_rv^\varepsilon$. In particular,
$f_r^{\varepsilon,(r)}=r^{-1}\Pi_rv^\varepsilon$. 
For \(l>r\), set \(f_l^{\varepsilon,(r)}:=0\). Let $f_1^{\varepsilon,(1)}:=1$ and $f_0^{\varepsilon,(1)}:=0$. 
We then define
\begin{align*}
	f_{l,\fP}^{\varepsilon}:=\sum_{r=1}^p ra_r f_l^{\varepsilon,(r)}, \quad 0\le l\le p.
\end{align*}
In particular, $f_{p,\fP}^\varepsilon=a_pv^\varepsilon$.
Let $\vartheta_\fP$, $E_{0,\fP}$, and
$R_{l,\fP}(v^\varepsilon)$ be defined by \eqref{def:R_l} and \eqref{def:nu} with $f_m^\varepsilon$ replaced by $f_{m,\fP}^\varepsilon$. We set
\begin{align}
	\HH_{\lambda,\fP}
	:=E_{0,\fP}\oplus\bigoplus_{n=1}^\infty
	&\left(\lambda\sum_{i=1}^n(-\Delta_{x_i}-\vartheta_\fP)
	+a_p\lambda^p\sum_{i_{1:p}=1}^n
	v^\varepsilon(x_{i_2}-x_{i_1},\ldots,x_{i_p}-x_{i_1})\right.\nonumber\\
	&\left.+\sum_{l=2}^{p-1}\lambda^l\sum_{i_{1:l}=1}^n
	R_{l,\fP}(v^\varepsilon)(x_{i_2}-x_{i_1},\ldots,x_{i_l}-x_{i_1})+E_{0,\fP}\right).\label{def:polynomialH}
\end{align}
Since the polynomial interaction is a finite linear combination of the monomial ones and \(a_pv^\varepsilon\geq0\), we can obtain similarly as Theorem~\ref{thm} that the
\(\fP(\Phi)_2\) measure \(\mu_{\fP}\) from the corresponding quantum Gibbs state.

\begin{theorem}\label{cor:poly}
	Let $\fP$ be as above. Consider the Gibbs state $$\Gamma_{\lambda,\fP}=(Z_{\lambda,\fP})^{-1}e^{-\HH_{\lambda,\fP}},\quad Z_{\lambda,\fP} = \Tr(e^{-\HH_{\lambda,\fP}}),$$ 
	where $\HH_{\lambda,\fP}$ is defined in \eqref{def:polynomialH}.
	Under the assumptions of Theorem~\ref{thm}, conclusions
	\eqref{thm-eq1}--\eqref{thm-eq3} remain valid with
	$\Gamma_\lambda$, $Z_\lambda$, $\mu_p$, and $\cZ_p$ replaced by
	$\Gamma_{\lambda,\fP}$, $Z_{\lambda,\fP}$, $\mu_\fP$, and
	$\cZ_\fP$, respectively.
\end{theorem}

\subsubsection*{Strategy of the proof of Theorem \ref{thm}.}
As in the overall approach in \cite{FKSS25} for $p=2$, the proof involves two limiting procedures: 
the quantum-to-classical limit $\lambda\to 0$ and the zero-range limit $\varepsilon \to 0$, and these limits are connected through an $\varepsilon$-dependent nonlocal Hartree measure. However, for $p\ge 3$, Wick renormalization generates a coupled hierarchy of effective interactions of all lower orders, and hence the interaction entering this intermediate Hartree measure is not simply the original $p$-body interaction. 
The main difficulties are therefore to identify this hierarchy, to establish stability estimates that are uniform in the interaction range, and to implement the variational quantum-to-classical method introduced by \cite{LNR15} and further developed in \cite{LNR18, LNR21} simultaneously for all terms in the hierarchy. 
In particular, for $p\ge 3$, the renormalized interaction contains terms of orders $(0,1,\ldots,p)$, with no common quadratic representation and no common sign, and hence neither the classical stability estimate nor its quantum analogue follows formally from the existing argument for $p=2$.  

\noindent\textbf{Hartree measure and classical stability.}
To be precise, for fixed $p\geq2$, we connect the $\Phi_2^{2p}$ measure and the many-body quantum problem through an $\varepsilon$-dependent Hartree measure \begin{align}\label{def:Hartree measure}
	\ud\mu^\varepsilon_p = \frac{1}{\cZ_p^\varepsilon}e^{-W_p^\varepsilon}\ud\mu_0\,,\quad  \cZ_p^\varepsilon=\int e^{-W_p^\varepsilon}\ud\mu_0\,,
\end{align} 
where \(W_p^\varepsilon\) is given in \eqref{def:W^epsilon,p}. Its rigorous construction is given in Section \ref{sec3}. 
In the first step, we need to study the Hartree measure and its limit as $\varepsilon\to 0$. A direct Wick expansion of the nonlocal monomial $\Wick{\prod_{j=1}^p|u(x_j)|^2}$ produces many contraction patterns and nonlocal remainder terms. In order to keep track of and control all these terms, we encode the contractions by graphs and use a recursive reduction to approximate the nonlocal terms as products of $\Wick{|u(x_j)|^2}$. This determines the coefficients $C_\fm$ and the effective kernels $f_l^\varepsilon$ in \eqref{def:F_l}. We then prove that the probability measure \(\mu_p^{\varepsilon}\) converges to \(\mu_p\) in total-variation, as $\varepsilon\to0$.
The key uniform estimate is a Nelson-type logarithmic lower bound
\begin{align}
	W_p^{\varepsilon,N}\ge-C(\log N)^p,\label{lwb:W}
\end{align}
where $C$ is independent of $\varepsilon$ and $W_p^{\varepsilon,N}$ is the Fourier truncation of $W_p^\varepsilon$.
Its proof uses the positivity of $v^\varepsilon$ to control local averages by the $p$-order term, followed by interpolation estimates and a finite-covering argument for all lower-order marginals. 

\noindent\textbf{Fock-space stability and quantum a priori estimates.}
The same coercive argument is lifted to Fock space in Section \ref{sec4}. Namely, once the classical renormalization step is done, we also need to extend the analysis upon quantization. To be precise, expanding each Wick density as $\Wick{\ud\Gamma(e_k)}=\ud\Gamma(e_k)-N_0\1_{\{k=0\}}$ shows that an $m$-th order classical term contributes to every quantum interaction order $l\leq m$. This gives the counterterms $R_l(v^\varepsilon)$, the chemical-potential  $\vartheta$, and the scalar shift $E_0$ in \eqref{Hamiltonian}--\eqref{def:nu}.
With these choices, the quantum interaction $\WW$ is the second-quantized counterpart of $W_p^\varepsilon$ appearing in the Hartree measure above, and we establish the stability estimate
$$\WW\geq-C|\log\lambda|^p\,,$$
which is the many-body quantum analogue of \eqref{lwb:W}.
Note that this stability estimate will be used not only for the original interacting Gibbs state, but also for the  Gibbs state with perturbed Hamiltonians, providing the input needed to apply the correlation inequalities of \cite{DNN25}. In this way, we obtain uniform control in $\varepsilon$ of every moment of $\lambda\mathcal{N}$ and quantitative estimates on high-momentum correlations, in the regime $\varepsilon \ge \lambda^\eta$.

\noindent\textbf{Comparison of the variational problems.}
We connect the quantum Gibbs state $\Gamma_\lambda$ to the Hartree measure by a variational argument following \cite{LNR15,LNR18,LNR21,NZZ25}. On the quantum side, \(\Gamma_\lambda\) is the
unique minimizer of the Gibbs variational principle
\begin{equation}\label{variational problem}
	-\log \frac{Z_\lambda}{Z_0}
	=\inf_{\Gamma\ge 0,\Tr\Gamma=1}
	\Big\{
	\mathcal H(\Gamma,\Gamma_0)+\Tr(\WW\Gamma)
	\Big\},
\end{equation}
where
\begin{align*}
	\mathcal H(\Gamma,\Gamma_0)
	=
	\Tr\big(\Gamma(\log\Gamma-\log\Gamma_0)\big)
\end{align*}
is the von Neumann relative entropy with respect to the Gaussian state \(\Gamma_0\) defined in \eqref{def:Gaussian state}, and $\WW$ is the interaction operator in $\HH_\lambda$. We have chosen $\vartheta$, $E_0$ and $R_l(v^\varepsilon)$ such that $\WW$ is exactly the second quantized version of the classical term $W_p^\varepsilon(u)$ in \eqref{def:W^epsilon,p}. 
On the classical side, the Hartree measure \(\mu_p^{\varepsilon}\) is the unique minimizer of
\begin{equation}\label{eq:var-classical}
	-\log \cZ_p^\varepsilon
	=
	\inf_{\nu\ll\mu_0}
	\left\{
	\mathcal H_{\mathrm{cl}}(\nu,\mu_0)+\int W_p^\varepsilon(u)\,\ud\nu(u)
	\right\},
\end{equation}
where
\begin{align*}
	\mathcal H_{\mathrm{cl}}(\nu,\mu_0)
	=
	\int \frac{\ud\nu}{\ud\mu_0}\log\!\left(\frac{\ud\nu}{\ud\mu_0}\right)\ud\mu_0
\end{align*}
is the classical relative entropy between \(\nu\) and \(\mu_0\).

The two variational problems are compared after a finite-dimensional spectral localization, similar to \cite{LNR21}.
For the lower bound of the relative free energy $-\log\frac{Z_\lambda}{Z_0}$, we localize the interaction $\WW$ to low-momentum modes through a frequency cutoff, and apply a quantitative quantum de Finetti estimate to identify the limiting classical interaction. 
The errors from the high-momentum part are controlled by the particle-number moments and the quantum variance of the
high-momentum modes from Section~\ref{sec4}.
A key point is that this argument must be
performed simultaneously for all lower-order interactions
\(R_l(v^\varepsilon)\), \(2\le l<p\). Their Fourier norms are estimated uniformly in the regime $\varepsilon\geq\lambda^\eta$, and the resulting errors are
controlled by the positive \(p\)-body contribution.

For the upper bound, we use a trial state 
consisting of an interacting Gibbs state on the low-energy Fock space \(\cF(P\cH)\) and the free Gibbs state on \(\cF(Q\cH)\), which is the same as the one used in \cite[Section 8]{NZZ25} for $p=2$. However, for $p\geq3$, the definition of the low-energy interacting Gibbs state requires a lower bound for the projected interaction that is uniform in the cutoff dimension. 
This estimate is not an immediate consequence of the lower bound for $\WW$, since the projection operator does not commute with multiplication by the interaction kernels, and H\"older's inequality fails. 
We focus on the contributions with pairwise distinct labels and use a partition of unity to obtain the lower bound (see Lemma \ref{lem 5.5}). Then, we perform a semiclassical approximation in finite dimensions as in \cite[Section 8]{NZZ25} to connect with \eqref{eq:var-classical}.

\noindent\textbf{Reduced density matrices.} 
Finally, following the strategy of \cite{LNR21}, we combine the lower- and upper-bound arguments to control the relative entropy between the quantum Gibbs state and the decoupled state formed from its low- and high-energy marginals.
Then using Pinsker's inequality, the quantitative de Finetti estimate, and the particle-number moment bounds in the spirit of \cite{NZZ25}, we obtain the Hilbert–Schmidt convergence of all fixed-order reduced density matrices.
\vspace{0.8em}

A natural further question is to study the analogous higher-order local measures in three dimensions. The corresponding local problem is substantially more singular and is not addressed in the present work. 
For a related result in the nonlocal setting, concerning Hartree measure arising from bosonic Gibbs states with three-body interactions in dimensions two and three, we refer to the work of Liang \cite{Lia26}.


\subsubsection*{Organization of the paper}
The paper is organized as follows. 
In Section \ref{sec3}, we construct the Hartree measure and prove its convergence to the $\Phi_2^{2p}$ measure.
In Section~\ref{sec4}, we establish the lower bound for the interaction operator and derive a priori estimates on the particle-number moments, and high-momentum correlation inequalities.
In Section~\ref{sec5}, we derive lower and upper bounds on the relative free energy and conclude the proof of \eqref{thm-eq1}.
In Section \ref{sec6}, we establish the convergence of the reduced density matrices. 
Appendix \ref{Appendix} is devoted to the expansion of Wick products. The basic estimates for the Green function and the
properties of functions $f_l^\varepsilon$ in \eqref{def:F_l} are given in Appendix \ref{appB}.
We recall two elementary estimates for the correlation functions in Appendix \ref{appC}.

\subsubsection*{Notations}
We write $a\lesssim b$ if there exists a constant $c>0$ such that $a\leq cb$, and $\lesssim_p$ means the constant $c=c(p)$ depending on $p$.
We will denote by $C$ a general positive constant independent of relevant variables, whose value may change from line to line. 
We use the convention that the inner product $\langle \cdot, \cdot\rangle$ of a (complex) Hilbert space is linear in the second argument and antilinear in the first. In particular, $$\langle f,g\rangle_{L^2(\Lambda)}=\int_{\Lambda} \overline f g\,\ud x,\quad f,g\in L^2(\Lambda).$$
For $n\geq1$, we write \([n]:=\{1,\dots,n\}\), and $S_n$ denotes the set of all permutations of $[n]$.
Let $e_k(x):=e^{2\pi\iota k\cdot x}$, $k\in\ZZ^2$. The Fourier transform on $\Lambda$ is defined by
\[\widehat f(k)=\fF(f)(k):=\int_\Lambda f(x)e_{-k}(x)\,\ud x.\]
For functions on $\RR^d$, the same convention is used with the Fourier variable $\xi\in\RR^d$. We write
\[\langle k\rangle^2:=1+4\pi^2|k|^2.\]
For $F\in L^2(\Lambda^n)$, define
\[
\operatorname{sym}F(x_{1:n})
:=\frac1{n!}\sum_{\sigma\in S_n}
F(x_{\sigma(1)},\ldots,x_{\sigma(n)}),
\]
and set
$f_1\otimes_{\mathrm s}\cdots\otimes_{\mathrm s}f_n
:=\operatorname{sym}(f_1\otimes\cdots\otimes f_n)$.
For $m\ge1$ and $y\in\RR^m$, let
\[d_m(y):=\min_{z\in\ZZ^m}|y+z|.\]
When the dimension is clear, we simply write $d(y)$. For $y\in\Lambda$ and $r>0$,
\[B(y,r):=\{x\in\Lambda:\ d(x-y)\le r\}.\]

\section{Convergence of the Hartree measure to the $\Phi_2^{2p}$ measure}\label{sec3}

In this section we construct the Hartree measure
\(\mu_p^\varepsilon\) in \eqref{def:Hartree measure} and prove that its finite-dimensional approximations converge to the \(\Phi_2^{2p}\) measure. More precisely, we prove the convergence of the truncated partition functions and of the associated $k$-point correlation functions for $k\geq1$.
The argument has two steps. First, we rewrite the $p$-point Wick monomial as the product of quadratic Wick products, plus lower-order terms up to small errors. We introduce graph notations for this procedure. This yields the effective kernels $f_l^{\varepsilon,N}$ and the truncated Hartree functional $W_p^{\varepsilon,N}$.
Second, we prove the $L^m$ convergence of $\exp(-W_p^{\varepsilon,N})$ to $\exp(-V_p)$ as $N\to\infty$ and $\varepsilon\to0$ for every $m\geq1$, where the local interaction $V_p$ is defined in \eqref{def:V^p}.
The proof follows from the Nelson argument. A key estimate is a uniform logarithmic lower bound for \(W_p^{\varepsilon,N}\). 
For the quartic case, the lower order term can be absorbed by completing the square, which directly gives the desired lower bound as in \cite[Proposition 4.3]{FKSS25}. 
For $p\geq3$, however, the graphical expansion produces effective kernels of different orders, and the full interaction cannot be reduced to a positive expression by a completion-of-squares argument. 
The main new ingredient of this section is the use of the positivity and finite range of potential $v$ to control the lower order terms by the leading \(p\)-body term; see Lemma~\ref{lem:lower bound of WN}.

\subsection{Wick product reduction and graph notations}

Let $u$ be a complex-valued Gaussian field on $\Lambda$, defined on a probability space $(\Omega,\mathcal{F}, \mathbb{P})$, with mean zero and covariance $(1-\Delta)^{-1}$.
Since $u\in H^s(\Lambda)\backslash L^2(\Lambda)$ for any $s<0$, almost surely, we introduce a smooth Fourier cutoff to regularize the field $u$.
Let $\rho\in L^1(\RR^2)$ be a radial function such that its Fourier transform satisfies $\widehat\rho\in C_c^\infty(\RR^2,[0,1])$ with $\widehat\rho(k)=1$ for $|k|\leq\frac12$ and $\widehat\rho(k)=0$ for $|k|>1$.
For $N\in\NN$, let $u_N:=P_Nu$ with  $P_N=\fF^{-1}(\widehat{\rho}(N^{-1}\cdot)\fF)$.
Then we have
\begin{align}
	\EE(u_N(x)\overline{u_N(y)})=(\rho_N*\rho_N*G)(x-y)=: G_N(x-y)\,,\label{Green}
\end{align}
with $\rho_N=\sum_{k\in\ZZ^2}\widehat\rho(N^{-1}k)e_k$.
Define
\[V^N_p(u):=\frac1p\int_\Lambda\Wick{|u_N(x)|^{2p}}\,\ud x.
\]
Here and below, Wick ordering is taken with respect to the Gaussian field $u$; see, for instance, \cite[(A.1)]{FKSS25}.
By \cite[Proposition~1.1]{OT18}, the sequence $(V^N_p)_{N\ge1}$ is Cauchy in
$L^2(\mathbb P)$. We denote its limit by
\begin{align}\label{def:V^p}
	V_p=\frac1p\int_{\Lambda}\Wick{|u(x)|^{2p}}\ud x.
\end{align}
To connect with the many body Gibbs state, we consider the following nonlocal approximation
\begin{align}\label{def:V_N}
	V_p^{\varepsilon,N}(u)=\frac1p \int_{\Lambda^p} v^\varepsilon(x_2-x_1,\cdots,x_p-x_1)\Wick{\prod_{k=1}^p|u_N(x_k)|^2}\ud x_{1:p}\,.
\end{align} 
To match the structure of the Hamiltonian in \eqref{Hamiltonian}, $\Wick{\prod_{k=1}^p|u_N(x_k)|^2}$ must instead be expressed through the product of $\Wick{|u_N(x_i)|^2}$.
We illustrate the reduction in the quartic case $p=2$. By the definition of Wick ordering \cite[(A.2)]{FKSS25},
\begin{align}
	\Wick{|u_N(x_1)|^2|u_N(x_2)|^2}=\prod_{k=1}^2\Wick{|u_N(x_k)|^2}-2\text{Re}\Wick{u_N(x_1)\overline{u_N(x_2)}}G_N(x_1-x_2)-G_N(x_1-x_2)^2\,.\label{3.3}
\end{align}
The mixed Wick product $\Wick{u_N(x_1)\overline{u_N(x_2)}}$ can be shifted to $\Wick{|u_N(x_1)|^2}$ or $\Wick{|u_N(x_2)|^2}$, with an error that vanishes after integration against $v^\varepsilon$; see \cite[Proposition~3.2]{FKSS25}.

In the following, we first write the $p$-body Wick product $\Wick{\prod_{i=1}^p|u_N(x_i)|^2}$ in terms of product of $\Wick{|u_N(x_i)|^2}$ plus some lower-order terms; these lower-order terms are then reduced up to vanishing remainders, to products involving fewer $\Wick{|u_N(x_i)|^2}$.
Since the procedure involves many possible contraction terms, it is convenient to use the graph notations.
Fix $N\in\NN$, $\varepsilon>0$, and ordered vertices $x_1,\dots,x_p\in \Lambda$.
A solid edge joining $x_1$ to $x_2,\dots,x_p$ represents the factor $v^{\varepsilon}(x_2-x_1,\dots,x_p-x_1)$.
A dashed edge between two vertices $x_i$ and $x_j$ represents a factor $G_N(x_i-x_j)$. 
An oriented red wavy half-edge represents $u_N(x_i)$ or $\overline{u_N(x_i)}$, according to its orientation, whereas a blue pair represents $\Wick{|u_N(x_i)|^2}$.
Two vertices are said to be \emph{connected} if they lie in the same dashed-connected component. We use the following conventions.

\begin{itemize}[leftmargin=1.6em,itemsep=0.25em]
	\item A graph containing the solid edges represents the integral over all of its vertex variables.
	\item All wavy half-edges in a given graph have the same color. We write $\fH$ for red graphs and $\fG$ for blue graphs. The graph without wavy edges belongs to both sets. The term represented by a red graph is Wick ordered as a whole.
	\item Graphs contain no dashed self-loops.
	\item Every vertex without wavy lines is incident to at least one dashed edge.
	\item Two vertices carrying single red wavy lines with the same orientation are not allowed to be connected.
\end{itemize}
\noindent
These conventions provide a graphical representation of the Wick expansion. For example, \eqref{3.3} becomes
\begin{align}
	\begin{tikzpicture}[baseline=-5]
		\node[dot] (x1) at (0,-0.1) {};
		\draw[Phi1] (x1) -- ++(-0.1,0.25);
		\draw[Phi1] (x1) -- ++(0.1,0.25);
		\node[dot] (x2) at (0.4,-0.1) {};
		\draw[Phi1] (x2) -- ++(-0.1,0.25);
		\draw[Phi1] (x2) -- ++(0.1,0.25);
	\end{tikzpicture}
	=\begin{tikzpicture}[baseline=-5]
		\node[dot] (x1) at (0,-0.1) {};
		\draw[Phi2] (x1) -- ++(-0.1,0.25);
		\draw[Phi2] (x1) -- ++(0.1,0.25);
		\node[dot] (x2) at (0.4,-0.1) {};
		\draw[Phi2] (x2) -- ++(-0.1,0.25);
		\draw[Phi2] (x2) -- ++(0.1,0.25);
	\end{tikzpicture}
	-\;\begin{tikzpicture}[baseline=-5]
		\node[dot] (x1) at (0,-0.1) {};
		\draw[Phi1] (x1) -- ++(-0.1,0.25);
		\node[dot] (x2) at (0.4,-0.1) {};
		\draw[Phi1] (x2) -- ++(0.1,0.25);
		\draw[Cr,bend left=70] (x1) to (x2);
	\end{tikzpicture}
-\begin{tikzpicture}[baseline=-5]
	\node[dot] (x1) at (0,-0.1) {};
	\draw[Phi1] (x1) -- ++(0.1,0.25);
	\node[dot] (x2) at (0.4,-0.1) {};
	\draw[Phi1] (x2) -- ++(-0.1,0.25);
	\draw[Cr,bend right=70] (x1) to (x2);
\end{tikzpicture}
	-\begin{tikzpicture}[baseline=-5]
		\node[dot] (x1) at (0,-0.1) {};
		\node[dot] (x2) at (0.4,-0.1) {};
		\draw[Cr,bend right=50] (x1) to (x2);
		\draw[Cr,bend left=70] (x1) to (x2);
	\end{tikzpicture}\,.\label{gph1}
\end{align}
\noindent
We now introduce several classes of graphs that will be used later.

\begin{definition}\label{def:graph}
	Fix $N\in\mathbb N$ and $x_{1:p}\in \Lambda^p$. We write $\fT^{N}(x_{1:p})$ for the set of graphs on the ordered vertices $x_1,\dots,x_p$ with no solid edges, in which dashed edges represent $G_N$ and wavy half-edges represent $u_N$ or $\overline{u_N}$. 
	We require that every dashed-connected component $C$ of $\fg$ satisfies
	\begin{align}\label{wC}
		\#\{\text{dashed edges in }C\} + \frac12\#\{\text{wavy lines in }C\}	= \#\{\text{vertices in }C\}.
	\end{align}

For a graph $\fg$, let $U_\fg$ be the set of vertices carrying wavy edges, and let $v_\fg$ be their total number. We denote by $F(\fg)$ the product of all Green functions associated with the dashed edges, and by
$m_\fg$ the number of dashed loops containing at least three vertices. By a slight abuse of notation, we do not distinguish a graph from the random variable or integral associated with it.

	
Let $\fH^{N}(x_{1:p})\subset \fT^N(x_{1:p})\cap \fH$ be the class of red contraction graphs obtained from $\prod_{i=1}^p\Wick{|u_N(x_i)|^2}$ as follows.
	At each vertex $x_i$, we start with two legs, one $u_N(x_i)$-leg and one $\overline{u_N(x_i)}$-leg. Each leg is either left uncontracted, in which it is drawn as a red wavy half-edge, or contracted with an oppsitely oriented leg at a different vertex $x_j$, in which case the contraction is drawn as a dashed edge $G_N(x_i-x_j)$. The remaining red wavy edges are Wick ordered together. Equivalently, $\fH^{N}(x_{1:p})$ consists of red graphs such that, at each vertex, the total number of dashed edges plus the number of red wavy lines is exactly two.
	
Let $\fG^{N}(x_{1:p})\subset \fT^{N}(x_{1:p})\cap \fG$
	be the class of effective blue graphs. In such a graph, a vertex either carries one blue pair, representing $\Wick{|u_N(x_i)|^2}$, or carries no wavy half-edges. A graph belongs to $\fG^{N}(x_{1:p})$ if each dashed-connected component contains at most one vertex of $U_\fg$. Such a graph represents the expression $F(\fg)\prod_{i\in U_\fg}\Wick{|u_N(x_i)|^2}$.
	
	Finally, let $\fG^N_{p}$ be the class of graphs obtained from graphs in $\fG^{N}(x_{1:p})$ by adding the solid edges $(x_1,x_j)$ for $2\le j\le p$. Equivalently, a graph belongs to $\fG^{N}_{p}$ if and only if deleting all of its solid edges produces an element of $\fG^{N}(x_{1:p})$.
	Throughout, graphs in $\fG^N_p$ are treated as labelled graphs. Hence distinct labelled graphs in $\fG^{N}(x_{1:p})$ are kept as separate combinatorial objects, even when their integrals against the symmetric kernel $v^\varepsilon$ coincide by \eqref{con:sym}.
	
\end{definition}
\noindent
Using Lemma \ref{lem:complex-wick}, we can write the Wick product in terms of red graphs in $\fH^N(x_{1:p})$. Pointwise in $x_{1:p}$, we have
\begin{align}
	\prod_{k=1}^p\Wick{|u_N(x_k)|^2}
	=\sum_{\fg\in \fH^{N}(x_{1:p})} 2^{m_\fg}\fg .\label{3.0}
\end{align}
Now we rewrite the red graphs in $\fH^{N}(x_{1:p})$ as effective blue graphs in \(\fG^N(x_{1:p})\), up to small errors.
This step is guided by the structure of the Hartree functional that we seek to recover from the many-body problem. As mentioned before, in the many-body expansion, the surviving lower-order contributions should be expressed in terms of products of \(\Wick{|u_N(x)|^2}\). The role of the blue graphs is precisely to record these \(\Wick{|u_N(x)|^2}\).

Before stating the lemma, we introduce an auxiliary graph class and a relation between graphs, which will serve as an intermediate form in the proof. Let $\widetilde \fT^{N}(x_{1:p})\subset \fT^{N}(x_{1:p})$ be the subclass satisfying that
\begin{itemize}[leftmargin=1.5em]
	\item every vertex carrying wavy lines carries exactly one \emph{pair} and is incident to at most one dashed edge;
	
	\item no two distinct vertices carrying wavy lines belong to the same dashed-connected component;
	
	\item every vertex without wavy lines is incident to either one or two dashed edges.
\end{itemize}
For $h\in\fT^N(x_{1:p})$, $\fg'\in\fG^N(x_k,k\in V)$ and $V\subseteq[p]$ with $1\leq |V|<p$, we write \(\fg'\prec_V h\) if there exists a graph $\fg\in\widetilde\fT^N(x_{1:p})\cap\fH$, such that $U_\fg=V$ and $h=F(\fg)\fg'$. Here $U_\fg$ and $F(\fg)$ are defined in Definition \ref{def:graph}.
If such \(V\) and \(\fg'\) exist, then \(h\in\fG^N(x_{1:p})\).
Moreover, for fixed \(h\), \(V\), and \(\fg'\), the corresponding graph \(\fg\), if it exists, is unique.
Indeed, \(\fg\) is obtained from \(h\) by removing the dashed edges and wavy lines belonging to \(\fg'\) on the vertex set \(V\), and then placing red pairs
at the vertices indexed by \(V\).

We are now ready to state the graphical reduction lemma, which converts the red graph expansion into the blue one, up to negligible remainder terms.

\begin{lem}\label{Lem2.2}
	For every \(p\ge 2\), \(N\in\mathbb N\), and \(x_1,\cdots,x_p\in\Lambda\), we have
	\begin{align}\label{3.1'}
		\prod_{k=1}^p\Wick{|u_N(x_k)|^2}=\Wick{\prod_{k=1}^p|u_N(x_k)|^2}+\sum_{\substack{\fg\in\fG^{N}(x_{1:p})\\v_\fg<2p}} C_\fg \fg + \fR_N(u;x_{1:p})\,,
	\end{align}
	where $\fR_N(u;x_{1:p})$ is a finite linear combination of remainder terms of the form
	\begin{align}\label{gph2}
		\begin{tikzpicture}[baseline=-5]
			\node[dot] (x1) at (0,0) {};
			\draw[Phi1] (x1) -- ++(-0.1,0.25);
			\draw[Phi1] (x1) -- ++(0.1,0.25);
			\node[dot] (x2) at (0.3,0) {};
			\draw[Phi1] (x2) -- ++(-0.1,0.25);
			\node[dot] (x1') at (0.6,0) {};
			\node[dot] (x2') at (0.9,0) {};
			\draw[Phi1] (x2') -- ++(0.1,0.25);
			\node at (1.3,0) {$\cdots$};
			\node[dot] (x3) at (1.6,0) {};
			\draw[Phi1] (x3) -- ++(-0.1,0.25);
			\node[dot] (x5) at (1.9,0) {};
			\draw[Phi1] (x5) -- ++(0.1,0.25);
			\draw[Cr,bend left=50] (x2) to (x1');
			\draw[Cr,bend left=50] (x1') to (x2');
			\draw[Cr,bend left=70] (x3) to (x5);
		\end{tikzpicture}-\begin{tikzpicture}[baseline=-5]
		\node[dot] (x1) at (0,0) {};
		\draw[Phi1] (x1) -- ++(-0.1,0.25);
		\draw[Phi1] (x1) -- ++(0.1,0.25);
		\node[dot] (x2) at (0.3,0) {};
		\draw[Phi1] (x2) -- ++(-0.1,0.25);
		\draw[Phi1] (x2) -- ++(0.1,0.25);
		\node[dot] (x1') at (0.6,0) {};
		\node[dot] (x2') at (0.9,0) {};
		\node at (1.3,0) {$\cdots$};
		\node[dot] (x3) at (1.6,0) {};
		\draw[Phi1] (x3) -- ++(-0.1,0.25);
		\node[dot] (x5) at (1.9,0) {};
		\draw[Phi1] (x5) -- ++(0.1,0.25);
		\draw[Cr,bend left=50] (x2) to (x1');
		\draw[Cr,bend left=50] (x1') to (x2');
		\draw[Cr,bend left=70] (x3) to (x5);
	\end{tikzpicture}\,,
	\end{align}
namely differences of two graphs in
$\fT^N(x_{1:p})\cap\fH$ having the same dashed part and differing only by the position of one red wavy line.
	The constant $C_\fg$, independent of $x_{1:p}$ and $N$, is determined recursively by
	\begin{align}\label{def:C_g}
		C_\fg = 2^{m_\fg} \Bigl(\1_{\{\fg\in\widetilde\fT^{N}(x_{1:p}):v_{\fg}=0\}}-\sum_{\substack{V\subseteq[p]\\1\leq|V|<p}}\sum_{\fg'\prec_V\fg}2^{-m_{\fg'}}C_{\fg'}\Bigr).
	\end{align}
For every $V\subseteq[p]$, if $\fg$ is the completely uncontracted blue graph on \(V\), i.e. \(v_{\fg}=2|V|\), we set $C_{\fg}:=-1$.
\end{lem}

\begin{proof} 
	We argue by induction on $p$. For $p=2$, by \eqref{3.3}, we have
	\begin{align*}
		\prod_{k=1}^2\Wick{|u_N(x_k)|^2}=\Wick{\prod_{k=1}^2|u_N(x_k)|^2}+\sum_{1\leq i\neq j\leq 2}\Wick{u_N(x_i)\overline{u_N(x_j)}}G_N(x_i-x_j)+G_N(x_1-x_2)^2\,.
	\end{align*}
	We rewrite the two mixed red terms by moving one red half-edge so as to create a local pair at one vertex. Whenever such a local red pair is created, we identify it with the corresponding blue pair, representing $\Wick{|u_N(x_i)|^2}$.
	\begin{align*}
		\Wick{u_N(x_1)\overline{u_N(x_2)}}G_N(x_1-x_2) = \begin{tikzpicture}[baseline=-5]
			\node[dot] (x1) at (0,-0.1) {};
			\draw[Phi1] (x1) -- ++(-0.1,0.25);
			\node[dot] (x2) at (0.4,-0.1) {};
			\draw[Phi1] (x2) -- ++(0.1,0.25);
			\draw[Cr,bend left=70] (x1) to (x2);
		\end{tikzpicture}
	=&\begin{tikzpicture}[baseline=-5]
		\node[dot] (x1) at (0,-0.1) {};
		\draw[Phi2] (x1) -- ++(-0.1,0.25);
		\draw[Phi2] (x1) -- ++(0.1,0.25);
		\node[dot] (x2) at (0.4,-0.1) {};
		\draw[Cr,bend left=70] (x1) to (x2);
	\end{tikzpicture}
+\left(\begin{tikzpicture}[baseline=-5]
	\node[dot] (x1) at (0,-0.1) {};
	\draw[Phi1] (x1) -- ++(-0.1,0.25);
	\node[dot] (x2) at (0.4,-0.1) {};
	\draw[Phi1] (x2) -- ++(0.1,0.25);
	\draw[Cr,bend left=70] (x1) to (x2);
\end{tikzpicture}-\begin{tikzpicture}[baseline=-5]
\node[dot] (x1) at (0,-0.1) {};
\draw[Phi1] (x1) -- ++(-0.1,0.25);
\draw[Phi1] (x1) -- ++(0.1,0.25);
\node[dot] (x2) at (0.4,-0.1) {};
\draw[Cr,bend left=70] (x1) to (x2);
\end{tikzpicture}\right),\\
\Wick{u_N(x_2)\overline{u_N(x_1)}}G_N(x_1-x_2) = \begin{tikzpicture}[baseline=-5]
	\node[dot] (x1) at (0,-0.1) {};
	\draw[Phi1] (x1) -- ++(0.1,0.25);
	\node[dot] (x2) at (0.4,-0.1) {};
	\draw[Phi1] (x2) -- ++(-0.1,0.25);
	\draw[Cr,bend left=70] (x1) to (x2);
\end{tikzpicture}
=&\begin{tikzpicture}[baseline=-5]
	\node[dot] (x1) at (0,-0.1) {};
	\node[dot] (x2) at (0.4,-0.1) {};
	\draw[Phi2] (x2) -- ++(-0.1,0.25);
	\draw[Phi2] (x2) -- ++(0.1,0.25);
	\draw[Cr,bend left=70] (x1) to (x2);
\end{tikzpicture}
+\left(\begin{tikzpicture}[baseline=-5]
	\node[dot] (x1) at (0,-0.1) {};
	\draw[Phi1] (x1) -- ++(0.1,0.25);
	\node[dot] (x2) at (0.4,-0.1) {};
	\draw[Phi1] (x2) -- ++(-0.1,0.25);
	\draw[Cr,bend left=70] (x1) to (x2);
\end{tikzpicture}
-\begin{tikzpicture}[baseline=-5]
	\node[dot] (x1) at (0,-0.1) {};
	\node[dot] (x2) at (0.4,-0.1) {};
	\draw[Phi1] (x2) -- ++(-0.1,0.25);
	\draw[Phi1] (x2) -- ++(0.1,0.25);
	\draw[Cr,bend left=70] (x1) to (x2);
\end{tikzpicture}\right).
	\end{align*}
Each bracket is a remainder of the form \eqref{gph2}, and
the resulting local paired terms are exactly the two graphs in \(\fG^N(x_{1:2})\) with \(v_\fg=2\). For each of them, the only possible lower-order graph satisfying \(\fg'\prec_V\fg\) is the completely uncontracted
graph on the corresponding single vertex \(V=\{1\}\) or \(V=\{2\}\). By the initial convention, \(C_{\fg'}=-1\), and thus \eqref{def:C_g} gives \(C_\fg=1\). 
Moreover, $G_N(x_1-x_2)^2$ is the unique graph in \(\fG^N(x_{1:2})\) with $v_\fg=0$.
This proves \eqref{3.1'} for \(p=2\).
	
	Assume now that $p\geq 3$ and that the statement is already	known for all orders $2,\dots,p-1$. The unique graph in $\fH^N(x_{1:p})$ with $v_\fg=2p$ is the totally uncontracted graph $\Wick{\prod\limits_{k=1}^p|u_N(x_k)|^2}$, hence \eqref{3.0} gives
	\begin{align}\label{3.4}
		\prod_{k=1}^p\Wick{|u_N(x_k)|^2}-\Wick{\prod_{k=1}^p|u_N(x_k)|^2}=\sum_{\substack{\fg\in \fH^{N}(x_{1:p})\\v_\fg<2p}}2^{m_\fg}\fg \,.
	\end{align}
\medskip
\noindent\textbf{Step 1: Reduction from $\fH^N(x_{1:p})$ to $\widetilde\fT^N(x_{1:p})$.}
We first record the structure of dashed-connected components in $\fH^N(x_{1:p})$.
Let $C$ be a dashed-connected component of a graph
$\fg\in \fH^N(x_{1:p})$, and set $w_C:=\#\{\text{wavy lines in }C\}$. 
Since $C$ is connected, its number of dashed edges is at least the number of vertices minus one. Combining this with \eqref{wC} shows that $w_C$ is even and satisfies $w_C\le2$. Hence $w_C\in\{0,2\}$, and exactly one of the following
alternatives occurs.
\begin{enumerate}[leftmargin=2em]
	\item[(i)] $w_C=0$, and then every vertex of $C$ has dashed degree $2$; hence $C$ is a dashed cycle.
	\item[(ii)] $w_C=2$ and one vertex of $C$ carries the two red wavy lines; then this vertex has dashed degree $0$, so $C$ is a single isolated vertex carrying one red pair.
	\item[(iii)] $w_C=2$ and no vertex of $C$ carries two red wavy lines; then $C$ is a dashed path whose two endpoints carry the two single red wavy lines with opposite orientations.
\end{enumerate}
Let $\fH_1:=\{\fg\in \fH^N(x_{1:p}):0\leq v_\fg<2p,\ \text{no vertex carries exactly one red wavy line}\}$. 
For $\fg\in \fH_1$, if $v_\fg=0$, then every vertex is incident to two dashed edges and $\fg\in\widetilde \fT^N(x_{1:p})$. If $v_\fg\geq2$, then every dashed-connected component carrying red wavy lines is of type~(ii) in the above classification. Hence wavy lines occur in pairs, and each vertex carrying wavy lines is incident to zero dashed edges. Therefore $\fH_1\subset \widetilde \fT^N(x_{1:p})\cap \fH.$
Now fix $\fg\in \fH^N(x_{1:p})\setminus \fH_1$ with $v_\fg<2p$. Then $\fg$ has at least one component of type~(iii). 
For such a component, we move the
left-oriented single red wavy line to the
endpoint carrying the right-oriented single red wavy line, leaving the dashed part unchanged. Graphically, this operation is represented as
\begin{align}\label{shift}
	\begin{tikzpicture}[baseline=-5]
		\node[dot] (x3) at (0,0) {};
		\node[dot] (x4) at (0.4,0) {};
		\node[dot] (x5) at (0.8,0) {};
		\node[dot] (x6) at (1.2,0) {};
		\node[dot] (x7) at (1.6,0) {};
		\draw[Phi1] (x3) -- ++(-0.1,0.25);
		\draw[Phi1] (x7) -- ++(0.1,0.25);
		\draw[Cr,bend left=70] (x4) to (x5);
		\draw[Cr,bend left=70] (x3) to (x4);
		\draw[Cr,bend left=70] (x5) to (x6);
		\draw[Cr,bend left=70] (x6) to (x7);
	\end{tikzpicture}
=\begin{tikzpicture}[baseline=-5]
	\node[dot] (x3) at (0,0) {};
	\node[dot] (x4) at (0.4,0) {};
	\node[dot] (x5) at (0.8,0) {};
	\node[dot] (x6) at (1.2,0) {};
	\node[dot] (x7) at (1.6,0) {};
	\draw[Phi1] (x3) -- ++(-0.1,0.25);
	\draw[Phi1] (x3) -- ++(0.1,0.25);
	\draw[Cr,bend left=70] (x4) to (x5);
	\draw[Cr,bend left=70] (x3) to (x4);
	\draw[Cr,bend left=70] (x5) to (x6);
	\draw[Cr,bend left=70] (x6) to (x7);
\end{tikzpicture} + \left(\begin{tikzpicture}[baseline=-5]
\node[dot] (x3) at (0,0) {};
\node[dot] (x4) at (0.4,0) {};
\node[dot] (x5) at (0.8,0) {};
\node[dot] (x6) at (1.2,0) {};
\node[dot] (x7) at (1.6,0) {};
\draw[Phi1] (x3) -- ++(-0.1,0.25);
\draw[Phi1] (x7) -- ++(0.1,0.25);
\draw[Cr,bend left=70] (x4) to (x5);
\draw[Cr,bend left=70] (x3) to (x4);
\draw[Cr,bend left=70] (x5) to (x6);
\draw[Cr,bend left=70] (x6) to (x7);
\end{tikzpicture}
-\begin{tikzpicture}[baseline=-5]
\node[dot] (x3) at (0,0) {};
\node[dot] (x4) at (0.4,0) {};
\node[dot] (x5) at (0.8,0) {};
\node[dot] (x6) at (1.2,0) {};
\node[dot] (x7) at (1.6,0) {};
\draw[Phi1] (x3) -- ++(-0.1,0.25);
\draw[Phi1] (x3) -- ++(0.1,0.25);
\draw[Cr,bend left=70] (x4) to (x5);
\draw[Cr,bend left=70] (x3) to (x4);
\draw[Cr,bend left=70] (x5) to (x6);
\draw[Cr,bend left=70] (x6) to (x7);
\end{tikzpicture}\right)\,.
\end{align}
The term in the bracket is a remainder term of the form \eqref{gph2}. The first graph on the right-hand side has a red pair at a single vertex. This vertex is incident to exactly one dashed edge, the remaining vertices in the same dashed-connected component are incident to one or two
dashed edges, and the component contains no other wavy lines. Repeating this operation for every component of type (iii), we reduce \(\fg\) to a graph in \(\widetilde\fT^N(x_{1:p})\cap\fH\), up to a linear
combination of remainder terms of the form \eqref{gph2}.
Moreover, this operation does not change the dashed part of the graph, and therefore the multiplicity factor \(2^{m_{\fg}}\) remains unchanged.

Conversely, every graph in $\widetilde\fT^N(x_{1:p})\cap\mathcal H$ with \(v_{\fg}<2p\) arises uniquely from a graph in $\fH^N(x_{1:p})$ by the above procedure. Indeed, consider a dashed-connected component $C$. Since $C$ is connected, by \eqref{wC}, $w_C\in\{0,2\}$.
Then, by the definition of $\widetilde\fT^N(x_{1:p})$,
exactly one of the following alternatives holds:
\begin{enumerate}[leftmargin=2.2em]
	\item[(i')] $w_C=0$, and then $C$ is a dashed cycle.
	\item[(ii')] $w_C=2$ and $C$ is a single isolated vertex carrying one red pair.
	\item[(iii')] $w_C=2$ and $C$ is a dashed path whose one endpoint carries one red pair.
\end{enumerate}
The first two cases already correspond to components of type (i) and (ii) above. In the third case, we reconstruct the preimage by moving the left-oriented red wavy line in the pair to the other endpoint of the dashed path, leaving the right-oriented red wavy line fixed. This gives a component of type (iii), with two single red wavy lines of opposite orientations. Applying the operation in \eqref{shift} to this reconstructed component returns precisely the original component in $C$.
Doing this independently for every component of type (iii') gives a unique graph in \(\fH^N(x_{1:p})\). 
Hence there exists a one-to-one correspondence from graphs
in \(\fH^N(x_{1:p})\) to the graphs in \(\widetilde\fT^N(x_{1:p})\cap\fH\).
Therefore, \eqref{3.4} becomes
	\begin{align}\label{3.41}
		\prod_{k=1}^p\Wick{|u_N(x_k)|^2}-\Wick{\prod_{k=1}^p|u_N(x_k)|^2}=\sum_{\substack{\fg\in \widetilde\fT^N(x_{1:p})\cap\fH\\v_\fg<2p}}2^{m_\fg}\fg + \fR_N^{(1)}(u;x_{1:p})\,,
	\end{align}
where $\fR_N^{(1)}(u;x_{1:p})$ is a linear combination of terms of the form \eqref{gph2}. 

\medskip
\noindent\textbf{Step 2: Recursive reduction to blue graphs.}
Fix $\fg\in \widetilde \fT^N(x_{1:p})\cap \fH$ with $2\leq v_\fg<2p$.
Recall from Definition~\ref{def:graph} that \(U_\fg\) denotes the set of vertices carrying wavy lines, and \(F(\fg)\) denotes the product of all dashed edge factors in \(\fg\). Since $\fg\in \widetilde T^N(x_{1:p})$, every dashed-connected component of $\fg$ contains at most one vertex of $U_\fg$. Thus, the fixed dashed
factor $F(\fg)$ does not connect two different red-pair vertices, and except $F(\fg)$ the remaining part is exactly the Wick-ordered monomial supported on the vertices indexed by $U_\fg$, i.e.
\begin{align}\label{gph4}
	\fg = F(\fg)\, \Wick{\prod_{k\in U_\fg}|u_N(x_k)|^2}.
\end{align}
If $v_\fg=2$, then no further expansion is needed. The graph \(\fg\) contains a single red pair, which already represents \(\Wick{|u_N(x)|^2}\). We simply regard this pair as a blue pair, and
the resulting graph belongs to \(\fG^N(x_{1:p})\). 

If $4\leq v_\fg<2p$, we apply the induction hypothesis at the order $v_\fg/2$ to the
Wick-ordered monomial $\Wick{\prod_{k\in U_\fg}|u_N(x_k)|^2}$, while keeping the dashed factor $F(\fg)$ fixed. 
Since $C_\fg=-1$ for every $\fg\in\fG^N(x_k,k\in V)$ and $V\subseteq[p]$ with $v_\fg=2|V|$, by \eqref{3.1'}, we have
\begin{align*}
	\Wick{\prod_{k\in U_\fg}|u_N(x_k)|^2} = - \sum_{\fg'\in \fG^N(x_k,\,k\in U_\fg)} C_{\fg'}\,\fg'+ \fR_N(u;x_k,\,k\in U_\fg),
\end{align*}
where $\fR_N(u;x_k,\,k\in U_\fg)$ is again a finite linear combination of remainder terms of the form in \eqref{gph2}.
Here we include the completely uncontracted blue graph $\fg_0$ in the sum with the convention $C_{\fg_0}=-1$.
Multiplying by $F(\fg)$, we obtain
\begin{align}\label{3.6}
	\fg=-\sum_{\fg'\in\mathcal G^N(x_k,k\in U_\fg)}C_{\fg'}F(\fg)\,\fg' + F(\fg)\,\fR_N(u;x_k,k\in U_\fg)\,.
\end{align}
The same identity also holds in the case \(v_\fg=2\) by the convention \(C_{\fg'}=-1\) for the completely uncontracted blue graph on \(U_\fg\) and $\fR_N=0$.
The term $F(\fg)\,\fR_N$ in \eqref{3.6} is still a linear combination of remainder terms of the form in \eqref{gph2}, since multiplying by a fixed dashed factor does not change the fact that the two graphs in each remainder have the same dashed part and differ only by the position of one red wavy line.
Substituting \eqref{3.6} into \eqref{3.41} and absorbing all remainder terms into a single term $\fR_N(u;x_{1:p})$, we have
\begin{align}
	&\prod_{k=1}^p\Wick{|u_N(x_k)|^2}-\Wick{\prod_{k=1}^p|u_N(x_k)|^2}\nonumber\\
	=&\sum_{\substack{\fg\in \widetilde\fT^N(x_{1:p})\\v_\fg=0}}2^{m_\fg}\fg-\sum_{\substack{\fg\in \widetilde\fT^N(x_{1:p})\cap\fH\\2\leq v_\fg<2p}}\sum_{\fg'\in\mathcal G^N(x_k,k\in U_\fg)}2^{m_\fg}C_{\fg'}F(\fg)\,\fg'
	+ \fR_N(u;x_{1:p}).\label{3.6'}
\end{align}

\medskip
\noindent\textbf{Step 3: Identification of the coefficients.}
We now identify the coefficient of each graph
\(h\in\fG^N(x_{1:p})\) with \(v_h<2p\) in the expansion \eqref{3.6'}.
First suppose $v_h=0$. There are two possible contributions. The first is direct: if \(h\in\widetilde\fT^N(x_{1:p})\), then the first sum in the right-hand side of \eqref{3.6'} contributes \(2^{m_h}\).
The second contribution comes from the second sum. Such a term contributes to the coefficient of \(h\) precisely when
$F(\fg)\fg' = h$ for some $\fg\in \widetilde\fT^N(x_{1:p})\cap\fH$ with $2\leq v_\fg<2p$ and $\fg'\in\mathcal G^N(x_k,k\in U_\fg)$. 
By definition, this is exactly the relation $\fg'\prec_V h$ for some $V\subseteq[p]$ with $1\leq |V|<p$ and $\fg'\in\mathcal G^N(x_k,k\in V)$. 
In this case, the dashed loops of \(h\) are precisely the dashed loops of \(\fg\) together with those inside \(\fg'\). Hence $m_h=m_\fg+m_{\fg'}$. Therefore the total coefficient of \(h\) is
$$C_h=2^{m_h}\1_{\{h\in\widetilde\fT^N(x_{1:p})\}} - \sum_{\substack{V\subseteq[p]\\1\leq|V|<p}}\sum_{\fg'\prec_Vh}2^{m_h-m_{\fg'}}C_{\fg'}.$$
If $2\leq v_h<2p$, the only contribution to the coefficient of $h$ comes from the second sum in the right-hand side of \eqref{3.6'}. 
The same argument as above shows that its coefficient is
$$C_h=-\sum_{\substack{V\subseteq[p]\\1\leq|V|<p}}\sum_{\fg'\prec_V h}2^{m_h-m_{\fg'}}C_{\fg'}.$$
Combining the two cases, we obtain the recursive formula \eqref{def:C_g}. Substituting this coefficient identity into \eqref{3.6'} yields \eqref{3.1'}. Thus, the result follows.

\end{proof}

Now we can define the effective interaction functional $W_p^{\varepsilon,N}$.
For any graph $\fg\in \fG_p^N$, let $\fg(x_{1:p})\in \fG^N(x_{1:p})$ denote the graph obtained by removing all solid lines. 
Since the coefficient $C_{\fg(x_{1:p})}$ defined in \eqref{def:C_g} is purely combinatorial, it is independent of the positions $x_{1:p}$ and invariant under permutations of the vertices.
Since graphs in \(\fG_p^N\) are treated as labelled graphs, Lemma \ref{Lem2.2} yields
\begin{align*}
	pV_p^{\varepsilon,N}(u) = \int v^\varepsilon(x_2-x_1,\cdots,x_p-x_1) \left(\prod_{k=1}^p\Wick{|u_N(x_k)|^2} - \fR_N(u;x_{1:p})\right)\ud x_{1:p} -\sum_{\substack{\fg\in\fG_p^N\\v_\fg<2p}}C_{\fg(x_{1:p})}\,\fg .
\end{align*}
Then we define the effective interaction functional by keeping these blue-graph contributions and discarding the remainder:
\begin{align}\label{W}
	W_p^{\varepsilon,N}(u):=\frac1p\int v^\varepsilon(x_2-x_1,\cdots,x_p-x_1) \prod_{k=1}^p\Wick{|u_N(x_k)|^2}\ud x_{1:p}-\frac1p\sum_{\fg\in\fG_p^N:v_\fg<2p}C_{\fg(x_{1:p})}\,\fg\,.
\end{align}

We next regroup the blue graphs in the right-hand side of \eqref{W} according to the number $l:=\frac{v_\fg}{2}$ of blue pairs, i.e. according to the number of $\Wick{|u_N(x_i)|^2}$. 
Recall $\fM_{p,l}$ from \eqref{def:M}.
For fixed $l<p$, the contraction pattern is described by a multi-index $\fm=(m_{ab})_{1\le a<b\le p}\in \fM_{p,l}$, where $m_{ab}$ counts how many covariance lines join the
$a$-th and $b$-th vertices. The condition 
$\sum_{1\le a<b\le p} m_{ab}=p-l$ means that exactly $p-l$ contractions have been performed, leaving $l-\Wick{|u_N(x_i)|^2}$.
Thus \(\fm\) uniquely determines the following standard term:
\begin{align}\label{graph1}
	\fg_\fm(x_{1:p}):=\prod_{k=1}^l\Wick{|u_N(x_k)|^2}\prod_{1\leq a<b\leq p}G_N(x_a-x_b)^{m_{ab}},\quad \fg_\fm:=\int v^\varepsilon(x_2-x_1,\cdots,x_p-x_1)\fg_\fm(x_{1:p})\ud x_{1:p}.
\end{align}
Notice that \(\fg_{\fm}(x_{1:p})\) is a formal blue-graph term determined by \(\fm\); it may fail to belong to the graph class \(\fG^N(x_{1:p})\) if it violates one of the graph conventions. 
To record all blue graphs that produce the same contraction pattern, we define the coefficient:
\begin{align}\label{C_g2}
	C_\fm:=\begin{cases}
		-\frac1p \binom pl\;C_{\fg_\fm(x_{1:p})}, & \fg_\fm\in\fG_p^N,\\
		0, & \fg_\fm\notin\fG_p^N.
	\end{cases}
\end{align}
The factor $\binom pl$ comes from the choice of the $l$ vertices carrying blue pairs. By the symmetry of $v$ in \eqref{con:sym} and the invariance of $C_\fg$, all such choices are equivalent to the standard choice $\{x_1,\cdots,x_l\}$.
Then we integrate out the variables $x_{l+1},\dots,x_p$ against the Hartree kernel
$v^\varepsilon$ and the Green-function factors prescribed by $\fm$. This gives the effective kernel $f_l^{\varepsilon,N}$:
\begin{align}
	f_l^{\varepsilon,N}(x_2-x_1,\cdots,x_l-x_1) :=& \sum_{\mathbf m\in \fM_{p,l}}C_\fm\int_{\Lambda^{p-l}} v^\varepsilon(x_2-x_1,\cdots,x_p-x_1)\prod_{1\leq a<b\leq p}G_N(x_a-x_b)^{m_{ab}}\ud x_{l+1:p}\,.\label{def:F_lN}
\end{align}
By translation invariance, $f_0^{\varepsilon,N}$ and $f_1^{\varepsilon,N}$ are scalars.
Then we have
\begin{align*}
	\int_{\Lambda^l} f_l^{\varepsilon,N}(x_2-x_1,\cdots,x_l-x_1)\prod_{k=1}^l\Wick{|u_N(x_k)|^2}\ud x_{1:l}=&\sum_{\fm\in\fM_{p,l}} C_{\fm}\, \fg_\fm
	=-\sum_{\fg\in\fG_p^N:v_\fg=2l}\frac1p\, C_{\fg(x_{1:p})}\;\fg\,.
\end{align*}
Substituting this identity into \eqref{W} gives
\begin{align}\label{W1}
	W_p^{\varepsilon,N}(u)=\sum_{l=2}^{p}\int_{\Lambda^l} f_l^{\varepsilon,N}(x_2-x_1,\cdots,x_l-x_1)\prod_{k=1}^l\Wick{|u_N(x_k)|^2}\ud x_{1:l}+f_1^{\varepsilon,N}\int_\Lambda\Wick{|u_N(x)|^2}\ud x +f_0^{\varepsilon,N}\,,
\end{align}
where $f_p^{\varepsilon,N}=\frac1p v^\varepsilon$.

\begin{remark}\label{rmk 4.1}
	For $l\leq p-1$, the kernel \(f_l^{\varepsilon,N}\) should be viewed as the effective \(l\)-body interaction generated by the original \(p\)-body potential \(v^\varepsilon\). The \(l\) variables \(x_1,\ldots,x_l\) correspond to $\Wick{|u_N(x_i)|^2}$, while the remaining \(p-l\) variables are integrated out. Each multi-index \(\fm\) specifies a possible pattern of Wick contractions among the \(p\) vertices; the product of Green functions in \eqref{def:F_lN} is the weight of that pattern, and \(C_\fm\) is the corresponding combinatorial coefficient.
 The function $f_l^\varepsilon$ in \eqref{def:F_l} is exactly the $L^1(\Lambda^{l-1})$ limit of $f_l^{\varepsilon,N}$ as $N\to\infty$, see \eqref{limit} for the proof.
\end{remark}

\begin{remark}\label{rmk3.4}
	For illustration, we compute the coefficients $C_\fm$ and the kernels $f_l^\varepsilon$, $0\le l<p$, explicitly for $p=2$ and $p=3$.
	When \(p=2\), by \eqref{def:C_g}, for every $\fg\in\fG^N(x_{1:2})$ with $v_\fg\leq2$, $C_\fg=1$. Moreover, for every $\fm=(m_{12})\in\fM_{2,l}$ with $l=0,1$, the standard term $\fg_\fm$ defined in \eqref{graph1} belongs to $\fG_2^N$. Then, by \eqref{C_g2}, $C_\fm=-1$ if $m_{12}=1$ and $C_\fm=-\frac12$ if $m_{12}=2$. Therefore, 
	$$f_0^\varepsilon = -\frac12\int v^\varepsilon(x)G(x)^2\ud x,\quad f_1^\varepsilon = -\int v^\varepsilon(x)G(x)\ud x.$$
	
	When $p=3$, recall the explicit formulas for
	$f_i^\varepsilon$, $i=0,1,2$, in Remark~\ref{rmk2.2}. We explain below how their coefficients arise from \eqref{def:C_g} and \eqref{C_g2}, writing
	$\fm=(m_{12},m_{13},m_{23})$.
	For \(l=2\), the first two vertices \(x_1,x_2\) carry blue pairs and \(|\fm|=1\). 
	The term \(\fg_{(1,0,0)}\) is not in $\fG_3^N$ because its dashed edge connects the two blue-pair vertices. For the other two terms, the only graph $\fg'$ and $V\subset\{1,2,3\}$ satisfying $\fg'\prec_V\fg_\fm(x_{1:3})$ is $\fg'=\Wick{|u_N(x_1)|^2}\Wick{|u_N(x_2)|^2}$ and $V=\{1,2\}$. Then, by \eqref{def:C_g}, for $\fm=(0,1,0)$ or $(0,0,1)$,
	the graph coefficient $$C_\fm=-C_{\fg_\fm(x_{1:3})}=-1.$$
	These are the two coefficients appearing in $f_2^\varepsilon$.
	
	For \(l=1\), only \(x_1\) carries a blue pair and \(|\fm|=2\). Since in graphs, every vertex without wavy lines is incident to at least one dashed edge, the terms $\fg_{(2,0,0)}$ and $\fg_{(0,2,0)}$ are not in $\fG_3^N$.
	For \(\fm=(0,0,2)\), the only contribution $\fg'\prec_V\fg_\fm(x_{1:3})$ for some $V$ is $\fg'=\Wick{|u_N(x_1)|^2}$, so \(C_{\fg_\fm(x_{1:3})}=1\).
	For \(\fm=(1,1,0)\), $\fg_\fm$ is not in $\widetilde\fT^N(x_{1:3})$. 
	There are two contributions, corresponding to $V=\{1,2\}$ and $V=\{1,3\}$. In both cases, $\fg'$ is an admissible two-vertex blue subgraph with exactly one vertex carrying a blue pair. By $p=2$, $C_{\fg'}=1$.
	Hence $C_{\fg_\fm(x_{1:3})}=-2.$
	For the two path graphs $\fm=(1,0,1)$ and $\fm=(0,1,1)$, $\fg_\fm(x_{1:3})\in\widetilde\fT^N(x_{1:3})$, and there are two contributions: one from the one-vertex
	uncontracted graph $\Wick{|u_N(x_1)|^2}$ with coefficient \(-1\), and one from the two-vertex blue subgraph with coefficient \(1\). They cancel, giving
	$C_{\fg_\fm(x_{1:3})}=0.$ 
	By \eqref{C_g2}, these graph coefficients yield the
	coefficients in $f_1^\varepsilon$ recalled above.
	
	Finally, for \(l=0\), there are no blue pairs and \(|\fm|=3\). As before, if $m_{ij}=3$ for some $1\leq i<j\leq 3$, then $\fg_\fm$ is not in $\fG^N_3$ since one of the vertices is isolated. For $\fm=(1,1,1)$, \(m_{\fg_\fm}=1\), and there is no pair $(\fg',V)$ such that $\fg'\prec_V\fg_\fm(x_{1:3})$. Thus $C_{\fg_\fm(x_{1:3})}=2$. For the remaining $\fm\in\fM_{3,0}$, $m_{\fg_\fm}=0$ and $\fg_\fm$ is not in $\widetilde\fT^N(x_{1:3})$. Therefore, the only contribution comes from the \(p=2\) scalar subgraph, whose coefficient is \(1\), hence \(C_{\fg_\fm(x_{1:3})}=-1\). 
	Applying \eqref{C_g2}, with
	$C_\fm=-\frac13 C_{\fg_\fm(x_{1:3})}$ in this case,
	we obtain the formula for $f_0^\varepsilon$
	in Remark~\ref{rmk2.2}.
\end{remark}

\subsection{Convergence of measures}
Let \(p\ge2\), and let \(\mu_p\) be the \(\Phi^{2p}_2\) measure defined in \eqref{def:Phi measure}. We define the truncated Hartree measure \(\mu_p^{\varepsilon,N}\) by
\begin{align}
	\ud\mu_p^{\varepsilon,N}=(\cZ_p^{\varepsilon,N})^{-1}\exp(-W_p^{\varepsilon,N})\ud\mu_0,\quad \cZ_p^{\varepsilon,N}=\int\exp(-W_p^{\varepsilon,N})\ud\mu_0.\label{def:mu_p^N}
\end{align}
For fixed $\varepsilon>0$, the sequence $(W_p^{\varepsilon,N})_{N\in\NN}$ is Cauchy in $L^2(\PP)$; see \eqref{3.15} below. 
We denote its limit by $W_p^\varepsilon$. The exponential
integrability proved in \eqref{3.17} shows that the Hartree measure $\mu_p^\varepsilon$ in \eqref{def:Hartree measure} is well defined. 
The main result of this section is the following approximation theorem.

\begin{prop}\label{Prop}
	For $p\geq2$, $N\in\NN$ and $\varepsilon>0$, let $W_p^{\varepsilon,N}$ and $V_p$ be defined in \eqref{W} and \eqref{def:V^p} respectively, with $v$ satisfying Assumption \ref{assum:w}. Then, for every $m\geq1$,
	\begin{align}
		\lim_{\varepsilon\to0,N\to\infty}\EE|e^{- W_p^{\varepsilon,N}(u)}-e^{-V_p(u)}|^m&=0.\label{1.2}
	\end{align}
In particular, 
\begin{align}
	\lim_{\varepsilon\to0,N\to\infty}\cZ_p^{\varepsilon,N}=\cZ_p\,,\label{partition function conv}
\end{align}
and we have the convergence of the associated \(k\)-point correlation function for every $k\geq1$,
\begin{align}\label{matrices conv}
	\lim_{\varepsilon\to0,N\to\infty}\left\|\int|u^{\otimes k}\rangle\langle u^{\otimes k}|\ud\,(\mu_p^{\varepsilon,N}(u)-\mu_p(u))\right\|_{\textnormal{HS}}=0,
\end{align}
where $\|\cdot\|_{\textnormal{HS}}$ is the Hilbert-Schmidt norm in $\cH^k$.

\end{prop}

The proof follows the Nelson argument used in \cite[Proposition 4.3]{FKSS25} for the quartic case. The main difficulty for higher-order interactions is the uniform logarithmic lower bound for \(W_p^{\varepsilon,N}\). Indeed, the expansion of $\Wick{\prod_{j=1}^p|u_N(x_j)|^2}$ in \eqref{3.1'} produces the effective kernels $f_l^{\varepsilon,N}$ with no fixed sign. As mentioned before, the argument for the quartic case in \cite[Proposition 4.3]{FKSS25} cannot be applied.
The key observation is that the leading \(p\)-body term controls a \(p\)-th power of a localized average of $|u_N|^2$ due to the positivity and finite range of $v$, which in turn dominates all lower-order terms through interpolation estimates and a finite-covering argument.

We first prove this logarithmic lower bound needed for Nelson's argument.
For $1\leq l\leq p$ and every bounded periodic function $f\geq0$, set
\begin{align}\label{def:fI}
	\fI_{\varepsilon,l}(f)
	:=\int_{\Lambda^l}(\Pi_l v^\varepsilon)(x_2-x_1,\ldots,x_l-x_1)
	\prod_{j=1}^l f(x_j)\ud x_{1:l}.
\end{align}
Since $v\in C_c(\RR^{2(p-1)})$, there exists $R\geq1$ such that
\begin{align}\label{def:R}
	\supp v\subseteq\{x\in\mathbb R^{2(p-1)}:|x|\leq R\}.
\end{align}

	\begin{lem}\label{lem:lower bound of WN}
		For every $p\geq2$, there exists $C>0$ depending on $v$ and $p$ such that, for all $N\geq2$ and  $0<\varepsilon\leq\frac{1}{4R}$, almost surely
		\begin{align}\label{lower bound:W_N}
			W_p^{\varepsilon,N}\geq-C(\log N)^p \,.
		\end{align}
	\end{lem}

\begin{proof}
	Recall the definition of $f_l^{\varepsilon,N}$ in \eqref{def:F_lN}.
	Since $|G_N(x)|\leq G_N(0)\lesssim\log N$ and $0\leq v^\varepsilon\in L^1(\Lambda^{p-1})$, we have
	\begin{align*}
		|f_0^{\varepsilon,N}|\lesssim(\log N)^p,\quad
		|f_1^{\varepsilon,N}|\lesssim(\log N)^{p-1},\quad |f_l^{\varepsilon,N}|\lesssim(\log N)^{p-l}\Pi_lv^\varepsilon,\, 2\leq l\leq p-1.
	\end{align*}
	Then, we use $\Wick{|u_N(x)|^2}=|u_N(x)|^2-G_N(0)$ and expand the product in \eqref{W1} to obtain
	\begin{align}
		W_p^{\varepsilon,N}(u)
		\geq\frac1p\fI_{\varepsilon,p}(|u_N|^2)-C\left(\sum_{l=1}^{p-1}(\log N)^{p-l}\fI_{\varepsilon,l}(|u_N|^2)+(\log N)^p\right)\,.\label{3.31}
	\end{align}
	We claim that, for every \(1\le l\le p-1\), there exists \(C_l>0\), independent of
	\(\varepsilon\), such that for every bounded periodic function $f\geq0$,
	\begin{align}\label{If}
		\fI_{\varepsilon,l}(f)
		\leq C_l \,\fI_{\varepsilon,p}(f)^{l/p}.
	\end{align}
	Since $u_N=\rho_N*u$ is bounded in $\Lambda$ almost surely, we take $f=|u_N|^2$ in \eqref{If}. 
	Young's inequality
	then yields, for every $1\le l\le p-1$,
	\[
	C(\log N)^{p-l}\fI_{\varepsilon,l}(|u_N|^2)
	\le \frac{1}{2p(p-1)}\fI_{\varepsilon,p}(|u_N|^2)
	+C_l'(\log N)^p
	\]
	for some constant $C_l'>0$. Summing over $l$ proves \eqref{lower bound:W_N}, once \eqref{If} is established.
	
	In the following, we prove \eqref{If}. 
	By Assumption \ref{assum:w}, $\widehat v\geq0$ is continuous and $\widehat v(0)=1$. Hence
	$v(0)=\int\widehat v(\xi)\ud \xi>0$. By the continuity of \(v\), there exist $r_0>0$ and $a_0>0$ such that
	\begin{align}\label{lower bound v}
		v(z)\geq a_0\,,\quad \text{for all}\; |z|\leq 2r_0\sqrt{p-1}.
	\end{align}
	We choose $r_0$ smaller if necessary so that $2r_0\sqrt{p-1}\le R$. Then \(r_0\varepsilon\le 1/8\) whenever \(\varepsilon R\le 1/4\).
	Hence all balls of radius \(r_0\varepsilon\)
	on \(\Lambda\) have the Euclidean area \(\pi r_0^2\varepsilon^2\).
	We shall use this fact below without further comment.
	Moreover, since $\varepsilon R\leq1/4$, for every $y\in\Lambda^{p-1}$ with \(d(y)\le R\varepsilon\), there exists a unique \(m^y\in\mathbb Z^{2(p-1)}\) such that
	$|y+m^y|=d(y)$. Therefore, \eqref{poisson} becomes
	\begin{align}
		v^\varepsilon(y)
		=\varepsilon^{-2(p-1)}
		v\bigl(\varepsilon^{-1}(y+m^y)\bigr).\label{poisson1}
	\end{align}
	If \(d(y)>R\varepsilon\), then \(v^\varepsilon(y)=0\) by \eqref{def:R}.
	It follows that
	\begin{align}
		0\le v^\varepsilon(y)
		\le
		\|v\|_{L^\infty}\varepsilon^{-2(p-1)}
		\1_{\{d(y)\le R\varepsilon\}}.
		\label{periodiz}
	\end{align}
	Since \(2r_0\sqrt{p-1}\le R\), by \eqref{lower bound v}, we have
	\begin{align}
		v^\varepsilon(y)
		\ge a_0\varepsilon^{-2(p-1)},
		\quad\text{whenever}\quad
		d(y)\le2r_0\varepsilon\sqrt{p-1}.
		\label{periodiz1}
	\end{align}
	For later use, we record the corresponding representation of the marginals. Let \(2\le l\le p\) and
	\(y_{2:l}\in\Lambda^{l-1}\). Integrating \eqref{poisson} in the last \(p-l\) variables and changing the variables
	$y_j+m_j=\varepsilon z_j$ for $l+1\le j\le p$, we obtain
	\begin{align}
		(\Pi_l v^\varepsilon)(y_{2:l})
		={}&
		\varepsilon^{-2(l-1)}
		\sum_{m_{2:l}\in(\mathbb Z^2)^{l-1}}
		\int_{(\mathbb R^2)^{p-l}}
		v\Bigl(
		\varepsilon^{-1}(y_2+m_2),\ldots,
		\varepsilon^{-1}(y_l+m_l),
		z_{l+1:p}
		\Bigr)
		\,\ud z_{l+1:p}.
		\label{marginal-rescaling}
	\end{align}
	For $y\in\Lambda$, define the local average
	\begin{align}
		B^\varepsilon_1 f(y):=\varepsilon^{-2}\int_{B(y,r_0\varepsilon)} f(x)\ud x,\quad B^\varepsilon_2 f(y):=\varepsilon^{-2}\int_{B(y,(R+r_0)\varepsilon)} f(x)\ud x.\label{def:local average}
	\end{align}
	
	Now we estimate $\fI_{\varepsilon,p}(f)$. 
	Let \begin{align}
		A:=\{(x_1,\cdots,x_p)\in\Lambda^p: \text{there exists}\; y\in\Lambda,\,\text{such that}\; x_i\in B(y,r_0\varepsilon), 1\leq i\leq p\}.\label{def:A}
	\end{align}
	Then for every \(x_{1:p}\in\Lambda^p\), we have
	\begin{align}\label{3}
		\int_\Lambda
		\prod_{j=1}^p \1_{B(y,r_0\varepsilon)}(x_j)\,\ud y
		\le \1_A(x_{1:p})|B(x_1,r_0\varepsilon)| = \1_A(x_{1:p})\, \pi \, r_0^2\,\varepsilon^2.
	\end{align}
	If $x_{1:p}\in A$, then $d(x_2-x_1,\cdots,x_p-x_1)\leq2r_0\varepsilon\sqrt{p-1}$. Thus, by \eqref{periodiz1}, we have
	\begin{align}
		v^\varepsilon(x_2-x_1,\ldots,x_p-x_1)
		\ge
		a_0\,\varepsilon^{-2(p-1)}.\label{2}
	\end{align}
	Using $f\geq0$ and \eqref{3}, we obtain
	\begin{align}
		\fI_{\varepsilon,p}(f)
		&\geq a_0\,\varepsilon^{-2(p-1)}
		\int_{\Lambda^p}\1_A(x_{1:p})
		\prod_{j=1}^p f(x_j)\,\ud x_{1:p}\geq a_0\pi^{-1}r_0^{-2}
		\int_\Lambda (B_1^\varepsilon f(y))^p\,\ud y .\label{A}
	\end{align}
	Moreover, by Fubini's theorem,
	\[\int_\Lambda B_1^\varepsilon f(y)\ud y
	=\varepsilon^{-2}|B(0,r_0\varepsilon)|\int_\Lambda f(x)\ud x=\pi r_0^2\int_\Lambda f(x)\ud x.\]
	Since \(|\Lambda|=1\), H\"older's inequality and \eqref{A} imply
	$$\fI_{\varepsilon,1}(f) = \int_\Lambda f(x)\ud x \lesssim \fI_{\varepsilon,p}(f)^\frac1p.$$
	This proves \eqref{If} for $l=1$.
	
	It remains to consider \(2\le l\le p-1\). Define
	\begin{align}\label{def:dl}
		\widetilde d_l(x_{1:l}):=\max_{2\le i\le l}d(x_i-x_1).
	\end{align} 
	Applying \eqref{marginal-rescaling} with
	\(y_i=x_i-x_1\), we observe that at most one term in the sum over \(m_{2:l}\) can be non-zero, and by \eqref{def:R}, this term vanishes
	unless $\widetilde d_l(x_{1:l})\le R\varepsilon.$
	Since \(v\in L^\infty\), we therefore obtain
	\begin{align}
		(\Pi_l v^\varepsilon)
		(x_2-x_1,\ldots,x_l-x_1)
		\le{}&
		\|v\|_{L^\infty}
		(\pi R^2)^{p-l}
		\varepsilon^{-2(l-1)}
		\1_{\{
			\widetilde d_l(x_{1:l})\le R\varepsilon
			\}}.
		\label{marginal-pointwise}
	\end{align}
	Consequently,
	\begin{align}
		\fI_{\varepsilon,l}(f)
		&\leq\|v\|_{L^\infty}(\pi R^2)^{p-l}\varepsilon^{-2(l-1)}
		\int_{\Lambda^l}\1_{\{\widetilde d_l(x_{1:l})\leq R\varepsilon\}}\prod_{j=1}^lf(x_j)\ud x_{1:l}.\label{1'}
	\end{align}
	If $\widetilde d_l(x_{1:l})\leq R\varepsilon$, then
	\[B(x_1,r_0\varepsilon)\subseteq\{y\in\Lambda:x_1,\ldots,x_l\in B(y,(R+r_0)\varepsilon)\}.\]
	Therefore,
	\begin{align}
		\pi\,r_0^2\varepsilon^2\,\1_{\{\widetilde d_l(x_{1:l})\leq R\varepsilon\}}\leq
		\int_\Lambda\prod_{j=1}^l\1_{B(y,(R+r_0)\varepsilon)}(x_j)\,\ud y.\label{1}
	\end{align}
	Substituting into \eqref{1'}, it follows that
	\begin{align}
		\fI_{\varepsilon,l}(f)
		\leq \|v\|_{L^\infty}(\pi R^2)^{p-l}(\pi r_0^2)^{-1}\int_\Lambda (B_2^\varepsilon f(y))^l\,\ud y.\label{B}
	\end{align}
	
	Finally, choose $M=M(R,r_0)$ and $z_1,\cdots,z_M\in\RR^2$ such that $B(0,r_0+R)\subseteq\cup_{i=1}^MB(z_i,r_0)$. Then, for every $\varepsilon>0$ and $y\in\Lambda$,
	\begin{align}
		B(y,(R+r_0)\varepsilon)\subseteq\cup_{i=1}^MB(y+\varepsilon z_i,r_0\varepsilon),\label{4}
	\end{align}
	and hence
	\[B_2^\varepsilon f(y)\leq\sum_{m=1}^M B_1^\varepsilon f(y+\varepsilon z_m).\]
	Then, by H\"older's inequality, Minkowski's inequality and translation invariance on the torus, we have
	\begin{align*}
		\|B_2^\varepsilon f\|_{L^l(\Lambda)}\leq\|B_2^\varepsilon f\|_{L^p(\Lambda)}
		\leq M\|B_1^\varepsilon f\|_{L^p(\Lambda)}.
	\end{align*}
	Combining this with \eqref{A} and \eqref{B}, for $2\leq l\leq p-1$, we have
	\[\fI_{\varepsilon,l}(f)
	\lesssim_l\|B_2^\varepsilon f\|_{L^l(\Lambda)}^l
	\leq M^l\|B_1^\varepsilon f\|_{L^p(\Lambda)}^l
	\lesssim \fI_{\varepsilon,p}(f)^{l/p}.
	\]
\end{proof}

We then prove the $L^2(\mathbb P)$ convergence of the interaction functionals. 
The estimate below is similar as \cite[Lemma~4.1]{FKSS25} with combinatorial graph notations.

\begin{prop}
	Under the assumptions of Proposition~\ref{Prop}, for every $p\geq2$, it holds that
\begin{align}
	\lim_{\varepsilon\to0,N\to\infty}\EE|W_p^{\varepsilon,N}(u)-V_p(u)|^2&=0.\label{1.1}
\end{align}
\end{prop}

\begin{proof}
By the triangle inequality, we divide the proof into two steps. The first is to compare the effective Hartree functional $W_p^{\varepsilon,N}$ with the nonlocal Wick interaction $V_p^{\varepsilon,N}$, where the difference consists of the graphical remainder terms from Lemma \ref{Lem2.2}, and the main task is to show that these terms vanish in $L^2$ uniformly in $N$ as $\varepsilon\to 0$. The second step is to compare $V_p^{\varepsilon,N}$ with the local Wick interaction $V_p^N$, which is achieved by the same estimate using Wick's rule. The conclusion then follows from the classical convergence $V_p^N\to V_p$ as $N\to\infty$.

\medskip
\noindent\textbf{Step 1: Comparison of
	$W_p^{\varepsilon,N}$ and $V_p^{\varepsilon,N}$.}
We first estimate $\EE|W_p^{\varepsilon,N}(u)-V_p^{\varepsilon,N}(u)|^2$. By Lemma \ref{Lem2.2} and \eqref{W}, the difference $W_p^{\varepsilon,N}(u)-V_p^{\varepsilon,N}(u)$
is a finite linear combination of remainder terms of the form \eqref{gph2}, with the solid Hartree edges attached. 
It is therefore enough to estimate such term in $L^2(\mathbb P)$:
\begin{align}\label{gph3}
	\begin{tikzpicture}[baseline=-5]
		\node[dot] (x1) at (0,0) {};
		\draw[Phi1] (x1) -- ++(-0.1,0.25);
		\draw[Phi1] (x1) -- ++(0.1,0.25);
		\node[dot] (x2) at (0.3,0) {};
		\draw[Phi1] (x2) -- ++(-0.1,0.25);
		\node[dot] (x1') at (0.6,0) {};
		\node[dot] (x2') at (0.9,0) {};
		\draw[Phi1] (x2') -- ++(0.1,0.25);
		\node at (1.3,0) {$\cdots$};
		\node[dot] (x3) at (1.6,0) {};
		\draw[Phi1] (x3) -- ++(-0.1,0.25);
		\node[dot] (x5) at (1.9,0) {};
		\draw[Phi1] (x5) -- ++(0.1,0.25);
		\draw[C,bend left=60] (x5) to (x1);
		\draw[C,bend left=60] (x2) to (x1);
		\draw[C,bend left=60] (x1') to (x1);
		\draw[C,bend left=60] (x2') to (x1);
		\draw[C,bend left=60] (x3) to (x1);
		\draw[Cr,bend left=50] (x2) to (x1');
		\draw[Cr,bend left=50] (x1') to (x2');
		\draw[Cr,bend left=70] (x3) to (x5);
	\end{tikzpicture}-\begin{tikzpicture}[baseline=-5]
		\node[dot] (x1) at (0,0) {};
		\draw[Phi1] (x1) -- ++(-0.1,0.25);
		\draw[Phi1] (x1) -- ++(0.1,0.25);
		\node[dot] (x2) at (0.3,0) {};
		\draw[Phi1] (x2) -- ++(-0.1,0.25);
		\draw[Phi1] (x2) -- ++(0.1,0.25);
		\node[dot] (x1') at (0.6,0) {};
		\node[dot] (x2') at (0.9,0) {};
		\node at (1.3,0) {$\cdots$};
		\node[dot] (x3) at (1.6,0) {};
		\draw[Phi1] (x3) -- ++(-0.1,0.25);
		\node[dot] (x5) at (1.9,0) {};
		\draw[Phi1] (x5) -- ++(0.1,0.25);
		\draw[C,bend left=60] (x5) to (x1);
		\draw[C,bend left=60] (x2) to (x1);
		\draw[C,bend left=60] (x1') to (x1);
		\draw[C,bend left=60] (x2') to (x1);
		\draw[C,bend left=60] (x3) to (x1);
		\draw[Cr,bend left=50] (x2) to (x1');
		\draw[Cr,bend left=50] (x1') to (x2');
		\draw[Cr,bend left=70] (x3) to (x5);
	\end{tikzpicture}\,.
\end{align}
By Wick's rule, each contraction between a red wavy line at some $x_i$ and an oppositely oriented red wavy line at some $y_j$ produces a factor $G_N(x_i-y_j)$. Since the two graphs in \eqref{gph3} differ only by the position of one red wavy line, every resulting term contains exactly
one distinguished difference factor of the form $G_N(x_i-y_j)-G_N(x_{i'}-y_j)$ for some $i,i',j\in\{1,\dots,p\}$ and $i\neq i'$. Hence, up to a multiplicative constant, the squared $L^2(\PP)$ norm of \eqref{gph3} is bounded by
\begin{align}
	&\int_{\Lambda^p}\int_{\Lambda^p}\ud x_{1:p}\ud y_{1:p}v^\varepsilon(x_2-x_1,\cdots,x_p-x_1)v^\varepsilon(y_2-y_1,\cdots,y_p-y_1)\times\nonumber\\
	&\times\prod_{\{i,j\}\in\Pi}|G_N(x_i-x_j)G_N(y_i-y_j)|\prod_{\{i,j\}\in\widetilde\Pi}|G_N(x_i-y_j)||G_N(x_1-y_1)-G_N(x_2-y_1)|.\label{3.47}
\end{align}
Here $\Pi$ and $\widetilde\Pi$ are finite collections of pairs, determined by the contraction pattern, and satisfy $2|\Pi|+|\widetilde\Pi|=2p-1$ and $i<j$ for any $\{i,j\}\in\Pi$.
Using the periodized representation of \(v^\varepsilon\) in \eqref{poisson}, for each fixed $m_2,\cdots,m_p\in\ZZ^2$ and $n_2,\cdots,n_p\in\ZZ^2$, we change the variables by $x_j-x_1+m_j=\varepsilon h_j$ and
$y_j-y_1+n_j=\varepsilon z_j$ for $2\le j\le p$, as well as
$h_1=x_1-y_1$. Using periodicity of $G_N$, we obtain that \eqref{3.47} is bounded by

\begin{align}
&\int_{2\Lambda}\int_{(\RR^2)^{p-1}}\int_{(\RR^2)^{p-1}}v(h_{2:p})v(z_{2:p})\prod_{\substack{\{i,j\}\in\Pi\\i,j\neq1}}|G_N(\varepsilon(h_i-h_j))G_N(\varepsilon(z_i-z_j))|\prod_{\{i,j\}\in\widetilde\Pi}|G_N(\varepsilon(h_i-z_j)+h_1)|\times\nonumber\\
&\times|G_N(h_1)-G_N(\varepsilon h_2+h_1)||G_N(h_1)|^m\prod_{\{1,j\}_{j\geq2}\in\Pi}|G_N(\varepsilon h_j)G_N(\varepsilon z_j)|\ud h_1 \ud h_{2:p}\ud z_{2:p},\label{1.6}
\end{align}
for some $m\geq0$. 
We begin by estimating the integral in $h_1$. By H\"older's inequality, 
\begin{align}
	&\int_{2\Lambda}|G_N(h_1)-G_N(\varepsilon h_2+h_1)|\prod_{\{i,j\}\in\widetilde\Pi}|G_N(\varepsilon(h_i-z_j)+h_1)||G_N(h_1)|^m\ud h_1\nonumber\\
	\leq&\left(\int_{2\Lambda}|G_N(h_1)-G_N(\varepsilon h_2+h_1)|^2\ud h_1\right)^\frac12\left(\int_{2\Lambda}\prod_{\{i,j\}\in\widetilde\Pi}|G_N(\varepsilon(h_i-z_j)+h_1)|^2|G_N(h_1)|^{2m}\ud h_1\right)^\frac12.\label{4.23}
\end{align}
For the second term in \eqref{4.23}, since $v$ is compact supported, for some $q>0$ and $C>0$, whenever $v(h_{2:p})v(z_{2:p})\neq0$,
\begin{align*}
	&\int_{2\Lambda}\prod_{\{i,j\}\in\widetilde\Pi}|G_N(\varepsilon(h_i-z_j)+h_1)|^2|G_N(h_1)|^{2m}\ud h_1
	\lesssim\int_{|x|\leq C}(1+|\log d(x)|)^q\ud x\lesssim1,
\end{align*}
where in the first inequality we used H\"older's inequality and \eqref{C3_2}.
For the first term in \eqref{4.23}, by \eqref{C0} and \eqref{C0_1}, for almost every $x,y\in\RR^2$, $$|G_N(x)-G_N(y)|\leq\int_\Lambda|\widetilde\rho_N(z)||G(x-z)-G(y-z)|\ud z\lesssim \int_\Lambda|\widetilde\rho_N(z)|\left(|x-y|+\left|\log\frac{d(x-z)}{d(y-z)}\right|\right)\ud z,$$
where $\widetilde\rho_N=\rho_N*\rho_N$.
Hence, by H\"older's inequality, whenever $v(h_{2:p})\neq0$,
\begin{align}
	\int_{2\Lambda}|G_N(x)-G_N(\varepsilon h_2+x)|^2\ud x
	\lesssim&\|\widetilde\rho\|_{L^1(\RR^2)}\left(\varepsilon^2\|\widetilde\rho\|_{L^1(\RR^2)}+
	\int_\Lambda\ud z\int_{2\Lambda}\ud x \;|\widetilde\rho_N(z)|\left|\log\frac{d(\varepsilon h_2+x-z)}{d(x-z)}\right|^2\right).\label{integral1}
\end{align}
To control the integral in the right-hand side, fix $z\in\Lambda$ and $h_2\in\RR^2$ such that $v(h_{2:p})\neq0$, and decompose $2\Lambda$ into two regions: $\Omega_1=\{x:d(x-z)<2d(\varepsilon h_2)\}$ and $\Omega_2=2\Lambda\backslash\Omega_1$. 
On $\Omega_1$, we have $d(\varepsilon h_2+x-z)\le 3d(\varepsilon h_2)$, and therefore,
since $v$ is compactly supported, for some constant $C>0$,
$$\int_{\Omega_1}\left|\log\frac{d(\varepsilon h_2+x-z)}{d(x-z)}\right|^2\ud x\lesssim\int_{|x|\leq C}\1_{d(x)\leq3d(\varepsilon h_2)}|\log d(x)|^2\ud x\lesssim\int_\Lambda\1_{|x|\leq3d(\varepsilon h_2)}|\log|x||^2\ud x.$$
Let $\varepsilon$ be small enough so that $|\varepsilon h_2|\le\varepsilon R \leq\frac12$. Then
$3d(\varepsilon h_2)=3\varepsilon|h_2|\lesssim\varepsilon$, and the above is bounded by
$$\int_{|x|\lesssim\varepsilon}|\log|x||^2\ud x = \varepsilon^2\int_{|x|\lesssim1}|\log|\varepsilon x||^2\ud x\lesssim\varepsilon^2|\log\varepsilon|^2.$$
On $\Omega_2$, since $$\left|\frac{d(\varepsilon h_2+x-z)}{d(x-z)}-1\right|\leq\frac{d(\varepsilon h_2)}{d(x-z)}\leq\frac12,$$
we can use $|\log(1+t)|\lesssim |t|$ for $|t|<\frac12$ to obtain
$$\left|\log\frac{d(\varepsilon h_2+x-z)}{d(x-z)}\right|\lesssim \left|\frac{d(\varepsilon h_2+x-z)}{d(x-z)}-1\right|\leq\frac{d(\varepsilon h_2)}{d(x-z)}.$$
Then, similarly as above, when $\varepsilon$ is small enough,
\begin{align*}
	\int_{\Omega_2}\left|\log\frac{d(\varepsilon h_2+x-z)}{d(x-z)}\right|^2\ud x
	\lesssim& \int_{3\Lambda}\frac{d(\varepsilon h_2)^2}{d(x)^2}\1_{d(x)\geq2d(\varepsilon h_2)}\ud x\lesssim d(\varepsilon h_2)^2|\log d(\varepsilon h_2)|\lesssim\varepsilon^2|\log\varepsilon|.
\end{align*}
Substituting these bounds into \eqref{integral1} and using \eqref{C2_1}, we get
\begin{align*}
	\int_{2\Lambda}|G_N(x)-G_N(\varepsilon h_2+x)|^2\ud x
	\lesssim\varepsilon^2|\log\varepsilon|^2\|\widetilde\rho\|_{L^1(\RR^2)}^2.
\end{align*}
Combining this with Young's inequality and \eqref{C3_2}, we obtain that \eqref{1.6} is bounded by a constant times
\begin{align*}
	&\varepsilon|\log\varepsilon|\iint v(h_{2:p})v(z_{2:p})\prod_{\substack{\{i,j\}\in\Pi\\i,j\neq1}}|G_N(\varepsilon(h_i-h_j))G_N(\varepsilon(z_i-z_j))|\prod_{\{1,j\}_{j\geq2}\in\Pi}|G_N(\varepsilon h_j)G_N(\varepsilon z_j)|\ud h_{2:p}\ud z_{2:p}\\
	\lesssim&\,\varepsilon|\log\varepsilon|\left(\int_{|x|\leq 2R}(1+|\log d(\varepsilon x)|)^{|\Pi|}\ud x\right)^2\lesssim\varepsilon|\log\varepsilon|^{2p},
\end{align*}
where $R$ is defined in \eqref{def:R}. In the last inequality, we used that $2|\Pi|\leq2p-2$.
We have therefore proved that $$\sup_N\EE|W_p^{\varepsilon,N}(u)-V_p^{\varepsilon,N}(u)|^2\lesssim\varepsilon|\log\varepsilon|^{2p}.$$

\medskip
\noindent\textbf{Step 2: Comparison of
	$V_p^{\varepsilon,N}$ and $V_p^N$.}
Since $\widehat v(0)=1$, we have
\begin{align*}
	V_p^{\varepsilon,N}(u)-V_p^N(u)=\frac1p\int v^\varepsilon(x_2-x_1,\cdots,x_p-x_1) (\Wick{|u_N(x_1)|^2\cdots|u_N(x_p)|^2}-\Wick{|u_N(x_1)|^{2p}})\ud x_{1:p}.
\end{align*}
By Wick's rule, $\EE|V_p^{\varepsilon,N}(u)-V_p^N(u)|^2$ can again be expanded as a finite sum of terms of the form
\begin{align*}
	&\int_{\Lambda^p}\int_{\Lambda^p}v^\varepsilon(x_2-x_1,\cdots,x_p-x_1)v^\varepsilon(y_2-y_1,\cdots,y_p-y_1)\prod_{\{i,j\}\in\Pi}G_N(x_i-y_j)\times\\
	&\times(G_N(x_1-y_1)-G_N(x_2-y_1))\ud x_{1:p}\ud y_{1:p}\,.
\end{align*}
Applying the estimate derived above for \eqref{1.6}, we obtain $$\sup_N\EE|V_p^{\varepsilon,N}(u)-V_p^N(u)|^2\lesssim\varepsilon|\log\varepsilon|^{2p}.$$
Together with Step~1 and the convergence $V_p^N\to V_p$
in $L^2(\PP)$, this proves \eqref{1.1}.

\end{proof}

We are now in a position to prove Proposition \ref{Prop}. 
Since we already have the uniform lower bound on $W_p^{\varepsilon,N}$ in \eqref{lower bound:W_N}, the key point is to estimate the convergence rate of the Cauchy sequence $\EE |W_p^{\varepsilon,M}(u)-W_p^{\varepsilon,N}(u)|^2$.
Here we emphasize that the argument used in \cite[Lemma~4.6]{FKSS25} cannot be used here. In the quartic case, as $V^\varepsilon-W^\varepsilon$ belongs to the second Wiener chaos, the Taylor expansion of $\exp(V^\varepsilon-W^\varepsilon)$ is summable by hypercontractivity. 
For $p\ge3$, however, $V_p^{\varepsilon,N}-W_p^{\varepsilon,N}$ no longer belongs to the second Wiener chaos and similar argument leads to blow up for general $p$. 
Therefore, we prove the uniform exponential integrability by Nelson's argument directly for $W_p^{\varepsilon,N}$.
Here we use the remainder estimate derived from \eqref{1.6} and carefully balance the dependence on $\varepsilon$ and $N$ to conclude the proof.

\begin{proof}[Proof of Proposition \ref{Prop}]
	First, we estimate $\sup_{0<\varepsilon\le(4R)^{-1}}\EE |W_p^{\varepsilon,M}(u)-W_p^{\varepsilon,N}(u)|^2$ for $M>N$. 
	By the triangle inequality,
	\begin{align*}
		\EE|W_p^{\varepsilon,M}(u)-W_p^{\varepsilon,N}(u)|^2\leq2 \EE|V_p^{\varepsilon,M}(u)-V_p^{\varepsilon,N}(u)|^2+2\EE|(V_p^{\varepsilon,M}(u)-W_p^{\varepsilon,M}(u))-(V_p^{\varepsilon,N}(u)-W_p^{\varepsilon,N}(u))|^2\,.
	\end{align*}
We estimate the two terms on the right-hand side separately.
	
	\noindent {\bf Step 1.}
	We begin with $\EE|V_p^{\varepsilon,M}(u)-V_p^{\varepsilon,N}(u)|^2$ for $M> N$.
	Set $u_{N,M}:=u_M-u_N$. By \eqref{Green},
	we have\begin{align*}
		\EE(u_{N,M}(x)\overline{u_N(y)})=\rho_N*(\rho_M-\rho_N)*G(x-y)=:\rho_N*F_{N,M}(x-y),
	\end{align*}
and\begin{align*}
	\EE(u_{N,M}(x)\overline{u_{N,M}(y)})=(\rho_M-\rho_N)*F_{N,M}(x-y).
\end{align*}
	By \eqref{def:V_N}, it follows that
	\begin{align*}
		p(V_p^{\varepsilon,M}(u)-V_p^{\varepsilon,N}(u))
		=&\int_{\Lambda^p} v^\varepsilon(x_2-x_1,\cdots,x_p-x_1)\left(\Wick{\prod_{j=1}^p|u_N(x_j)+u_{N,M}(x_j)|^2}
		-\Wick{\prod_{j=1}^p|u_N(x_j)|^2}\right)\ud x_{1:p}\\
		=&\sum_{\substack{b_1+\cdots+b_{2p}>0\\b_i\in\{0,1\}}} V^\varepsilon_{N,M}(b_{1:2p})\,,
	\end{align*}
where $$V^\varepsilon_{N,M}(b_{1:2p})=\int_{\Lambda^p} v^\varepsilon(x_2-x_1,\cdots,x_p-x_1)\Wick{\prod_{k=1}^p\left(u_N(x_k)^{1-b_{2k-1}}\overline{u_N}(x_k)^{1-b_{2k}}u_{N,M}(x_k)^{b_{2k-1}}\overline{u_{N,M}}(x_k)^{b_{2k}}\right)}\ud x_{1:p}\,.$$
For fixed $b_{1:2p}$ with $b_1+\cdots+b_{2p}>0$ and $b_i\in\{0,1\}$, Wick's rule
implies that $\EE|V_{N,M}^\varepsilon(b_{1:2p})|^2$ can be written as a finite sum of
integrals involving products of covariance factors of the form
\[
G_N(x-y),\quad (\rho_N*F_{N,M})(x-y),\quad ((\rho_M-\rho_N)*F_{N,M})(x-y).
\]
Since $b_1+\cdots+b_{2p}>0$, at least one factor of the last two types must appear. 
Moreover, by \eqref{C3_2}, we bound all factors of the first
type by $|G_N|\lesssim 1+\log N$, and then apply H\"older's inequality to the remaining ones.
Since $\int_{\Lambda^{p-1}} v^\varepsilon(x)\ud x=1$, it follows that
\begin{align}\label{4.26}
	\EE|V_{N,M}^\varepsilon(b_{1:2p})|^2
	\lesssim(1+\log N)^r\|\rho_N*F_{N,M}\|_{L^{s+t}(\Lambda)}^s\|(\rho_M-\rho_N)*F_{N,M}\|_{L^{s+t}(\Lambda)}^t,
\end{align}
for some integers $r,s,t\ge0$ with $r+s+t=2p$ and $s+t\ge1$.
By \eqref{C3_1} (with $\widetilde\rho$ replaced by $\rho$), for any $M>N$ and $x\in\Lambda$, we have
\begin{align}\label{bound:GN}
	|F_{N,M}(x)|\lesssim(1+|\log(N|x|)|)\wedge\frac{1}{N^2|x|^2}.
\end{align}
Then, for any $q\geq1$, we have
\begin{align*}
	\int_{\Lambda}|F_{N,M}(x)|^q\ud x\lesssim&\int_{|x|\leq N^{-1}}(1+|\log(N|x|)|)^q\ud x +N^{-2q}\int_{N^{-1}<|x|\leq1}|x|^{-2q}\ud x\\
	=&2\pi N^{-2}\int_0^1\left(1+\log\frac{1}{r}\right)^qr\ud r + 2\pi N^{-2q}\int_{N^{-1}}^1 r^{1-2q}\ud r,
\end{align*}
and thus \begin{align}\label{F_{N,M}}
	\|F_{N,M}\|_{L^1(\Lambda)}\lesssim N^{-2}\log N\quad\text{and}\quad \|F_{N,M}\|_{L^q(\Lambda)}^q\lesssim N^{-2},\, q>1.
\end{align}
Using Young's convolution inequality together with $\int_{\Lambda}|\rho_N(x)|\ud x\leq\int_{\RR^2}|\rho(x)|\ud x$, we deduce from \eqref{4.26} that
\begin{align*}
	\EE|V_{N,M}^\varepsilon(b_{1:2p})|^2
	\lesssim(1+\log N)^r\|F_{N,M}\|_{L^{s+t}(\Lambda)}^{s+t}\lesssim N^{-2}(1+\log N)^{2p}.
\end{align*}
Therefore, \begin{align*}
	\EE|V^{\varepsilon,p}_M(u)-V^{\varepsilon,p}_N(u)|^2\lesssim N^{-2}(1+\log N)^{2p}.
\end{align*}

	\noindent {\bf Step 2.}
	Next we estimate $\EE|(V_p^{\varepsilon,M}(u)-W_p^{\varepsilon,M}(u))-(V_p^{\varepsilon,N}(u)-W_p^{\varepsilon,N}(u))|^2$ for $M\geq N$.
	It suffices to deal with each remainder term in \eqref{gph3}, which we denote by $\fR^{\varepsilon,p}_N$.
	For $\fR_M^{\varepsilon,p}-\fR_N^{\varepsilon,p}$, we add and subtract the graph obtained from $\fR_M^{\varepsilon,p}$ by replacing
	every dashed factor $G_M$ with $G_N$.
	Then the second part of $\fR_M^{\varepsilon,p}- \fR_N^{\varepsilon,p}$ has the form as $V_{N,M}^\varepsilon(b_{1:2p})$ with additional $G_N$ in the integral, and can be estimated as above. The first part is 
	\begin{align*}
		\left(\begin{tikzpicture}[baseline=-5]
			\node[dot] (x1) at (0,0) {};
			\draw[Phi1] (x1) -- ++(-0.1,0.25);
			\draw[Phi1] (x1) -- ++(0.1,0.25);
			\node[dot] (x2) at (0.3,0) {};
			\draw[Phi1] (x2) -- ++(0.1,0.25);
			\node[dot] (x1') at (0.6,0) {};
			\draw[Phi1] (x1') -- ++(-0.1,0.25);
			\node[dot] (x2') at (0.9,0) {};
			\node at (1.3,0) {$\cdots$};
			\node[dot] (x3) at (1.6,0) {};
			\draw[Phi1] (x3) -- ++(-0.1,0.25);
			\node[dot] (x4) at (1.9,0) {};
			\draw[Phi1] (x4) -- ++(-0.1,0.25);
			\draw[Phi1] (x4) -- ++(0.1,0.25);
			\node[dot] (x5) at (2.2,0) {};
			\draw[Phi1] (x5) -- ++(0.1,0.25);
			\draw[C,bend left=60] (x2) to (x1);
			\draw[C,bend left=60] (x1') to (x1);
			\draw[C,bend left=60] (x2') to (x1);
			\draw[C,bend left=60] (x3) to (x1);
			\draw[C,bend left=60] (x4) to (x1);
			\draw[C,bend left=60] (x5) to (x1);
			\draw[Cg,bend right=50] (x2) to (x2');
			\draw[Cg,bend left=70] (x1') to (x2');
			\draw[Cg,bend left=70] (x3) to (x5);
		\end{tikzpicture}-\begin{tikzpicture}[baseline=-5]
			\node[dot] (x1) at (0,0) {};
			\draw[Phi1] (x1) -- ++(-0.1,0.25);
			\draw[Phi1] (x1) -- ++(0.1,0.25);
			\node[dot] (x2) at (0.3,0) {};
			\node[dot] (x1') at (0.6,0) {};
			\draw[Phi1] (x1') -- ++(0.1,0.25);
			\draw[Phi1] (x1') -- ++(-0.1,0.25);
			\node[dot] (x2') at (0.9,0) {};
			\node at (1.3,0) {$\cdots$};
			\node[dot] (x3) at (1.6,0) {};
			\draw[Phi1] (x3) -- ++(-0.1,0.25);
			\node[dot] (x4) at (1.9,0) {};
			\draw[Phi1] (x4) -- ++(-0.1,0.25);
			\draw[Phi1] (x4) -- ++(0.1,0.25);
			\node[dot] (x5) at (2.2,0) {};
			\draw[Phi1] (x5) -- ++(0.1,0.25);
			\draw[C,bend left=60] (x2) to (x1);
			\draw[C,bend left=60] (x1') to (x1);
			\draw[C,bend left=60] (x2') to (x1);
			\draw[C,bend left=60] (x3) to (x1);
			\draw[C,bend left=60] (x4) to (x1);
			\draw[C,bend left=60] (x5) to (x1);
			\draw[Cg,bend right=50] (x2) to (x2');
			\draw[Cg,bend left=70] (x1') to (x2');
			\draw[Cg,bend left=70] (x3) to (x5);
		\end{tikzpicture}\right)-
	\left(\begin{tikzpicture}[baseline=-5]
		\node[dot] (x1) at (0,0) {};
		\draw[Phi1] (x1) -- ++(-0.1,0.25);
		\draw[Phi1] (x1) -- ++(0.1,0.25);
		\node[dot] (x2) at (0.3,0) {};
		\draw[Phi1] (x2) -- ++(0.1,0.25);
		\node[dot] (x1') at (0.6,0) {};
		\draw[Phi1] (x1') -- ++(-0.1,0.25);
		\node[dot] (x2') at (0.9,0) {};
		\node at (1.3,0) {$\cdots$};
		\node[dot] (x3) at (1.6,0) {};
		\draw[Phi1] (x3) -- ++(-0.1,0.25);
		\node[dot] (x4) at (1.9,0) {};
		\draw[Phi1] (x4) -- ++(-0.1,0.25);
		\draw[Phi1] (x4) -- ++(0.1,0.25);
		\node[dot] (x5) at (2.2,0) {};
		\draw[Phi1] (x5) -- ++(0.1,0.25);
		\draw[C,bend left=60] (x2) to (x1);
		\draw[C,bend left=60] (x1') to (x1);
		\draw[C,bend left=60] (x2') to (x1);
		\draw[C,bend left=60] (x3) to (x1);
		\draw[C,bend left=60] (x4) to (x1);
		\draw[C,bend left=60] (x5) to (x1);
		\draw[Cr,bend right=50] (x2) to (x2');
		\draw[Cr,bend left=70] (x1') to (x2');
		\draw[Cr,bend left=70] (x3) to (x5);
	\end{tikzpicture}-\begin{tikzpicture}[baseline=-5]
		\node[dot] (x1) at (0,0) {};
		\draw[Phi1] (x1) -- ++(-0.1,0.25);
		\draw[Phi1] (x1) -- ++(0.1,0.25);
		\node[dot] (x2) at (0.3,0) {};
		\node[dot] (x1') at (0.6,0) {};
		\draw[Phi1] (x1') -- ++(0.1,0.25);
		\draw[Phi1] (x1') -- ++(-0.1,0.25);
		\node[dot] (x2') at (0.9,0) {};
		\node at (1.3,0) {$\cdots$};
		\node[dot] (x3) at (1.6,0) {};
		\draw[Phi1] (x3) -- ++(-0.1,0.25);
		\node[dot] (x4) at (1.9,0) {};
		\draw[Phi1] (x4) -- ++(-0.1,0.25);
		\draw[Phi1] (x4) -- ++(0.1,0.25);
		\node[dot] (x5) at (2.2,0) {};
		\draw[Phi1] (x5) -- ++(0.1,0.25);
		\draw[C,bend left=60] (x2) to (x1);
		\draw[C,bend left=60] (x1') to (x1);
		\draw[C,bend left=60] (x2') to (x1);
		\draw[C,bend left=60] (x3) to (x1);
		\draw[C,bend left=60] (x4) to (x1);
		\draw[C,bend left=60] (x5) to (x1);
		\draw[Cr,bend right=50] (x2) to (x2');
		\draw[Cr,bend left=70] (x1') to (x2');
		\draw[Cr,bend left=70] (x3) to (x5);
	\end{tikzpicture}\right)\,,
	\end{align*}
	where the red wavy line represents $u_M$ and green (resp. black) dashed line represents $G_M$ (resp. $G_N$). A replacement of the dashed factors, one edge at a time, writes the preceding difference as a finite sum of terms of the form
	\begin{align}
		&\int_{\Lambda^p}\ud x_{1:p} v^\varepsilon(x_2-x_1,\cdots,x_p-x_1)F_{N,M}(x_2-x_1)\prod_{\{i_1,j_1\}\in\Pi_1}G_M(x_{i_1}-x_{j_1})\prod_{\{i_2,j_2\}\in\Pi_2}G_N(x_{i_2}-x_{j_2})\times\nonumber\\
		&\times\left(\begin{tikzpicture}[baseline=-5]
			\node[dot] (x1) at (0,0) {};
			\draw[Phi1] (x1) -- ++(-0.1,0.25);
			\draw[Phi1] (x1) -- ++(0.1,0.25);
			\node[dot] (x2) at (0.3,0) {};
			\draw[Phi1] (x2) -- ++(0.1,0.25);
			\node[dot] (x1') at (0.6,0) {};
			\draw[Phi1] (x1') -- ++(-0.1,0.25);
			\node at (1.1,0) {$\cdots$};
			\node[dot] (x3) at (1.6,0) {};
			\draw[Phi1] (x3) -- ++(-0.1,0.25);
			\node[dot] (x4) at (1.9,0) {};
			\draw[Phi1] (x4) -- ++(-0.1,0.25);
			\draw[Phi1] (x4) -- ++(0.1,0.25);
			\node[dot] (x5) at (2.2,0) {};
			\draw[Phi1] (x5) -- ++(0.1,0.25);
		\end{tikzpicture}-\begin{tikzpicture}[baseline=-5]
			\node[dot] (x1) at (0,0) {};
			\draw[Phi1] (x1) -- ++(-0.1,0.25);
			\draw[Phi1] (x1) -- ++(0.1,0.25);
			\node[dot] (x2) at (0.3,0) {};
			\node[dot] (x1') at (0.6,0) {};
			\draw[Phi1] (x1') -- ++(0.1,0.25);
			\draw[Phi1] (x1') -- ++(-0.1,0.25);
			\node at (1.1,0) {$\cdots$};
			\node[dot] (x3) at (1.6,0) {};
			\draw[Phi1] (x3) -- ++(-0.1,0.25);
			\node[dot] (x4) at (1.9,0) {};
			\draw[Phi1] (x4) -- ++(-0.1,0.25);
			\draw[Phi1] (x4) -- ++(0.1,0.25);
			\node[dot] (x5) at (2.2,0) {};
			\draw[Phi1] (x5) -- ++(0.1,0.25);
		\end{tikzpicture}\right),\label{G_N}
	\end{align}
where $i<j$ for any $\{i,j\}\in\Pi_1\cup\Pi_2$.
	Then by an argument similar to the one used in \eqref{1.6}, we obtain that the squared $L^2(\PP)$ norm of \eqref{G_N} is bounded by a constant times
	\begin{align}
		&\varepsilon|\log\varepsilon|\left(\int\ud h_{2:p} v(h_{2:p})|F_{N,M}(\varepsilon h_2)|\prod_{\substack{\{i,j\}\in\Pi_1\cup\Pi_2\\i,j\neq1}}(1+\log(d(\varepsilon (h_i-h_j))))\prod_{\{1,j\}_{j>1}\in\Pi_1\cup\Pi_2}(1+\log(d(\varepsilon h_j)))\right)^2\nonumber\\
		\lesssim&\varepsilon|\log\varepsilon|\left(\int_{|x|\leq C}|F_{N,M}(\varepsilon x)|^2\ud x\right)
		\left(\int_\Lambda(1+|\log(\varepsilon|x|)|)^{2p-2}\ud x\right)\nonumber\\
		\lesssim&\varepsilon^{-1}|\log\varepsilon|^{2p}\int_{|x|\leq C\varepsilon}|F_{N,M}(x)|^2\ud x.\label{G_N1}
	\end{align}
Here $C>0$ is a constant depending only on the support of $v$. We choose
$\varepsilon>0$ sufficiently small so that $\{|x|\le C\varepsilon\}\subset \Lambda$.
If $CN\varepsilon\ge1$, then by \eqref{F_{N,M}}, the right-hand side of \eqref{G_N1} is bounded by a constant times
	\begin{align*}
		\varepsilon^{-1}|\log\varepsilon|^{2p}N^{-2}\lesssim N^{-1}(1+\log N)^{2p}\,.
	\end{align*}
If instead $CN\varepsilon<1$, by \eqref{bound:GN}, \eqref{G_N1} is bounded by
	\begin{align*}
		\varepsilon^{-1}|\log\varepsilon|^{2p}\int_{|x|\leq C\varepsilon}(1+|\log(N|x|)|)^2\ud x
		=&\varepsilon|\log\varepsilon|^{2p}\int_{|x|\leq C}(1+|\log(N\varepsilon|x|)|)^2\ud x\\ \lesssim&\varepsilon\left(1+|\log CN\varepsilon|+\log N\right)^{2p}\left(1+|\log CN\varepsilon|\right)^2
		\lesssim N^{-1}(1+\log N)^{2p}\,,
	\end{align*}
where in the last step we used that \(a(1+|\log a|)^{2p+2}\lesssim 1\) for \(a=CN\varepsilon<1\).
In summary, we have
\begin{align}
	\sup_{0<\varepsilon\leq (4R)^{-1}}\EE|W_p^{\varepsilon,M}(u)-W_p^{\varepsilon,N}(u)|^2\lesssim  N^{-\frac12}.\label{3.15}
\end{align}

\noindent {\bf Step 3. Proof of \eqref{1.2}.} 
By \eqref{1.1}, it suffices to prove that for any $m\geq1$,
	\begin{align}\label{3.17}
		\sup_{\varepsilon,N}\|e^{-W_p^{\varepsilon,N}}\|_{L^m(\mu_0)}<\infty.
	\end{align}
We have
	\begin{align*}
	\|e^{-W_p^{\varepsilon,N}}\|_{L^m(\mu_0)}^m
	=&\int_0^\infty \mu_0(e^{-m W_p^{\varepsilon,N}}>t)\ud t
	\leq e+\int_e^\infty \mu_0(m W_p^{\varepsilon,N}<-\log t)\ud t \,.
\end{align*}
Let $C$ be the constant in \eqref{lower bound:W_N}, and for $t>e$, choose the unique integer $N_0=N_0(t)$ such that
\[mC(\log N_0)^p\le\log t-1
<mC(\log(N_0+1))^p.\]
If $N\le N_0$, by \eqref{lower bound:W_N},
$$\mu_0(mW_p^{\varepsilon,N}<-\log t)\leq\mu_0(m W_p^{\varepsilon,N}<-mC(\log N)^p-1)=0.$$
If $N>N_0$, then for every integer $r\ge2$,
Chebyshev's inequality and \eqref{3.15} give
\begin{align*}
	\mu_0\bigl(mW_p^{\varepsilon,N}< -\log t\bigr)
	&\le \mu_0\bigl(m|W_p^{\varepsilon,N}-W_p^{\varepsilon,N_0}|>1\bigr)
	\le m^r\|W_p^{\varepsilon,N}-W_p^{\varepsilon,N_0}\|_{L^r(\mu_0)}^r\\
	&\le \bigl(Cm r^p N_0^{-1/4}\bigr)^r\le \left(Cm r^p
	\exp\left[-c(\log t)^{1/p}\right]\right)^r
\end{align*}
for some constant $c,C>0$. In the third inequality we use hyper-contractivity \cite[Theorem 5.10]{Jan97} since $W_p^{\varepsilon,N}$ lives in at most $2p$-th Wiener chaos.
Choose an integer $r\ge2$ comparable to
$(eCm)^{-1/p}\exp(c(\log t)^{1/p}/p)$. Then, after decreasing $c>0$ if necessary, we have
\[\mu_0\bigl(mW_p^{\varepsilon,N}< -\log t\bigr)\lesssim
\exp\!\left[-c\exp\bigl(c(\log t)^{1/p}/p\bigr)\right],
\]
which is integrable on $(e,\infty)$. This proves \eqref{3.17}. Together with the $L^2$ convergence in \eqref{1.1}, uniform integrability yields \eqref{1.2}.

Taking \(m=1\) in \eqref{1.2}, we obtain
\[
\cZ_p^{\varepsilon,N}
=\int e^{-W_p^{\varepsilon,N}}\,\ud\mu_0
\to
\int e^{-V_p}\,\ud\mu_0
=\cZ_p,\quad \varepsilon\to0,N\to\infty,
\]
which proves \eqref{partition function conv}. Moreover, taking \(m=2\) in \eqref{1.2} and using the convergence of the normalizing constants, we get
\[
g_p^{\varepsilon,N}:=
(\cZ_p^{\varepsilon,N})^{-1}e^{-W_p^{\varepsilon,N}}
\to
g_p:=\cZ_p^{-1}e^{-V_p}
\quad\text{in}\;L^2(\mu_0).
\]
Therefore, by Lemma~\ref{lem:Gaussian-moment-maps}, as $\varepsilon\to0$ and $N\to\infty$,
\[
\begin{aligned}
	&\left\|
	\int |u^{\otimes k}\rangle\langle u^{\otimes k}|
	\,\bigl(\ud\mu_p^{\varepsilon,N}(u)-\ud\mu_p(u)\bigr)
	\right\|_{\mathfrak S^2}
	\le
	C_k\|g_p^{\varepsilon,N}-g_p\|_{L^2(\mu_0)}
	\to 0.
\end{aligned}
\]
This proves \eqref{matrices conv}.

\end{proof}

\section{A priori estimates for the quantum Gibbs state}\label{sec4}

We now establish the uniform estimates on the quantum Gibbs state \eqref{def:Gibbs state} that will be used in the variational analysis of Sections~\ref{sec5} and \ref{sec6}. We first prove a logarithmic lower bound for the interaction, which is the Fock-space analogue of the classical lower bound in Lemma~\ref{lem:lower bound of WN}:
the positive \(p\)-body interaction controls a \(p\)-th order local average, and the lower-order terms are absorbed by the same finite-covering argument.
We then derive moment estimates for the quasi-free state, arbitrary polynomial moments of the particle number, and a high-momentum correlation estimate. 

\subsection{The interaction operator on Fock space}\label{sec4.1}

We denote by $h=-\Delta+1$ the one-body operator on $L^2(\Lambda)$. Then the Gaussian state $\Gamma_0$ in \eqref{def:Gaussian state} is $Z_0^{-1}e^{-\lambda\ud\Gamma(h)}$.
We write $\mathbb W=\bigoplus_{n=0}^\infty \mathbb W_n,$
where \(\mathbb W_n\) acts on the \(n\)-th Fock space \(\cH^n\) as
\begin{align}\label{def:quantum-W1}
	\WW_n=&\sum_{l=2}^{p}\lambda^l\sum_{i_{1:l}=1}^n R_l(v^\varepsilon)(x_{i_2}-x_{i_1},\cdots,x_{i_l}-x_{i_1})-(1+\vartheta)\lambda n+E_0\,,
\end{align}
where $R_p(v^\varepsilon):=p^{-1}v^\varepsilon$ and $R_l(v^\varepsilon)$, $2\leq l\leq p-1$, are defined in \eqref{def:R_l}.
With this notation, the interacting Hamiltonian \(\mathbb H_\lambda\) in \eqref{Hamiltonian} can be written as $$\HH_\lambda=\lambda\ud\Gamma(h)+\WW.$$
We first establish a lower bound for \(\WW\). Together with the positivity of \(\ud\Gamma(h)\), this implies that \(\HH_\lambda\) is
bounded from below and hence can be defined as a self-adjoint operator by the Friedrichs extension. 
The proof is similar to that of  Lemma~\ref{lem:lower bound of WN}.
After bounding the effective kernels \(R_l(v^\varepsilon)\) by the marginals of \(v^\varepsilon\) with logarithmic weights, the interaction on each \(n\)-particle sector is reduced to a positive \(p\)-body term minus lower-order terms. The positivity and finite range of \(v\) allow us to control these lower-order terms by the \(p\)-body term through local averages.

\begin{lem}\label{lem 4.1}
	Fix $\eta>0$. There exists a constant $C>0$, such that for every $0<\lambda\leq1$ and $\lambda^\eta\leq\varepsilon\leq 1/(4R)$, one has
	\begin{align}\label{lower bound:W}
		\WW\geq-C|\log\lambda|^p.
	\end{align}
	Here $R$ is the support radius in \eqref{def:R}.
\end{lem}

\begin{proof}
For $2\le l\le p-1$, recall the definition of $f_l^\varepsilon$ in \eqref{def:F_l}.
	Using the logarithmic bound \eqref{C0} for \(G\), we obtain
	\begin{align*}
		|f_l^\varepsilon(x_{2:l})|
		&\lesssim
		\sum_{\substack{\fm\in\fM_{p,l}\\ C_\fm\neq0}}
		\int_{\Lambda^{p-l}}
		v^\varepsilon(x_{2:p})
		\prod_{1\le a<b\le p}
		\bigl(1+|\log d(x_a-x_b)|\bigr)^{m_{ab}}
		\,\ud x_{l+1:p}.
	\end{align*}
	Here and below we set \(x_1=0\). The definition of the graph coefficients in \eqref{C_g2} implies that $C_\fm=0$ whenever the term \(\fg_\fm\) is not in $\fG_p^N$ (In particular, a contributing multi-index cannot contain an edge joining two of the $l$ external vertices; hence $m_{ab}=0$ if $1\le a<b\le l$). 
	Since \(\sum_{a<b}m_{ab}=p-l\), the Young's inequality
	$\prod_{\me=(a,b)} A_\me^{m_\me}\lesssim\sum_{\me:m_\me>0} A_\me^{p-l},$ 
	$A_\me\ge0,$ gives
	\begin{align}
		|f_l^{\varepsilon}(x_{2:l})|\lesssim&\sum_{\substack{1\leq a<b\leq p\\b>l}}\int_{\Lambda^{p-l}}v^\varepsilon(x_{2:p})|\log d(x_a-x_b)|^{p-l}\ud x_{l+1:p}\nonumber\\
		\lesssim&\,\sum_{i=1}^{p-1}\int_{\Lambda^{p-l}}v^\varepsilon(x_{2:p})|\log d(x_i-x_p)|^{p-l}\ud x_{l+1:p}
		=:\fM_lv^\varepsilon(x_{2:l})\,.\label{F:bound}
	\end{align}
Here in the second step we used the invariance of $v^\varepsilon$ under permutation.
For each $1\leq i\leq p-1$, we use \eqref{poisson1} and change the variables $x_j+m^x_j=\varepsilon z_j$ for $l+1\leq j\leq p-1$, and $x_p+m^x_p=\varepsilon (z_p+z_i)$ to get
\begin{align}
	\int_{\Lambda^{l-1}}|f_l^\varepsilon(x)|\ud x\lesssim\int_{|z_p|\leq 2R}|\log d(\varepsilon z_p)|^{p-l}\ud z_p\lesssim|\log\varepsilon|^{p-l}.\label{5.1}
\end{align}
The same rescaling gives, for the scalar terms $l=0,1$,
\begin{align}
	|f_l^{\varepsilon}|\lesssim&\sum_{i=1}^{p-1}\int_{\Lambda^{p-l}}v^\varepsilon(x_{2:p})|\log d(x_i-x_p)|^{p-l}\ud x_{2:p}\lesssim|\log\varepsilon|^{p-l}, \quad l=0,1.\label{f_0}
\end{align}
Recall from \cite[(3.2)]{LNR21} that $N_0=\Tr(\Gamma_0^{(1)})\lesssim\lambda^{-1}|\log\lambda|$. Then, it follows that
\begin{align}\label{R1}
	|\vartheta|\lesssim(|\log\lambda|+|\log\varepsilon|)^{p-1},\quad |E_0|\lesssim(|\log\lambda|+|\log\varepsilon|)^p.
\end{align} 
Next, for \(2\le l<m\le p-1\), since $v$ is symmetric, by \eqref{F:bound} and Young's inequality, we have
\begin{align*}
	|\log\lambda|^{m-l}|\Pi_lf_m^\varepsilon(x_{1:l-1})|\leq|\log\lambda|^{m-l}\int\fM_mv^\varepsilon(x_{1:m-1})\ud x_{l:m-1}\lesssim|\log\lambda|^{p-l}\Pi_lv^\varepsilon(x_{1:l-1})+\fM_lv^\varepsilon(x_{1:l-1}).
\end{align*}
Therefore, $R_l(v^\varepsilon)$ defined in \eqref{def:R_l} has the bound:
\begin{align}\label{R}
	|R_l(v^\varepsilon)|\lesssim&|\log\lambda|^{p-l}\Pi_lv^\varepsilon + |f_l^\varepsilon|+ \sum_{m=l+1}^{p-1}|\log\lambda|^{m-l}|\Pi_lf_m^\varepsilon|\lesssim |\log\lambda|^{p-l}\Pi_lv^\varepsilon+\fM_lv^\varepsilon =: v^\varepsilon_l\,.
\end{align}
Define the operator on $\cH^n$:
\[\fV_{\varepsilon,l}^{(n)}
:=\sum_{i_1,\ldots,i_l=1}^nv_l^\varepsilon(x_{i_2}-x_{i_1},\ldots,x_{i_l}-x_{i_1}),\quad 2\leq l\leq p,\]
where $v_p^\varepsilon=v^\varepsilon$.
By \eqref{R1} and \eqref{R}, there exists $C>0$ such that for every $n\geq1$,
\begin{align}
	\WW_n\geq&\frac{\lambda^p}{p}\fV^{(n)}_{\varepsilon,p}- C\left(\sum_{l=2}^{p-1}\lambda^l\fV^{(n)}_{\varepsilon,l}+(|\log\lambda|+|\log\varepsilon|)^{p-1}(\lambda n+|\log\lambda|+|\log\varepsilon|)\right).\label{4.8}
\end{align}
We claim the following two estimates. First, for every $2\le l\le p-1$ and $C_1>0$, there exists $C_2=C_2(C_1,p,v,\eta)>0$, independent of $n$, $\varepsilon$, and $\lambda$, such that
\begin{align}
	\lambda^p\fV_{\varepsilon,p}^{(n)}
	-C_1\lambda^l\fV_{\varepsilon,l}^{(n)}
	\ge-C_2|\log\lambda|^p.\label{lwb1}
\end{align}
Second, there exists $c=c(p,v)>0$, independent of $n$ and $\varepsilon$, such that
\begin{align}
	\fV_{\varepsilon,p}^{(n)}\geq cn^p.\label{lwb2}
\end{align}
For $p=2$, the lower-order sum in \eqref{4.8} is empty. For $p\ge3$, applying \eqref{lwb1} and using $|\log\varepsilon|\le\eta|\log\lambda|$, we obtain
\[\WW_n\ge\frac{1}{2p}\lambda^p\fV_{\varepsilon,p}^{(n)}
-C(|\log\lambda|^{p-1}\lambda n+|\log\lambda|^p).\]
Then, by \eqref{lwb2} and Young's inequality, we obtain \eqref{lower bound:W}.

Now we prove \eqref{lwb1} and \eqref{lwb2}. Recall from the proof of \eqref{lower bound:W_N} that \(r_0,a_0>0\) were chosen so that \eqref{lower bound v} holds and $2r_0\sqrt{p-1}\leq R$. 
The choice of \(r_0\), together with \(R\varepsilon\le1/4\), gives $r_0\varepsilon\leq 1/8$. Hence for every $y\in\Lambda$, we have $|B(y,r_0\varepsilon)|=\pi r_0^2\varepsilon^2.$
Fix $n\ge1$. In analogy with \eqref{def:local average}, define on $\mathfrak H^n$ the nonnegative multiplication operators
\[B_1^\varepsilon(y):=\varepsilon^{-2}
\sum_{i=1}^n\1_{B(y,r_0\varepsilon)}(x_i),\quad B_2^\varepsilon(y)
:=\varepsilon^{-2}\sum_{i=1}^n\1_{B(y,(R+r_0)\varepsilon)}(x_i).
\]
All inequalities below are inequalities of multiplication operators, or equivalently pointwise inequalities for almost every configuration $(x_1,\ldots,x_n)\in\Lambda^n$. 
Recall the set $A$ from \eqref{def:A}.
By \eqref{3} and \eqref{2}, we have
\begin{align}\label{6}
	\fV_{\varepsilon,p}^{(n)}
	\geq a_0\varepsilon^{-2(p-1)}\sum_{i_{1:p}=1}^n \1_A(x_{i_1},\cdots,x_{i_p})
	\geq (\pi r_0^2)^{-1} a_0\int_\Lambda (B_1^\varepsilon(y))^p\,\ud y .
\end{align}
Since
\[\int_\Lambda B_1^\varepsilon(y)\,\ud y=\varepsilon^{-2}|B(0,r_0\varepsilon)|\,n=\pi r_0^2 n,\]
by H\"older's inequality and \eqref{6}, we obtain
$$n^p\lesssim\int_\Lambda (B_1^\varepsilon(y))^p\,\ud y\lesssim \fV_{\varepsilon,p}^{(n)}.$$
This proves \eqref{lwb2}.

For \(2\le l\le p-1\), by \eqref{R},
\begin{align*}
	&v_l^\varepsilon(x_2-x_1,\cdots,x_l-x_1)
	=\int_{\Lambda^{p-l}}v^\varepsilon(x_2-x_1,\cdots,x_p-x_1)\left(\sum_{i=1}^{p-1}|\log d(x_i-x_p)|^{p-l}+|\log\lambda|^{p-l}\right)\ud x_{l+1:p}.
\end{align*}
We then use the same argument as in the proof of \eqref{marginal-pointwise}.
More precisely, by \eqref{poisson1} and a change of variables $x_i-x_1+m_i^x=\varepsilon y_i$ for $l+1\leq i\leq p$, we have
\begin{align*}
	v_l^\varepsilon(x_2-x_1,\cdots,x_l-x_1)
	=&\varepsilon^{2-2l}\int_{(\RR^2)^{p-l}}\ud y_{l+1:p}\,v(\varepsilon^{-1}(x_2-x_1+m^x_1),\cdots,\varepsilon^{-1}(x_l-x_1+m^x_l),y_{l+1:p})\times\\
	&\left(\sum_{i=1}^l|\log d(x_i-x_1-\varepsilon y_p)|^{p-l}+\sum_{i=l+1}^{p-1}|\log d(\varepsilon(y_i-y_p))|^{p-l}+|\log\lambda|^{p-l}\right).
\end{align*}
Recall the definition of $\widetilde d_l$ in \eqref{def:dl}. 
Then, on the support of \(v\), all the variables \(y_{l+1},\ldots,y_p\) belong to \(B(0,R)\), and
$\widetilde d_l(x_{1:l})\le R\varepsilon$. Since
\(v\in L^\infty\) and \(\varepsilon\ge\lambda^\eta\), we have
\begin{align}
	v_l^\varepsilon(x_2-x_1,\ldots,x_l-x_1)
	\lesssim&\1_{\{\widetilde d_l(x_{1:l})\le R\varepsilon\}}\varepsilon^{2-2l}\int_{|x|\leq 2R}(|\log d(\varepsilon x)|^{p-l}+|\log\lambda|^{p-l})\ud x\nonumber\\
	\lesssim&\1_{\{\widetilde d_l(x_{1:l})\le R\varepsilon\}}\varepsilon^{2-2l}|\log\lambda|^{p-l}.\label{v:bound}
\end{align}
By \eqref{1}, it follows that
\begin{align}
	\fV_{\varepsilon,l}^{(n)}
	&=\sum_{i_{1:l}=1}^n
	v_l^\varepsilon(x_{i_2}-x_{i_1},\ldots,x_{i_l}-x_{i_1}) \lesssim|\log\lambda|^{p-l}
	\int_\Lambda(B_2^\varepsilon(y))^l\,\ud y,\label{5}
\end{align}
The finite-covering argument from \eqref{4} gives
\[B_2^\varepsilon(y)
\le\sum_{m=1}^MB_1^\varepsilon(y+\varepsilon z_m),\quad y\in\Lambda.\]
Since these are commuting multiplication operators, Minkowski's inequality and translation invariance on the torus imply
\[\|B_2^\varepsilon\|_{L^l(\Lambda)}
\le M\|B_1^\varepsilon\|_{L^l(\Lambda)}
\le M\|B_1^\varepsilon\|_{L^p(\Lambda)}.\]
Then, for any $C_1>0$, it follows from \eqref{6} and \eqref{5} that
\[\lambda^p\fV_{\varepsilon,p}^{(n)}
-C_1\lambda^l\fV_{\varepsilon,l}^{(n)}
\ge(\pi r_0^2)^{-1}a_0(\lambda\|B_1^\varepsilon\|_{L^p(\Lambda)})^p - \widetilde C_1|\log\lambda|^{p-l}(\lambda\|B_1^\varepsilon\|_{L^p(\Lambda)})^l\gtrsim-|\log\lambda|^p.
\]
This gives \eqref{lwb1}.
\end{proof}

We now rewrite \(\WW\) in Fourier variables. For \(2\le l\le p\), introduce the associated translation-invariant \(l\)-variable kernel
\begin{align*}
	F_l^\varepsilon(x_{1:l}):=f_l^\varepsilon(x_2-x_1,\ldots,x_l-x_1),
\end{align*}
where $f_p^\varepsilon=\frac1p v^\varepsilon$ and $f_l^\varepsilon,$ $2\leq l\leq p-1$, is defined in \eqref{def:F_l}.
Equivalently, with our Fourier convention,
\[\widehat{F^\varepsilon_l}(k_{1:l})
=\1_{\{k_1+\cdots+k_l=0\}}\widehat f_l^\varepsilon(k_{2:l}).\]
For \(2\le l\le p-1\), \(F_l^\varepsilon\) is symmetric by \eqref{sym:f}; for \(l=p\), this follows from Assumption~\ref{assum:w}.
For any $k\in\ZZ^2$, we define 
\begin{align*}
	\Wick{\ud\Gamma(e_k)}=\ud\Gamma(e_k)- \Tr(\ud\Gamma(e_k)\Gamma_0)\,.
\end{align*}
In particular, $\Wick{\fN}=\fN-N_0$ with $N_0=\Tr(\fN\Gamma_0)$.
By translation invariance, we have $\Tr(\ud\Gamma(e_k)\Gamma_0)=\1_{k=0}N_0.$ 
Then, for every $F\in\cH^l$, the following identity holds on Fock space:
\begin{align}
	\sum_{k_{1:l}}\widehat{F}(k_{1:l})\prod_{j=1}^l\Wick{\ud\Gamma(e_{k_j})}
	=&\sum_{m=0}^{l}\binom lm(-N_0)^{l-m}\sum_{k_{1:m}}\widehat{F}(k_{1:m},0,\cdots,0)\prod_{j=1}^m\ud\Gamma(e_{k_j})\label{4.2}\\
	=&\bigoplus_{n=0}^\infty\sum_{m=0}^{l}\binom lm(-N_0)^{l-m}\sum_{i_{1:m}=1}^n\Pi_{m+1}F(x_{i_{1:m}})\,.\nonumber
\end{align}
Applying \eqref{4.2} to $F=F_l^\varepsilon$, $2\leq l\leq p$, with our choice of the terms $\vartheta$, $E_0$ and $R_l(v^\varepsilon)$ in \eqref{def:R_l} and \eqref{def:nu}, we obtain
\begin{align}\label{def:quantum-W}
	\WW=&\sum_{l=2}^p\lambda^l\sum_{k_{1:l}}\widehat{F^\varepsilon_l}(k_{1:l})
	\prod_{j=1}^{l}\Wick{\ud\Gamma(e_{k_j})} + f_1^\varepsilon\,\lambda\Wick{\fN}+f_0^\varepsilon\,,
\end{align}
where $f_0^\varepsilon$ and $f_1^\varepsilon$ are constants.
In the following, we will use this Fourier expression \eqref{def:quantum-W} for convenience.

\subsection{Moment estimates for the Gaussian quantum state}
Let $\langle\cdot\rangle_0:=\Tr(\cdot\,\Gamma_0)$ denote the expectation against $\Gamma_0$ on the Fock space $\cF$.
We prove a formula for arbitrary moments of second-quantized one-body observables under the Gaussian state. The second-moment case is \cite[Lemma~5.11]{LNR21}.

For simplicity, we introduce the following notations. 
For a finite set \(E\), let \(S_E\) denote the permutation group of \(E\).
In particular, we write \(S_l:=S_{[l]}\) for $l\geq1$.
For \(\sigma\in S_E\), define its set of fixed points by $\Fix(\sigma):=\{i\in E:\sigma(i)=i\}.$
A cycle on a nonempty set \(B\) is a permutation \(c\in S_B\) for which
there exist pairwise distinct elements \(i_1,\dots,i_m\in B\) such that
\(B=\{i_1,\dots,i_m\}\) and
\[
c(i_r)=i_{r+1},\quad 1\le r\le m-1,\qquad c(i_m)=i_1.
\]
In this case, we write \(c=(i_1\,i_2\,\cdots\,i_m)\). In particular, a cycle of length $1$ is just a fixed point.
Every permutation \(\sigma\in S_E\) admits a unique decomposition into
pairwise disjoint cycles. We denote by \(\mathrm{cyc}(\sigma)\) the set of disjoint cycles in this decomposition, including \(1\)-cycles corresponding to fixed points.
For a family of trace class operators \(T=(T_i)_{i\in E}\) and a cycle \(c=(i_1\,\cdots\,i_m)\), define
\[
\Tr_c(T):=\Tr(T_{i_1}\cdots T_{i_m}).
\]
This is well defined, since $\Tr(T_{i_1}\cdots T_{i_m})$ is invariant under cyclic permutations of the factors.

With this notation, the higher-order moment formula takes the following form. Its proof separates the distinct-label contribution from the other terms. The latter are of lower particle order and are controlled by particle-number moments. The former is computed directly, yielding the cycle expansion below. The scalar term subtracted in each \(:\ud\Gamma(A_j):\) removes all one-cycles, so only fixed-point free permutations remain.

\begin{lem}\label{Lem5.1}
	Let $2\leq p\in\NN$, and let \(A_1,\dots,A_p\) be bounded
	self-adjoint operators on \(L^2(\Lambda)\). Set
	$T_j:=A_j\Gamma_0^{(1)},\, 1\le j\le p.$
	Then,
	\begin{align}
		\lambda^p\left\langle\prod_{j=1}^p\Wick{\ud\Gamma(A_j)}\right\rangle_0=\lambda^p
		\sum_{\substack{\sigma\in S_p\\ \Fix(\sigma)=\varnothing}}\prod_{c\in\mathrm{cyc}(\sigma)}\Tr_c(T)+O(\lambda|\log\lambda|^{p-1}\prod_{j=1}^p\|A_j\|)\,,\label{1.73}
	\end{align}
where $\Wick{\ud\Gamma(A_j)} = \ud\Gamma(A_j)-\Tr(A_j\Gamma_0^{(1)})$.
\end{lem}

\begin{proof}
	Since
\begin{align}
	\lambda^p\left\langle\prod_{j=1}^p\Wick{\ud\Gamma(A_j)}\right\rangle_0
	=\lambda^p\sum_{Q\subseteq[p]}\prod_{i\notin Q}\left(-\langle\ud\Gamma(A_i)\rangle_0\right)\left\langle\prod_{j\in Q}\ud\Gamma(A_j)\right\rangle_0\,,\label{1.7}
\end{align}
we calculate $\left\langle\prod_{j\in Q}\ud\Gamma(A_j)\right\rangle_0$ first for any subset \(Q\subset[p]\) with \(q:=|Q|\). 
All products over \(Q\) are taken in the increasing order inherited from \([p]\). 
For \(q=0\), the product is the identity. For \(q\ge1\), by the definition of $\ud\Gamma(A_j)$, we have
\begin{align}
	\prod_{j\in Q}\ud\Gamma(A_j)
	=&\bigoplus_{n\geq q}\sum_{\ii\in[n]^q_{\neq}}\prod_{j\in Q}(A_j)_{i_j}+B_Q \,,\label{1.8}
\end{align}
where $\ii:=(i_j)_{j\in Q}$ and $[n]^q_{\neq}$ denotes the set of ordered $q$-tuples of pairwise distinct labels. Here $B_Q=\oplus_{n\ge0}B_{Q,n}$ and each $B_{Q,n}$ collects all terms on $\cH^n$ containing at most $q-1$ distinct particle labels. Since the $A_i$ are bounded, by \cite[Lemma 5.10]{LNR21},
\begin{align}
	\left|\left\langle B_Q\right\rangle_0\right|\lesssim_q \Tr(\mathcal N^{q-1}\Gamma_0)
	\prod_{j\in Q}\|A_j\|
	\lesssim_q
	(\lambda^{-1}|\log\lambda|)^{q-1}
	\prod_{j\in Q}\|A_j\|.\label{1.72}
\end{align}
Moreover,
\[
|\langle \ud\Gamma(A_i)\rangle_0|
=
|\Tr(A_i\Gamma_0^{(1)})|
\leq
\|A_i\|\,\Tr(\Gamma_0^{(1)})
\lesssim
\lambda^{-1}|\log\lambda|\,\|A_i\|.
\]
Therefore, the total contribution of all the remainders \(B_Q\) in the
centered expansion is bounded by
\begin{align*}
	&\lambda^p
	\sum_{Q\subset[p]}
	\prod_{i\notin Q}|\langle d\Gamma(A_i)\rangle_0|\,
	\left|\left\langle B_Q\right\rangle_0\right|
	\lesssim
	\lambda^p\sum_{q=1}^p
	(\lambda^{-1}|\log\lambda|)^{p-q}
	(\lambda^{-1}|\log\lambda|)^{q-1}
	\prod_{j=1}^p\|A_j\|
	\lesssim
	\lambda|\log\lambda|^{p-1}
	\prod_{j=1}^p\|A_j\|.
\end{align*}
Then, by \eqref{1.7} and \eqref{1.8}, we have
\begin{align}
	\lambda^p\left\langle\prod_{j=1}^p\Wick{\ud\Gamma(A_j)}\right\rangle_0=\lambda^p\sum_{Q\subseteq [p]}\prod_{i\notin Q}\left(-\langle\ud\Gamma(A_i)\rangle_0\right)\,
	{q!}\,\Tr\left(\mathop{\bigotimes^{\text{sym}}}\limits_{j\in Q}A_j\Gamma_0^{(q)}\right)+O\left(\lambda|\log\lambda|^{p-1}\prod_{j=1}^p\|A_j\|\right)\,,\label{1.71}
\end{align}
where $\bigotimes^{\text{sym}}$ denotes the symmetric tensor product.
By \cite[Lemma 5.10]{LNR21}, $\Gamma_0^{(q)}=(\Gamma_0^{(1)})^{\otimes q}$. Then,
\begin{align*}
	q!\,\Tr\left(\mathop{\bigotimes^{\text{sym}}}\limits_{j\in Q}A_j\Gamma_0^{(q)}\right)
	=\sum_{\sigma\in S_Q}\sum_{\substack{i_k=1\\k\in Q}}^\infty\prod_{k\in Q}\langle u_{i_k},T_ku_{i_{\sigma(k)}}\rangle	=&\sum_{\sigma\in S_Q}
	\prod_{c\in\mathrm{cyc}(\sigma)} \Tr_c(T)\,,
\end{align*}
where $\{u_i\}_{i=1}^\infty$ is an orthonormal basis of $L^2(\Lambda)$.

Therefore, the first term in the right-hand side of \eqref{1.71} becomes
\begin{align}
	\lambda^p\sum_{Q\subseteq[p]}
	\prod_{i\notin Q}(-\Tr(T_i))\sum_{\sigma\in S_Q}
	\prod_{c\in\mathrm{cyc}(\sigma)} \Tr_c(T).\label{0.3}
\end{align}
For each pair \((Q,\sigma)\), let \(\overline\sigma\in S_p\) be the extension of \(\sigma\) such that 
$$\overline\sigma(i)=\begin{cases}
	\sigma(i), & i\in Q,\\
	i, & i\notin Q.
\end{cases}$$ 
Then we have
\[\prod_{i\notin Q}(-\Tr(T_i))\prod_{c\in\mathrm{cyc}(\sigma)}\Tr_c(T)
=(-1)^{p-|Q|}\prod_{c\in\mathrm{cyc}(\bar\sigma)}\Tr_c(T).\]
We now regroup the sum in \eqref{0.3} according to $\tau:=\bar\sigma\in S_p$.
Fix \(\tau\in S_p\), and let \(F:=\Fix(\tau)\). We determine all pairs
\((Q,\sigma)\) such that \(\bar\sigma=\tau\). If \(i\notin Q\), then
\(\bar\sigma(i)=i\), hence necessarily \(i\in F\). Therefore
$R:=[p]\setminus Q\subset F.$
Conversely, for every subset \(R\subset F\), set
$Q=[p]\setminus R$. Since we remove only fixed
points of \(\tau\), we have $\tau(Q)=Q$. Hence $\sigma=\tau|_Q\in S_Q$, and its extension $\overline\sigma$ is exactly \(\tau\). Thus the pairs \((Q,\sigma)\) which give the same permutation \(\tau\) are in one-to-one correspondence with the subsets \(R\subset\Fix(\tau)\) and $p-|Q|=|R|$.
Consequently, the total coefficient of the cycle product associated with \(\tau\) is
\[a(\tau)=\sum_{R\subset\Fix(\tau)}(-1)^{|R|}.\]
When \(\Fix(\tau)=\varnothing\), this sum consists only of the empty
subset and equals \(1\). When \(\Fix(\tau)\neq\varnothing\), we have
\[a(\tau)=(1-1)^{|\Fix(\tau)|}=0.\]
Therefore, after regrouping, \eqref{0.3} equals
\[
\lambda^p\sum_{\tau\in S_p}a(\tau)\prod_{c\in\mathrm{cyc}(\tau)}\Tr_c(T)
=\lambda^p\sum_{\substack{\tau\in S_p\\ \Fix(\tau)=\varnothing}}\prod_{c\in\mathrm{cyc}(\tau)}\Tr_c(T).
\]
Substituting this into \eqref{1.71} proves \eqref{1.73}.

\end{proof}

\subsection{A priori estimates for the relative partition function}
We next derive bounds on the relative free energy from the Gibbs variational principle \eqref{variational problem}.

\begin{lem}\label{Lem4.2}
	Let $0<\lambda\leq1$ and $\lambda^\eta\leq\varepsilon\leq(4R)^{-1}$ with $0<\eta<\frac{1}{2(p-1)}$ and $R$ defined in \eqref{def:R}. Then we have
	\begin{align}
		\left|\log\frac{Z_{\lambda}}{Z_0}\right|\lesssim
		|\log\lambda|^p \,.\label{partition function}
	\end{align}
In particular, the Gibbs state $\Gamma_{\lambda}$ satisfies
\begin{equation}\label{entropy}
	\mathcal{H}(\Gamma_{\lambda},\Gamma_0)\lesssim |\log\lambda|^p \,,
\end{equation}
and it holds that
\begin{equation} \label{eqa priori bound on RDM}
	\lambda \left \lVert \sqrt{h} \left(\Gamma_{\lambda}^{(1)} - \Gamma_0^{(1)} \right) \sqrt{h} \right \lVert_{\textnormal{HS}} \lesssim |\log\lambda|^p \,,
\end{equation}
where $\|\cdot\|_{\textnormal{HS}}$ denotes the Hilbert-Schmidt norm.
\end{lem}

\begin{proof}
We take $\Gamma=\Gamma_0$ in \eqref{variational problem}. Then by the definition of $\WW$ in \eqref{def:quantum-W}, we have
\begin{equation}\label{3.7}
	-\log\frac{Z_\lambda}{Z_0}
	\le\Tr(\WW\Gamma_0)=\sum_{l=2}^p\lambda^l\sum_{k_{1:l}}
	\widehat{F^\varepsilon_l}(k_{1:l})
	\left\langle\prod_{j=1}^l\Wick{\ud\Gamma(e_{k_j})}\right\rangle_0 + f_0^\varepsilon.
\end{equation}
In the following, we fix $2\leq l\leq p$.
By complex multilinearity, the identity in \eqref{1.73} extends from bounded self-adjoint operators to bounded multiplication operators. Hence, with $T_j:=e_{k_j}\Gamma_0^{(1)}, \, 1\le j\le l,$
Lemma~\ref{Lem5.1} yields
\begin{align}
	\lambda^l\left\langle\prod_{j=1}^l\Wick{\ud\Gamma(e_{k_j})}\right\rangle_0
	=\lambda^l\sum_{\substack{\sigma\in S_l\\ \Fix(\sigma)=\varnothing}}\prod_{c\in \mathrm{cyc}(\sigma)} \Tr_c(T)+ O(\lambda |\log\lambda|^{\,l-1}).\label{3.81}
\end{align}
Set $n_k:=\frac{1}{e^{\lambda\langle k\rangle^2}-1} \leq\lambda^{-1}\langle k\rangle^{-2}$. Then, for any cycle $c=(c_1\cdots c_{|c|})$ with $|c|\geq2$, we have
\begin{align}
	\lambda^{|c|}\Tr_c(T)
	=&\lambda^{|c|}\1_{\sum_{i\in c}k_i=0}\sum_{q\in\ZZ^2}\prod_{j=1}^{|c|}n_{q+k_{c_1}+\cdots+k_{c_j}}
	\leq\1_{\sum_{i\in c}k_i=0}\sum_{q\in\ZZ^2}\prod_{j=1}^{|c|}\widehat G(q+k_{c_1}+\cdots+k_{c_j}).\label{3.8}
\end{align}
For every $\fm\in\fM_{p,l}$ such that $C_\fm\neq0$, let $F_{l,\fm}^\varepsilon(x_{1:l})$ be defined in \eqref{def:Fl,m} for $l<p$ and $F_{p,\fm}^\varepsilon(x_{1:p})=v^\varepsilon(x_2-x_1,\cdots,x_p-x_1)$. Since $\fF(F_{l,\fm}^\varepsilon)\geq0$ by \eqref{f:positive Fourier}, we obtain
\begin{align*}
	&\lambda^l\sum_{k_{1:l}}
	\widehat{F_{l,\mathbf m}^\varepsilon}(k_{1:l})
	\prod_{c\in\operatorname{cyc}(\sigma)}\Tr_c(T)\le
	\sum_{\substack{k_{1:l}\\\sum_{i\in c}k_i=0,\ \forall c\in\operatorname{cyc}(\sigma)}}
	\widehat{F_{l,\mathbf m}^\varepsilon}(k_{1:l})
	\prod_{c\in\operatorname{cyc}(\sigma)}
	\left(\sum_{q\in\mathbb Z^2}
	\prod_{r=1}^{|c|}\widehat G(q+k_{c_1}+\cdots+k_{c_r})\right).
\end{align*}
Since every cycle has length at least two and the sum of the cycle lengths is $l$, by the Parseval equality and Young's inequality, the last expression equals
\[\int_{\Lambda^l}F_{l,\mathbf m}^\varepsilon(x_{1:l})
\prod_{c\in\operatorname{cyc}(\sigma)}
\left(\prod_{r=2}^{|c|}G(x_{c_r}-x_{c_{r-1}})\right)
G(x_{c_{|c|}}-x_{c_1})\ud x_{1:l}\lesssim\sum_{1\leq i<j\leq l}\int|F_{l,\mathbf m}^\varepsilon(x_{1:l})||G(x_i-x_j)|^l\ud x_{1:l}.\]
Therefore, by \eqref{3.81}, we have
\begin{align}
	\sum_{k_{1:l}}\fF(F_{l,\fm}^\varepsilon)(k_{1:l})\lambda^l\left\langle\prod_{j=1}^l\Wick{\ud\Gamma(e_{k_j})}\right\rangle_0
	\lesssim&\sum_{1\leq i<j\leq l}\int|F_{l,\mathbf m}^\varepsilon(x_{1:l})||G(x_i-x_j)|^l\ud x_{1:l} + C\lambda|\log\lambda|^{l-1}f_{l,\fm}^\varepsilon(0).\label{4.29}
\end{align}
Here we used
$\sum_{k_{1:l}}\widehat{F_{l,\mathbf m}^\varepsilon}(k_{1:l})
=F_{l,\mathbf m}^\varepsilon(0,\ldots,0)
=f_{l,\mathbf m}^\varepsilon(0)$ for the error term.
By \eqref{F:bound} and \eqref{v:bound}, for $l<p$ and $C_\fm\neq0$, we have
\begin{align}\label{F:bound1}
	|f_{l,\fm}^\varepsilon(0)|\lesssim v_l^\varepsilon(0)\lesssim\varepsilon^{2-2l}|\log\lambda|^{p-l},
\end{align}
and by \eqref{poisson1}, for $\varepsilon\leq(4R)^{-1}$,
$$f_{p,\fm}^\varepsilon(0)=v^\varepsilon(0)\lesssim\varepsilon^{2-2p}.$$
Moreover, by \eqref{con:sym} and Young's inequality, for every $1\leq i<j\leq l$,
$$\int_{\Lambda^l} |F_{l,\fm}^\varepsilon(x_{1:l})||G(x_i-x_j)|^l\ud x_{1:l}\lesssim\int_{\Lambda^{p-1}} v^\varepsilon(x_{1:p-1})(|G(x_1-x_2)|^p+|G(x_1)|^p)\ud x_{1:p-1}\lesssim|\log\varepsilon|^p.$$
Since $\varepsilon\geq\lambda^\eta$ with $0<\eta<\frac{1}{2(p-1)}$ and $|f_0^\varepsilon|\lesssim|\log\varepsilon|^p$ by \eqref{f_0}, substituting these upper bounds into \eqref{3.7}, we obtain
\begin{align*}
	-\log\frac{Z_{\lambda}}{Z_0}\leq \Tr(\WW\Gamma_0)\lesssim|\log\varepsilon|^p+\lambda|\log\lambda|^{p-1}\varepsilon^{2-2p}\lesssim|\log\lambda|^p \,.
\end{align*}
On the other hand, by the lower bound \eqref{lower bound:W}, 
\begin{align*}
	-\log\frac{Z_{\lambda}}{Z_0}= \fH(\Gamma_{\lambda},\Gamma_0)+\Tr(\WW\Gamma_{\lambda}) \geq \fH(\Gamma_{\lambda},\Gamma_0)-C|\log\lambda|^p \,.
\end{align*}
This gives both \eqref{partition function} and
\eqref{entropy}. Finally, \eqref{eqa priori bound on RDM} follows from \cite[Theorem 6.1]{LNR21}.

\end{proof}

\subsection{Moments of particle number and correlation estimate for high momenta}\label{sec4.4}
Recall that for each $k\geq1$, $\Gamma_0^{(k)}=\left(\frac{1}{e^{\lambda h}-1}\right)^{\otimes k}$, which implies that $\lambda^k\Tr(\fN^k\Gamma_0)\lesssim_k|\log\lambda|^k$. Then by \eqref{eqa priori bound on RDM} and the Cauchy-Schwarz inequality, if $\varepsilon\geq\lambda^\eta$ with $0<\eta<\frac{1}{2(p-1)}$, we have
\begin{align*}
	\lambda\Tr(\fN\Gamma_{\lambda})\leq \lambda\Tr(\fN\Gamma_0) + \lambda \left \lVert \sqrt{h} \left(\Gamma_{\lambda}^{(1)}-\Gamma_0^{(1)} \right) \sqrt{h} \right \lVert_{\textnormal{HS}}(\Tr(h^{-2}))^\frac12\lesssim |\log\lambda|^p \,.
\end{align*}
Now, we aim to extend this inequality to the $k$-th moment $\lambda^k\Tr(\fN^k\Gamma_\lambda)$ for any $k\geq1$.
Due to the lower bound of $\WW$ in \eqref{lower bound:W}, the estimates used in the proof of \cite[(5.51)]{LNR21} are no longer sufficient. Instead, we employ a higher-order correlation inequality from \cite{DNN25} to achieve this extension.

\begin{lem}\label{Lemma4.3}
	Under the assumptions of Lemma~\ref{Lem4.2}, for every $k\ge1$, it holds that
	\begin{align}
		\lambda^k\Tr(\fN^k\Gamma_{\lambda})\lesssim_{k}|\log\lambda|^{pk} \,.\label{0.2}
	\end{align}
\end{lem}

\begin{proof}
	For $-1\leq t\leq1$, define 
	$$\Gamma_{\lambda,t} :=Z_{\lambda,t}^{-1}e^{-\lambda\ud\Gamma(h_t)-\WW},\quad \Gamma_{0,t} :=Z_{0,t}^{-1}e^{-\lambda\ud\Gamma(h_t)},$$ 
	where $h_t = h-t|\log\lambda|^{-p}$.
	Let $\lambda$ be sufficiently small such that $|\log\lambda|^{-p}\leq\frac12$. Since $h\geq1$, for every $|t|\leq1$, we have $\frac12 h\leq h_t\leq\frac32 h$.
	
	\noindent {\bf Step 1.} 
	First, we mimic the proof of \eqref{eqa priori bound on RDM} to show that there exists $C>0$ such that for $|t|\leq1$,
	\begin{equation} \label{RDM1}
		\lambda \left \lVert \sqrt{h} \left(\Gamma_{\lambda,t}^{(1)} - \Gamma_{0,t}^{(1)} \right) \sqrt{h} \right \lVert_{\textnormal{HS}} \leq C |\log\lambda|^p \,.
	\end{equation}
	As in \eqref{3.7}, by the variational formula \eqref{variational problem} for $\Gamma_{\lambda,t}$, we have
	\begin{align}
		-\log\frac{Z_{\lambda,t}}{Z_{0,t}}\leq&\sum_{l=2}^p\sum_{k_{1:l}}\widehat{F_l^\varepsilon}(k_{1:l})\lambda^l\left\langle\prod_{j=1}^l\Wick{\ud\Gamma(e_{k_j})}\right\rangle_{0,t} + f_1^\varepsilon\lambda (\Tr(\Gamma_{0,t}^{(1)})-\Tr(\Gamma_0^{(1)})) + f_0^\varepsilon\,.\label{3.9}
	\end{align}
	Since $\|h_t-h\|\leq\frac12$ and $h_t\geq\frac12 h$, by \cite[Lemma 5.4]{DNN25}, we have
	\begin{align}
		\lambda |\Tr(\Gamma_{0,t}^{(1)})-\Tr(\Gamma_0^{(1)})|
		\lesssim\|h_t-h\|(\Tr(h^{-2}) + \Tr(h_t^{-2}))^{\frac12}
		\lesssim(\Tr(h^{-2}))^\frac12\lesssim1\,.\label{4.32}
	\end{align} 
	We estimate the first term in \eqref{3.9} as in Lemma~\ref{Lem4.2}, replacing
	$\Gamma_0$ with $\Gamma_{0,t}$. The only additional point is that the operators in \eqref{def:quantum-W} are centered with respect to $\Gamma_0$ rather than $\Gamma_{0,t}$.
	Recall that \(\Wick{\ud\Gamma(e_k)}\) is centered with respect to \(\Gamma_0\). Hence
	\[\Wick{\ud\Gamma(e_k)}=\Wick{\ud\Gamma(e_k)}_{0,t}
	+\delta_{k,t},\quad \delta_{k,t}:=\Tr\bigl(e_k(\Gamma_{0,t}^{(1)}-\Gamma_0^{(1)})\bigr),\]
	where \(\Wick{\ud\Gamma(e_k)}_{0,t}\) denotes the centering with respect to \(\Gamma_{0,t}\). By translation invariance and \eqref{4.32}, we have
	\begin{align}\label{5.29}
		\lambda|\delta_{k,t}|=\1_{k=0}\lambda |\Tr(\Gamma_{0,t}^{(1)})-\Tr(\Gamma_0^{(1)})|\lesssim\1_{k=0}.
	\end{align}
	For every \(2\le l\le p\), we decompose the contribution of the \(l\)-th term in \eqref{3.9} as
	\begin{align*}
		&\lambda^l\sum_{k_{1:l}}
		\widehat{F^\varepsilon_l}(k_{1:l})
		\left\langle\prod_{j=1}^l\Wick{\ud\Gamma(e_{k_j})}
		\right\rangle_{0,t}
		=\sum_{S\subseteq[l]}\lambda^l\sum_{k_{1:l}}
		\widehat{F^\varepsilon_l}(k_{1:l})
		\prod_{j\notin S}\delta_{k_j,t}
		\left\langle\prod_{i\in S}\Wick{\ud\Gamma(e_{k_i})}_{0,t}\right\rangle_{0,t}.
	\end{align*}
	For $l<p$, let $F_{l,\fm}^\varepsilon$ be defined in \eqref{def:Fl,m} for any $\fm\in\fM_{p,l}$ and $F_{p,\fm}^\varepsilon(x_{1:p}):=v^\varepsilon(x_2-x_1,\cdots,x_p-x_1)$.
	Since $F_l^\varepsilon$ is symmetric and $\fF(F_{l,\fm}^\varepsilon)\geq0$ by \eqref{sym:f} and \eqref{f:positive Fourier}, we use \eqref{5.29} to derive that
	\begin{align}
		&\lambda^l\sum_{k_{1:l}}
		\widehat{F^\varepsilon_l}(k_{1:l})
		\left\langle\prod_{j=1}^l\Wick{\ud\Gamma(e_{k_j})}
		\right\rangle_{0,t}\nonumber\\
		\lesssim&
		\sum_{r=2}^l\lambda^r\sum_{\substack{\fm\in\fM_{p,l}\\C_\fm\neq0}}\sum_{k_{1:r}}\fF(F_{l,\fm}^\varepsilon)(k_{1:r},0,\cdots,0)\left|\left\langle\prod_{i=1}^r\Wick{\ud\Gamma(e_{k_i})}_{0,t}\right\rangle_{0,t}\right|+\sum_{\substack{\fm\in\fM_{p,l}\\C_\fm\neq0}}\fF(F_{l,\fm}^\varepsilon)(0)\,.\label{4.37}
	\end{align}
	For the last term, by \eqref{5.1} (which also holds for $ f_{l,\fm}^\varepsilon$ with $C_\fm\neq0$) and $\widehat{v^\varepsilon}(0)=1$, we have
	\begin{align}\label{4.38}
		\sum_{\substack{\fm\in\fM_{p,l}\\C_\fm\neq0}}\fF(F_{l,\fm}^\varepsilon)(0)=\sum_{\substack{\fm\in\fM_{p,l}\\C_\fm\neq0}}\int f_{l,\fm}^\varepsilon(x)\ud x\lesssim |\log\varepsilon|^{p-l}.
	\end{align}
	Since $h_t\simeq h$ uniformly in $t$, for every $2\leq r\leq l$, the proof of Lemma~\ref{Lem5.1} applies with $\Gamma_0$ replaced by $\Gamma_{0,t}$ and gives
	\begin{align*}
		\lambda^r\left\langle\prod_{j=1}^r\Wick{\ud\Gamma(e_{k_j})}_{0,t}\right\rangle_{0,t}
		=\lambda^r\sum_{\substack{\sigma\in S_r\\ \Fix(\sigma)=\varnothing}}\prod_{c\in \mathrm{cyc}(\sigma)} \Tr_c(T^t)+ O(\lambda |\log\lambda|^{r-1}),
	\end{align*}
	where $T^t_j:=e_{k_j}\Gamma_{0,t}^{(1)},$ $1\le j\le r.$
	Then, the estimate is the same as in the first term of \eqref{3.7}.
	The only modification is that \eqref{3.8} becomes 
	\begin{align*}
		\lambda^{|c|}\Tr_c(T^t)
		=&\lambda^{|c|}\1_{\sum_{i\in c}k_i=0}\sum_{k\in\ZZ^2}\prod_{j=1}^{|c|}n_{k+k_{c_1}+\cdots+k_{c_j},t}
		\lesssim\1_{\sum_{i\in c}k_i=0}\sum_{k\in\ZZ^2}\prod_{j=1}^{|c|}\widehat G(k+k_{c_1}+\cdots+k_{c_j}),
	\end{align*}
	where
	$$n_{k,t}=(e^{\lambda(\langle k\rangle^2-t|\log\lambda|^{-p})}-1)^{-1}\lesssim\lambda^{-1}\langle k\rangle^{-2}$$
	uniformly for \(|t|\le1\) since $\langle k\rangle^2-t|\log\lambda|^{-p}\geq\langle k\rangle^2-1/2\geq1/2\langle k\rangle^2$.
Therefore, similarly to \eqref{4.29}, for $2\leq r\leq l$ and each $\fm\in\fM_{p,l}$ such that $C_\fm\neq0$, it follows that
	\begin{align*}
		&\lambda^r\sum_{k_{1:r}}\fF(F_{l,\fm}^\varepsilon)(k_{1:r},0,\cdots,0)\left|\left\langle\prod_{i=1}^r\Wick{\ud\Gamma(e_{k_i})}_{0,t}\right\rangle_{0,t}\right|\\
		\lesssim&\sum_{1\leq i<j\leq r}\int|F_{l,\mathbf m}^\varepsilon(x_{1:l})||G(x_i-x_j)|^l\ud x_{1:l}+\lambda|\log\lambda|^{r-1}f_{l,\fm}^\varepsilon(0)\lesssim|\log\lambda|^p.
	\end{align*}
	Since $\varepsilon\geq\lambda^\eta$, combining with \eqref{4.38}, \eqref{4.37} is bounded by a constant times $|\log\lambda|^p$.
	Substituting into \eqref{3.9}, by \eqref{4.32} and \eqref{f_0}, we have
	\[
	-\log\frac{Z_{\lambda,t}}{Z_{0,t}}
	\lesssim |\log\lambda|^p.\]
	Combining this with the lower bound \eqref{lower bound:W}, we also get $\mathcal H(\Gamma_{\lambda,t},\Gamma_{0,t})\lesssim|\log\lambda|^p$, which implies 
	\begin{align*}
		\lambda \left \lVert \sqrt{h_t} \left(\Gamma_{\lambda,t}^{(1)} - \Gamma_{0,t}^{(1)} \right) \sqrt{h_t} \right \lVert_{\textnormal{HS}} \lesssim |\log\lambda|^p
	\end{align*}
	by \cite[Theorem 6.1]{LNR21}. Since $h_t\geq\frac12 h$, the operator $h^\frac12h_t^{-\frac12}$ is bounded, and hence we can replace $\sqrt{h_t}$ by $\sqrt{h}$ and obtain \eqref{RDM1}.
	
	\noindent {\bf Step 2.} 
	Set $\BB :=|\log\lambda|^{-p}\lambda\fN$. Then for all $|t|\leq1$, by the Cauchy-Schwarz inequality,
\begin{align*}
	|\Tr(\BB\Gamma_{\lambda,t})| =& |\log\lambda|^{-p}\lambda|\Tr(\Gamma_{\lambda,t}^{(1)})|
	\leq|\log\lambda|^{-p}\,\lambda\left \lVert \sqrt{h} \left(\Gamma_{\lambda,t}^{(1)} - \Gamma_{0,t}^{(1)} \right) \sqrt{h} \right \lVert_{\textnormal{HS}}(\Tr(h^{-2}))^\frac12+|\log\lambda|^{-p}\lambda\Tr(\Gamma_{0,t}^{(1)}),
\end{align*}
Since $h_t\geq\frac12 h$, we have 
$$\lambda\Tr(\Gamma_{0,t}^{(1)})=\Tr\left(\frac{\lambda}{e^{\lambda h_t}-1}\right)\leq\Tr\left(\frac{\lambda}{e^{1/2\lambda h}-1}\right)\lesssim|\log\lambda|.$$
Therefore, combining with \eqref{RDM1}, it holds that $$a:=\sup_{|t|\leq 1}|\Tr(\BB\Gamma_{\lambda,t})|\lesssim1.$$
Since $\BB\geq0$ and $[\BB,\HH_{\lambda}]=0$, \cite[Theorem 3]{DNN25} gives, for every $k\geq1$,
$$\Tr(\BB^k\Gamma_{\lambda})\lesssim_k e^{2a}\lesssim1,$$
which proves \eqref{0.2}.

\end{proof}

We next use these particle-number moments to prove the high-momentum correlation estimate needed for the finite-dimensional reduction.

\begin{theorem}
	Under the assumptions of Lemma~\ref{Lem4.2}, 
	let $N\in\NN$ and $\lambda^{-6p\eta}\leq N\leq\lambda^{-\frac18}$. For $k\in\ZZ^2$, we denote $e_k=e_k^-+e_k^+$ with $e_k^- = P_Ne_kP_N$.
	Then
	\begin{align}
		\lambda^2\langle|\Wick{\ud\Gamma(e_k^+)}|^2\rangle_{\lambda}\lesssim N^{-\frac12}\varepsilon^{2-2p}|\log\lambda|^p +(|k|^2+|k|N)\lambda^2|\log\lambda|^p\,,\label{1.4}
	\end{align}
where $\langle\cdot\rangle_{\lambda}$ denotes the expectation against $\Gamma_{\lambda}$ in the Fock space $\cF$.
\end{theorem}
\begin{proof}
	By linearity, it suffices to prove \eqref{1.4} with $e_k^+$ replaced by $f_k^+$, where $f_k(x)=\cos(2\pi k\cdot x)$ or $\sin(2\pi k\cdot x)$.
	Set $\widetilde\BB = \frac14\lambda\Wick{\ud\Gamma(f_k^+)}$ and 
	$$\widetilde\Gamma_{\lambda,t} := \widetilde Z_{\lambda,t}^{-1}\, e^{-\frac14\lambda t\langle\ud\Gamma(f_k^+)\rangle_0}\,e^{-\lambda\ud\Gamma(\widetilde h_t)-\WW},\quad \widetilde\Gamma_{0,t} := \widetilde Z_{0,t}^{-1}\,e^{-\frac14\lambda t\langle\ud\Gamma(f_k^+)\rangle_0}\,e^{-\lambda\ud\Gamma(\widetilde h_t)},$$ 
	where $\widetilde h_t = h-\frac14tf_k^+$. 
	Since $h\geq1$ and $\|f_k^+(x)\|\leq2$, for every $|t|\leq1$, we have $\frac12 h\leq \widetilde h_t\leq\frac32 h$.
	
	First, we estimate $\Tr(\widetilde \BB\widetilde\Gamma_{\lambda,t})$ for every $t\in[-1,1]$ as in the proof of \eqref{0.2}. We use the variational formula for $\widetilde\Gamma_{\lambda,t}$, \eqref{f_0} and \eqref{4.32} with $\Gamma_{0,t}^{(1)}$ and $h_t$ replaced by $\widetilde\Gamma_{0,t}^{(1)}$ and $\widetilde h_t$ respectively to get
	\begin{align}\label{4.39}
		-\log\frac{\widetilde Z_{\lambda,t}}{\widetilde Z_{0,t}}\lesssim&|\log\varepsilon|^p+\sum_{l=2}^{p}\sum_{k_{1:l}}\widehat{F_l^\varepsilon}(k_{1:l})\,\lambda^l\,
		\Tr\left(\prod_{j=1}^{l}\Wick{\ud\Gamma(e_{k_j})}\widetilde\Gamma_{0,t}\right)\,.
	\end{align} 
Since $\widetilde h_t\geq\frac12 h$, we have 
	$$\lambda\Tr(\widetilde\Gamma_{0,t}^{(1)})=\Tr\left(\frac{\lambda}{e^{\lambda \widetilde h_t}-1}\right)\leq\Tr\left(\frac{\lambda}{e^{1/2\lambda h}-1}\right)\lesssim|\log\lambda|.$$
	Thus, for any $2\leq l\leq p$, and $k_{1:l}\in\ZZ^{2l}$, it follows that
	\begin{align*}
		\lambda^l\left|\Tr\left(\prod_{j=1}^{l}\Wick{\ud\Gamma(e_{k_j})}\widetilde\Gamma_{0,t}\right)\right|
		\lesssim&\lambda^l\Tr(\fN^l\widetilde\Gamma_{0,t})+|\log\lambda|^l
		\lesssim(\lambda\Tr(\widetilde\Gamma_{0,t}^{(1)}))^l+|\log\lambda|^l\lesssim|\log\lambda|^l \,,
	\end{align*}
	where in the second inequality, we use $\widetilde\Gamma_{0,t}^{(l)}=(\widetilde\Gamma_{0,t}^{(1)})^{\otimes l}$, and use \eqref{1.8}, \eqref{1.72} for $A_1=\cdots=A_l=\1_\cH$. 
	Substituting into \eqref{4.39}, by \eqref{f:positive Fourier} and \eqref{F:bound1}, it follows that
	\begin{align*}
		-\log\frac{\widetilde Z_{\lambda,t}}{\widetilde Z_{0,t}}\lesssim|\log\varepsilon|^p+\sum_{l=2}^p\sum_{\substack{\fm\in\fM_{p,l}\\C_\fm\neq0}}|\log\lambda|^lf_{l,\fm}^{\varepsilon}(0)
		\lesssim\varepsilon^{2-2p}|\log\lambda|^p \,,
	\end{align*}
	where $f_{p,\fm}^\varepsilon=v^\varepsilon$.
	Then, by \eqref{lower bound:W}, we have $\fH(\widetilde\Gamma_{\lambda,t},\widetilde\Gamma_{0,t})\lesssim \varepsilon^{2-2p}|\log\lambda|^p$. 
By \cite[Theorem 6.1]{LNR21} and $h_t\geq\frac12 h$, this implies that
	\begin{align}
		\lambda \left \lVert \sqrt{h} \left(\widetilde\Gamma_{\lambda,t}^{(1)} - \widetilde\Gamma_{0,t}^{(1)} \right) \sqrt{h} \right \lVert_{\textnormal{HS}}
		\leq
		\lambda \left \lVert \sqrt{\widetilde h_t} \left(\widetilde\Gamma_{\lambda,t}^{(1)} - \widetilde\Gamma_{0,t}^{(1)}\right) \sqrt{\widetilde h_t} \right \lVert_{\textnormal{HS}}\|h^\frac12(\widetilde h_t)^{-\frac12}\|^2 \lesssim \varepsilon^{2-2p}|\log\lambda|^p,\label{4.41}
	\end{align}
where $\|\cdot\|$ is the operator norm on $\cH$.
Therefore, arguing as in \cite[(7.32)]{NZZ25},
by the Cauchy-Schwarz inequality and \cite[Lemma 6.3]{LNR21}, we obtain
\begin{align}
	|\Tr(\widetilde\BB\widetilde\Gamma_{\lambda,t})|=\frac14\lambda|\Tr(f_k^+\widetilde\Gamma_{\lambda,t}^{(1)})-\Tr(f_k^+\Gamma_0^{(1)})|\lesssim
	\lambda \left \lVert \sqrt{h} \left(\widetilde\Gamma_{\lambda,t}^{(1)} - \widetilde\Gamma_{0,t}^{(1)} \right) \sqrt{h} \right \lVert_{\textnormal{HS}}\|h^{-\frac12}f_k^+h^{-\frac12}\|_{\textnormal{HS}}+\Tr(f_k^+h^{-1}f_k^+h^{-1}).\label{4.40}
\end{align}
Set $Q_N=\1-P_N$. Then $[Q_N,h]=0$, $0\leq Q_N\leq \1_{h\geq1/4(2\pi)^2N^2+1}$, and 
$$\|Q_Nh^{-1}\|^2_{\textnormal{HS}}\lesssim\sum_{|k|>N}\frac{1}{\langle k\rangle^4}\lesssim N^{-2}.$$
Since $f_k^+=P_Nf_kQ_N+Q_Nf_kP_N+Q_Nf_kQ_N$, by H\"older's inequality, we have
$$|\Tr(f_k^+h^{-1}f_k^+h^{-1})|\leq\|h^{-\frac12}f_k^+h^{-\frac12}\|_{\textnormal{HS}}^2\lesssim \|Q_Nh^{-1}\|_{\textnormal{HS}}\|h^{-1}\|_{\textnormal{HS}}\lesssim N^{-1}.$$
Therefore, by \eqref{4.41} and \eqref{4.40}, it follows that
$$a:=\sup_{|t|\leq 1}|\Tr(\widetilde\BB\widetilde\Gamma_{\lambda,t})|\lesssim N^{-\frac12}\varepsilon^{2-2p}|\log\lambda|^p\leq1.$$
Then by \cite[Theorem 2]{DNN25}, we have
	\begin{align}
		\Tr(\widetilde\BB^2\Gamma_\lambda)\leq ae^a+\frac14\Tr([[\widetilde\BB,\HH_{\lambda}],\widetilde\BB]\Gamma_{\lambda})
		\lesssim N^{-\frac12}\varepsilon^{2-2p}|\log\lambda|^p+\Tr([[\widetilde\BB,\HH_{\lambda}],\widetilde\BB]\Gamma_{\lambda})\,,\label{correlation ineq}
	\end{align}

Next, we estimate the second term in the right-hand side of \eqref{correlation ineq}. By linearity, we have
$$[[\widetilde\BB,\HH_{\lambda}],\widetilde\BB]=\frac{\lambda^3}{16}\ud\Gamma([[f_k^+,h],f_k^+])+\frac{\lambda^2}{16}[[\ud\Gamma(f_k^+),\WW],\ud\Gamma(f_k^+)].$$
Then by \cite[(7.34)]{NZZ25} and \eqref{0.2}, the expectation of the first part is bounded by
\begin{align}\label{4.35}
	|\lambda^3\Tr(\ud\Gamma([[f_k^+,h],f_k^+])\Gamma_{\lambda})|\lesssim(|k|^2+|k|N)\lambda^3\Tr(\Gamma_{\lambda}^{(1)})\lesssim(|k|^2+|k|N)\lambda^2|\log\lambda|^p\,.
\end{align}
For the second part, we have
\begin{align}\label{4.36}
	\lambda^2[[\ud\Gamma(f_k^+),\WW],\ud\Gamma(f_k^+)]=\sum_{l=2}^p\lambda^{2+l}\sum_{k_{1:l}}\widehat{F_l^\varepsilon}(k_{1:l})
	\left[\left[\ud\Gamma(f_k^+),\prod_{j=1}^l\Wick{\ud\Gamma(e_{k_j})}\right],\ud\Gamma(f_k^+)\right]\,.
\end{align}
For $n\geq1$ and every one-body operator $T$ on $\cH$, we define the corresponding operator on $\cH^n$:
\begin{align}\label{def:Gamma_n}
	\ud\Gamma_n(T):=\sum_{i=1}^nT_i.
\end{align}
Set $A_n=\ud\Gamma_n(f_k^+)$ and $B_{j,n}=\ud\Gamma_n(e_{k_j})-\1_{k_j=0}N_0$, $1\leq j\leq p$. 
Since \(\|f_k^+\|\lesssim1\) and \(\|e_{k_j}\|=1\), on \(\cH^n\) one has $\|B_{j,n}\|\leq n+N_0$ and
\begin{align}\label{eq:sector-commutator-bounds}
	\|[A_n,B_{j,n}]\|
	=
	\|\ud\Gamma_n([f_k^+,e_{k_j}])\|
	\lesssim n,\quad 
	\|[[A_n,B_{j,n}],A_n]\|
	=
	\|\ud\Gamma_n([[f_k^+,e_{k_j}],f_k^+])\|
	\lesssim n.
\end{align}
Here the constants are independent of \(k,k_j,n,\lambda,\varepsilon\).
For $2\leq l\leq p$, by the Leibniz rule, the double commutator
\begin{align}\label{change}
	\left[[A_n,B_{1,n}\cdots B_{l,n}],A_n\right]
\end{align}
is a finite sum, depending only on \(l\), of ordered products of the following
two types:
\[
B_{1,n}\cdots B_{i-1,n}\,[[A_n,B_{i,n}],A_n]\,B_{i+1,n}\cdots B_{l,n},
\]
and
\[
B_{1,n}\cdots [A_n,B_{i,n}]\cdots [B_{j,n},A_n]\cdots B_{l,n},
\quad i\neq j.
\]
Hence \eqref{eq:sector-commutator-bounds} implies that, on \(\cH^n\),
\begin{align}\label{eq:double-commutator-sector-bound}
	\left\|
	\left[[A_n,B_{1,n}\cdots B_{l,n}],A_n\right]
	\right\|
	\lesssim_l
	n(n+N_0)^{l-1}+n^2(n+N_0)^{l-2}
	\lesssim_l
	\sum_{r=1}^l n^rN_0^{\,l-r}.
\end{align}
Since $$\left[\left[\ud\Gamma(f_k^+),\Wick{\ud\Gamma(e_{k_1})}\cdots \Wick{\ud\Gamma(e_{k_l})}\right],\ud\Gamma(f_k^+)\right]=\bigoplus_{n=1}^\infty\left[[A_n,B_{1,n}\cdots B_{l,n}],A_n\right],$$
we have
\begin{align}\label{eq:double-commutator-expectation-bound}
	\Big|
	\Tr\Big(\left[\left[\ud\Gamma(f_k^+),\Wick{\ud\Gamma(e_{k_1})}\cdots \Wick{\ud\Gamma(e_{k_l})}\right],\ud\Gamma(f_k^+)\right]\Gamma_\lambda
	\Big)
	\Big|
	\lesssim_l
	\sum_{r=1}^l
	\Tr(\fN^r\Gamma_\lambda)N_0^{\,l-r}.
\end{align}
We apply this to the interaction \(\WW\).  Using the decomposition
$F_l^\varepsilon=\sum_{\mathbf m\in\fM_{p,l}} C_{\mathbf m}
F_{l,\mathbf m}^\varepsilon$ and \eqref{f:positive Fourier}, we obtain from \eqref{eq:double-commutator-expectation-bound}
\begin{align*}
	&\left|
	\lambda^2
	\Tr\left(
	[[\ud\Gamma(f_k^+),\WW],\ud\Gamma(f_k^+)]\Gamma_\lambda
	\right)
	\right|\lesssim
	\lambda^2
	\sum_{l=2}^p
	\sum_{\substack{\mathbf m\in\fM_{p,l}\\ C_{\mathbf m}\neq0}}
	f_{l,\mathbf m}^\varepsilon(0)
	\sum_{r=1}^l
	\lambda^r\Tr(\fN^r\Gamma_\lambda)(\lambda N_0)^{l-r}.
\end{align*}
Here, for \(l=p\), we interpret
\(f_{p,\mathbf m}^\varepsilon(0)=v^\varepsilon(0)\).  
By \eqref{F:bound1}, \eqref{0.2} and
\(\lambda N_0\lesssim|\log\lambda|\), we have
\begin{align*}
	\lambda^2\,\Tr([[\ud\Gamma(f_k^+),\WW],\ud\Gamma(f_k^+)]\Gamma_{\lambda})\,\lesssim
	\lambda^2
	\sum_{l=2}^p
	\varepsilon^{2-2l}
	|\log\lambda|^{p-l}
	\sum_{r=1}^l
	|\log\lambda|^{pr+l-r}
	\lesssim
	\lambda^2\varepsilon^{2-2p}|\log\lambda|^{p^2}.
\end{align*}
Since \(N\le \lambda^{-1/8}\), for \(0<\lambda<1\), we have $\lambda^2|\log\lambda|^{p^2}\lesssim N^{-\frac12}|\log\lambda|^p$. 
Combining with \eqref{4.35} and \eqref{correlation ineq}, the proof is complete.

\end{proof}

\section{Convergence of the relative free energy}\label{sec5}

In this section we prove the free-energy convergence \eqref{thm-eq1}. The argument follows the finite-dimensional reduction developed in \cite[Sections~9--10]{LNR21}. The high-momentum correlation estimate from the preceding section allows us to replace the interaction $\WW$ by its low-frequency truncation, to which a semiclassical approximation applies. A pointwise comparison between the Gaussian lower symbol and the cylindrical Gaussian measure, derived in \cite{NZZ25}, is then used to obtain estimates when $\varepsilon$ depends polynomially on $\lambda$. For the upper bound, we additionally introduce an auxiliary negative number-operator term, which is bounded by the projected $p$-body interaction uniformly.

First, we recall the localization method on Fock space and introduce the quantum de Finetti measure.
For $f\in \cH$, we define the creation operator $a^\dagger(f)$ and the annihilation operator $a(f)$ on $\cF$ by
\begin{align*}
	(a^\dagger(f)\psi)(x_1,...,x_{n+1})&=\frac{1}{\sqrt{n+1}} \sum_{i=1}^{n+1} f(x_i)  \psi(x_1,...,x_{i-1},x_{i+1},...,x_{n+1}),
	\\(a(f)\psi)(x_1,...,x_{n-1})&=\sqrt{n} \int_\Lambda \overline{f(x)}\psi(x,x_1,...,x_{n-1})\ud x,
\end{align*}
where $\psi\in\cH^n$.
They satisfy the canonical commutation relations for all $f,g\in \cH$ :
\begin{align*}
	[a(f),a^\dagger(g)]= \langle f,g\rangle, \quad [a^\dagger(f),a^\dagger(g)]=[a(f),a(g)]=0.
\end{align*}
Let $P$ be an orthogonal projection on $\cH=L^2(\Lambda)$ and let $Q=\1-P$. Then there exists a unitary
\begin{align}\label{eq:Fock factor}
	\fU : \cF ( P\cH \oplus Q\cH) \mapsto \cF(P\cH) \otimes \cF (Q\cH)
\end{align}
satisfying $\fU\fU^*=\fU^*\fU=\1$ and
\begin{align}\label{eq:loc creation}
	\fU a^\dagger(f) \fU ^*  = a ^\dagger(Pf) \otimes \1 + \1 \otimes a^\dagger(Qf)
\end{align}
for every $f\in\cH$, with the analogous identity for $a(f)$.
For any state $\Gamma$ on $\cF$ and any orthogonal projector $P$, we define its localization $\Gamma_{P}$ as a state on $\cF(P\cH)$ obtained by taking the partial trace over $\cF(Q\cH)$:
$$\Gamma_{P}:=\Tr_{\cF(Q\cH)}\left[ \fU \Gamma \fU^* \right].$$
The density matrices of $\Gamma_P$ can be shown to be equal to
\begin{align}\label{eq:GammaV-k}
	(\Gamma_P)^{(k)}=P^{\otimes k}\Gamma^{(k)}P^{\otimes k},\quad k\geq1.
\end{align}

Next, we introduce an important probability measure on $P\cH$.
For $u\in P\cH$, define the coherent state on $\cF(P\cH)$:
\begin{align}\label{eq:coherent state}
	W(u):= \exp(a^\dagger(u)-a(u))|0\rangle = e^{-\|u\|_{L^2(\Lambda)}^2/2} \bigoplus_{n=0}^\infty \frac{u^{\otimes n}} {\sqrt{n!}}\,,
\end{align}
where $|0\rangle$ is the vacuum in $\cF(P\cH)$.
It satisfies the resolution of identity
\begin{align}\label{resolution of identity}
	\1_{\cF(P\cH)}=\pi^{-\Tr(P)}\int_{P\cH}|W(u)\rangle\langle W(u)|\ud u \,,
\end{align}
where $\ud u$ is the usual Lebesgue measure on $P\cH \simeq \CC^{\Tr(P)}$,
and its $k$-particle density matrix is given by
\begin{equation}\label{density matrix of coherent state}
	|W(u)\rangle \langle W(u)|^{(k)} = \frac{1}{k!} |u^{\otimes k}\rangle \langle u^{\otimes k}|\,. 
\end{equation}
Then we have the following quantum de Finetti theorem  \cite[Lemma 6.2 and Remark 6.4]{LNR15}, which relates the density matrices of a localized quantum state to correlation function of a probability measure.

\begin{theorem}[Lower symbols as de Finetti measures]\label{thm:quant deF}
	For any state $\Gamma$ on $\cF$, using the coherent states in \eqref{eq:coherent state}, we define the {lower symbol} of $\Gamma$ on $P\cH$ at scale $\lambda$ by
	\begin{align}\label{eq:Husimi}
		\ud\mu_{P,\Gamma}^\lambda(u):=(\lambda\pi)^{-\Tr(P)}\Big\langle W(u/\sqrt{\lambda}),\Gamma_P W(u/\sqrt{\lambda})\Big\rangle_{\cF(P\cH)} \ud u \,.
	\end{align}
	Then, for all $k\in\NN$, it holds that
	\begin{align}\label{eq:Chiribella}
		\int_{P\cH}|u^{\otimes k}\rangle\langle u^{\otimes k}|\;\ud\mu^\lambda_{P,\Gamma}(u) = k!\lambda^k(\Gamma_P)^{(k)} + k! \lambda ^k \sum_{l = 0} ^{k-1} {k \choose l} (\Gamma_P)^{(l)} \otimes_s \1_{\otimes_s ^{k-l} P\cH}.
	\end{align}
	Thus, with $d=\Tr [P]$,
	\begin{align}\label{eq:quantitative}
		\Tr \left| k!\lambda^k(\Gamma_P)^{(k)}-\int_{P\cH}|u^{\otimes k}\rangle\langle u^{\otimes k}|\;\ud\mu^\lambda_{P,\Gamma}(u) \right| \leq \lambda^k \sum_{l=0}^{k-1}{k\choose l}^2  \frac{(k-l +d-1)!}{(d-1)!}\Tr \left[ \fN^{l}\Gamma_P\right].
	\end{align}
\end{theorem}

Finally, we recall a Berezin-Lieb type inequality given in \cite[Theorem 7.1]{LNR15}, which links the relative entropy of two quantum states to the classical entropy of their de Finetti measures.

\begin{theorem}[Relative entropy: quantum to classical] \label{thm:rel-entropy}
	Let $\Gamma$ and $\Gamma'$ be two states on $\cF$. Let $\mu_{P,\Gamma}^\lambda$ and $\mu_{P,\Gamma'}^\lambda$ be the lower symbols defined in \eqref{eq:Husimi}.  Then we have
	\begin{align}\label{eq:Berezin-Lieb}
		\fH(\Gamma,\Gamma')\geq \fH(\Gamma_P,\Gamma'_P)\geq \fH_{\textnormal{cl}}(\mu_{P,\Gamma}^\lambda,\mu_{P,\Gamma'}^\lambda).
	\end{align}
\end{theorem}

\subsection{Free energy lower bound}

We first prove the lower bound on the relative free energy. The high-momentum truncation and the finite-dimensional semiclassical analysis follow the same general strategy as in the quartic case \(p=2\), but the localization of the interaction requires additional estimates. When \(p=2\), the interaction \(\mathbb W\) contains only the
two-body kernel \(v^\varepsilon\), together with the one-body and scalar counterterms. For \(p\ge3\), Wick renormalization produces counterterms of every order \(0\le l<p\), so we perform the localization for all the kernels \(F_l^\varepsilon\).


\begin{lem}\label{Lem:lower bound}
Let $\theta_0:=\frac12\min\{1,\delta_0\}$ with $\delta_0$ defined in \eqref{v-fourier}.
	Assume that $0<\eta<\min\left\{\frac{1}{96p},
	\frac{15\theta_0}{16(2p-2+\theta_0)}\right\}$.
	Let $0<\lambda<1/2$, and
	$\lambda^\eta\leq\varepsilon\leq(4R)^{-1}$ with $R$ defined in \eqref{def:R}. Let $N\in\NN$ satisfy
	\(\lambda^{-12p\eta}\leq N\leq\lambda^{-1/8}\). 
	Let $W_p^{\varepsilon,N}$ be defined in \eqref{W}.
	Consider the partition functions $Z_{\lambda}$ and $Z_0$ defined in \eqref{def:Gibbs state} and \eqref{def:Gaussian state} respectively. Then there exists $C>0$ depending only on $p$ and the function $v$, such that 
	\begin{align}
		-\log\frac{Z_{\lambda}}{Z_0}\geq&-\log\int e^{-W_p^{\varepsilon,N}(u)}\ud\mu_{0}(u)-C|\log\lambda|^{p^2}((\log N)^{p-2}+1)
		(N^{-1/4}\varepsilon^{3-3p}+\lambda^\delta)\,,\label{2.0}
	\end{align}
where $\delta:=15\theta_0/16-(2p-2+\theta_0)\eta>0$.
\end{lem}

\begin{proof}
	Recall that $e_k^-=P_Ne_kP_N$ for any $k\in\ZZ^2$, with $P_N=\fF^{-1}(\widehat\rho(N^{-1}\cdot)\fF)$. We define the quantum interaction on $\cF$ related to the truncated interaction term $W_p^{\varepsilon,N}$:
	\begin{align}
		\WW_N:=&\sum_{l=2}^p\lambda^l\sum_{k_{1:l}}\fF(F^{\varepsilon,N}_l)(k_{1:l})
		\prod_{j=1}^{l}\Wick{\ud\Gamma(e_{k_j}^-)} + f_1^{\varepsilon,N}\lambda\Wick{\ud\Gamma(P_N^2)} + f_0^{\varepsilon,N}\,,\label{def:W_N}
	\end{align}
where $F_l^{\varepsilon,N}(x_{1:l}):=f_l^{\varepsilon,N}(x_2-x_1,\cdots,x_l-x_1)$ as in \eqref{F},
with $f_p^{\varepsilon,N}=\frac1pv^\varepsilon$ and $f_l^{\varepsilon,N}$, $l<p$, given by \eqref{def:F_l} with $G$ replaced by $G_N$.

\noindent {\bf Step 1: High-frequency truncation.} First, we estimate the expectation of $\WW-\WW_{N}$. From \eqref{def:quantum-W} and \eqref{def:W_N},  $$\Tr((\WW-\WW_{N})\Gamma_{\lambda})= \textnormal{I}+\textnormal{II}\,,$$ where
\begin{align*}
	\textnormal{I}=&\sum_{l=2}^{p-1}\lambda^l\sum_{k_{1:l}}\fF(F^{\varepsilon}_l-F^{\varepsilon,N}_l)(k_{1:l})
	\left\langle\prod_{j=1}^{l}\Wick{\ud\Gamma(e_{k_j})}\right\rangle_{\lambda} + (f_1^\varepsilon-f_1^{\varepsilon,N})\lambda\Tr(\Gamma_\lambda^{(1)}-\Gamma_0^{(1)}) + f_0^\varepsilon-f_0^{\varepsilon,N}\,,\\
	\textnormal{II}=&\sum_{l=2}^{p}\lambda^l\sum_{k_{1:l}}\fF(F^{\varepsilon,N}_l)(k_{1:l})\left\langle\prod_{j=1}^{l}\Wick{\ud\Gamma(e_{k_j})}-\prod_{j=1}^{l}\Wick{\ud\Gamma(e_{k_j}^-)}\right\rangle_{\lambda} + f_1^{\varepsilon,N}\lambda\Tr((1-P_N^2)(\Gamma_\lambda^{(1)}-\Gamma_0^{(1)}))\,.
\end{align*}
Here I and II correspond to two different approximation errors. The term I measures the error caused by replacing the singular Green function $G$ with its regularization $G_N$ in the interaction kernels. In contrast, II measures the error caused by restricting the interaction to low momentum modes, which requires the high-momentum correlation estimate from Section \ref{sec4.4}.

We first bound the term I. 
For the one-body term, by \eqref{eqa priori bound on RDM} and
\(\Tr(h^{-2})<\infty\),
\[
\lambda\left|\Tr(\Gamma_\lambda^{(1)}-\Gamma_0^{(1)})\right|
\le\lambda\left\|h^{1/2}(\Gamma_\lambda^{(1)}-\Gamma_0^{(1)})h^{1/2}\right\|_{\mathrm{HS}}\|h^{-1}\|_{\mathrm{HS}}
\lesssim|\log\lambda|^p.
\]
Hence
\[
\left|
(f_1^\varepsilon-f_1^{\varepsilon,N})
\lambda\Tr(\Gamma_\lambda^{(1)}-\Gamma_0^{(1)})
\right|
\lesssim|\log\lambda|^p
|f_1^\varepsilon-f_1^{\varepsilon,N}|.
\]
For each \(2\le l\le p-1\), by \eqref{0.2} and \(\lambda N_0\lesssim|\log\lambda|\), we have
\[
\lambda^l
\left|
\left\langle
\prod_{j=1}^l\Wick{\ud\Gamma(e_{k_j})}
\right\rangle_\lambda
\right|
\lesssim
\lambda^l
\left\langle(\mathcal N+N_0)^l\right\rangle_\lambda
\lesssim
|\log\lambda|^{pl}.
\]
In order to estimate the Fourier coefficients of
\(F_l^\varepsilon-F_l^{\varepsilon,N}\), we use the component decomposition,
\[
F_l^\varepsilon-F_l^{\varepsilon,N}
=
\sum_{\substack{\fm\in\fM_{p,l}\\ C_\fm\neq0}}
C_\fm
\bigl(F_{l,\fm}^\varepsilon-F_{l,\fm}^{\varepsilon,N}\bigr),
\]
where $F_{l,\fm}^\varepsilon$ and $F_{l,\fm}^{\varepsilon,N}$
are defined in \eqref{def:Fl,m}.
By \eqref{f:positive Fourier}, we have
\[
\sum_{k_{1:l}}
\left|\mathcal F(F_l^\varepsilon-F_l^{\varepsilon,N})(k_{1:l})
\right|
\lesssim
\sum_{\substack{\fm\in\fM_{p,l}\\ C_\fm\neq0}}
\sum_{k_{1:l}}
\mathcal F(F_{l,\fm}^\varepsilon-F_{l,\fm}^{\varepsilon,N})(k_{1:l})
=\sum_{\substack{\fm\in\fM_{p,l}\\ C_\fm\neq0}}(f_{l,\fm}^\varepsilon(0)-f_{l,\fm}^{\varepsilon,N}(0)).
\]
Combining the above estimates, we obtain
\begin{align}
	|\textnormal{I}|\lesssim&|f_0^\varepsilon-f_0^{\varepsilon,N}|+|\log\lambda|^p|f_1^\varepsilon-f_1^{\varepsilon,N}|
	+\sum_{l=2}^{p-1}|\log\lambda|^{pl}\sum_{\substack{\fm\in\fM_{p,l}\\C_\fm\neq0}}(f^{\varepsilon}_{l,\fm}(0)-f^{\varepsilon,N}_{l,\fm}(0)).\label{5.13}
\end{align}
For $i=0,1$, a telescopic expansion replaces one factor $G$ by $G_N$ at a time.  Every resulting term contains one factor $F_N:=G-G_N$ and otherwise only factors $G$ or $G_N$. The proof of \eqref{G_N1}, with \eqref{F_{N,M}} and the H\"older's inequality, therefore gives
\begin{align}
	|f_i^\varepsilon-f_i^{\varepsilon,N}|
	\lesssim&|\log\varepsilon|^{p-1-i}\left(\int_{|x|\lesssim1}|F_N(\varepsilon x)|^2\ud x\right)^\frac12\lesssim\varepsilon^{-1}|\log\varepsilon|^{p-1-i}N^{-1}.\label{5.14}
\end{align}
Fix $2\leq l\leq p-1$ and $\fm\in\fM_{p,l}$ with $C_\fm\neq0$. By \eqref{C_g2}, $m_{ab}=0$ for $1\leq a<b\leq l,$ so every Green-function factor contains at least one of the integrated variables $x_{l+1},\ldots,x_p$. Enumerate the $p-l$ edge copies encoded by $\fm$, including repetitions, as $(a_s,b_s),$ $1\leq s\leq p-l.$ Then, we have
\begin{align*}
	f_{l,\fm}^{\varepsilon}(0)
	&=\int_{\Lambda^{p-l}}v^\varepsilon(0,\ldots,0,x_{l+1},\ldots,x_p)
	\prod_{s=1}^{p-l}G(z_s)\,\ud x_{l+1:p},
\end{align*}
where $z_s:=x_{a_s}-x_{b_s}$ for $1\leq s\leq p-l$ and $x_{1:l}:=0$. $f_{l,\fm}^{\varepsilon,N}(0)$ is analogous with $G$ replaced by $G_N$. Since $F_N=G-G_N$, the exact telescopic identity is
\begin{align*}
	\prod_{s=1}^{p-l}G(z_s)-\prod_{s=1}^{p-l}G_N(z_s)
	=\sum_{r=1}^{p-l} F_N(z_r)
	\prod_{s<r}G_N(z_s)\prod_{s>r}G(z_s).
\end{align*}
Since $|G(x)|+|G_N(x)|\lesssim1+|\log|x||$ for $x\in\Lambda$, it follows that
\begin{align*}
	|f_{l,\fm}^{\varepsilon}(0)-f_{l,\fm}^{\varepsilon,N}(0)|
	\lesssim\sum_{r=1}^{p-l} \int_{\Lambda^{p-l}}v^\varepsilon(0,\ldots,0,x_{l+1},\ldots,x_p)|F_N(z_r)|\prod_{s\neq r}(1+|\log d(z_s)|)\,\ud x_{l+1:p}.
\end{align*}
As in \eqref{G_N1}, by the change of variables $x_j=\varepsilon h_j$, and using \eqref{F_{N,M}} together with H\"older's inequality, we obtain
\begin{align*}
	|f_{l,\fm}^{\varepsilon}(0)-f_{l,\fm}^{\varepsilon,N}(0)|
	&\lesssim \varepsilon^{-2(l-1)}|\log\varepsilon|^{p-l-1}
	\left(\int_{|z|\leq C}|F_N(\varepsilon z)|^2\,\ud z\right)^{1/2}\lesssim \varepsilon^{-2(l-1)-1}|\log\varepsilon|^{p-l-1}N^{-1},
\end{align*}
where $C$ depends only on the support of $v$. 
Therefore, by \eqref{5.14} and \eqref{entropy}, \eqref{5.13} becomes
\begin{align}
	|\textnormal{I}|\lesssim N^{-1}\varepsilon^{3-2p}|\log\lambda|^{p(p-1)}\,.\label{4.1}
\end{align}

Now we estimate $\textnormal{II}$. For $l=1$, by \eqref{eqa priori bound on RDM}, we have
\begin{align}
	\lambda|\Tr(1-P_N^2)(\Gamma_{\lambda}^{(1)}-\Gamma_0^{(1)})|
	\lesssim\lambda\|h^{\frac12}(\Gamma_{\lambda}^{(1)}-\Gamma_0^{(1)})h^{\frac12}\|_{\text{HS}}\|(1-P_N^2)h^{-1}\|_{\text{HS}}\lesssim N^{-1}|\log\lambda|^p \,.\label{5.15}
\end{align}
For $2\leq l\leq p$, we write
\begin{align}\label{5.18}
	\left\langle\prod_{j=1}^{l}\Wick{\ud\Gamma(e_{k_j})}-\prod_{j=1}^{l}\Wick{\ud\Gamma(e_{k_j}^-)}\right\rangle_{\lambda}
	 =\sum_{i=1}^l\left\langle B_1\cdots B_{i-1}A_iC_{i+1}\cdots C_l\right\rangle_\lambda,
\end{align}
where $A_j=\Wick{\ud\Gamma(e_{k_j}^+)},$ $B_j=\Wick{\ud\Gamma(e_{k_j}^-)}$ and $C_j=A_j+B_j$ for $1\leq j\leq l$. 
Fix $1\leq i\leq l$. We now move \(A_i\) to the left of the factors \(B_j\). By the commutator identity, we have
\begin{align}
	\left\langle B_1\cdots B_{i-1}A_iC_{i+1}\cdots C_l\right\rangle_\lambda
	=
\left\langle A_iB_1\cdots B_{i-1}C_{i+1}\cdots C_l\right\rangle_\lambda
	+\left\langle
	\left[B_1\cdots B_{i-1},A_i\right]C_{i+1}\cdots C_l
	\right\rangle_\lambda.\label{5.16}
\end{align}
For the second term, expanding the commutator by the Leibniz rule as in \eqref{change}, we obtain
\[
\lambda^l\left|\left\langle
\left[B_1\cdots B_{i-1},A_i\right]C_{i+1}\cdots C_l
\right\rangle_\lambda\right|
\lesssim\lambda^l
\left\langle\fN(\fN+N_0)^{l-2}\right\rangle_\lambda\lesssim\lambda|\log\lambda|^{p(l-1)},
\]
where we used \eqref{0.2} and $\lambda N_0\lesssim|\log\lambda|$ in the last inequality.
Similarly, it holds that
$$\lambda^{l-1}\left\langle \left|B_1\cdots B_{i-1}C_{i+1}\cdots C_l\right|^2\right\rangle_\lambda^\frac12
\lesssim\langle(\lambda(\fN+N_0))^{2(l-1)}\rangle_\lambda^\frac12
\lesssim|\log\lambda|^{p(l-1)}.$$
Therefore, by \eqref{5.16} and the Cauchy-Schwarz inequality, \eqref{5.18} is bounded by
\begin{align}\label{1.5}
	\lambda^l\left|\left\langle\prod_{j=1}^{l}\Wick{\ud\Gamma(e_{k_j})}-\prod_{j=1}^{l}\Wick{\ud\Gamma(e_{k_j}^-)}\right\rangle_{\lambda}\right|
	\lesssim|\log\lambda|^{p(l-1)}\sum_{i=1}^l
	\langle|\lambda A_i^*|^2\rangle_\lambda^\frac12+\lambda|\log\lambda|^{p(l-1)}.
\end{align}
For each $1\leq i\leq l$, by \eqref{1.4} and $\langle|\lambda A_i^*|^2\rangle_\lambda\leq\langle|\lambda(\fN+N_0)|^2\rangle_\lambda\lesssim|\log\lambda|^{2p}$, we have
\begin{align*}
	\langle|\lambda A_i^*|^2\rangle_\lambda^\frac12
	\lesssim&\min\left\{(N^{-\frac12}\varepsilon^{2-2p}|\log\lambda|^p
	+\lambda^2|\log\lambda|^pN|k_i|^2)^\frac12,|\log\lambda|^p\right\},
\end{align*}
where we used $|k_i|^2+N|k_i|\lesssim N |k_i|^2$.
Therefore, by \eqref{1.5} and the symmetry of $F_l^{\varepsilon,N}$, summing over $k_{1:l}\in(\ZZ^2)^l$, we have 
\begin{align}
	&\lambda^l\left|\sum_{k_{1:l}\in(\ZZ^2)^l}\fF(F^{\varepsilon,N}_l)(k_{1:l})\left\langle\prod_{j=1}^{l}\Wick{\ud\Gamma(e_{k_j})}-\prod_{j=1}^{l}\Wick{\ud\Gamma(e_{k_j}^-)}\right\rangle_{\lambda}\right|
	\lesssim|\log\lambda|^{p(l-1)}\sum_{k_{1:l}\in(\ZZ^2)^l}|\fF(F^{\varepsilon,N}_l)(k_{1:l})|\Theta(|k_1|)\,,\label{5.17}
\end{align}
where $$\Theta(|k_1|):=\min\left\{(N^{-\frac12}\varepsilon^{2-2p}|\log\lambda|^p
+\lambda^2|\log\lambda|^pN|k_1|^2)^\frac12,|\log\lambda|^p\right\}+\lambda.$$

In the following, we prove that, for each $2\leq l\leq p$,
\begin{align}\label{1.4'}
	\sum_{k_{1:l}\in(\mathbb Z^2)^l}
	|\fF(F_l^{\varepsilon,N})(k_{1:l})|\Theta(|k_1|)
	\lesssim
	((\log N)^{p-l}+1)\sum_{t_{1:p-1}\in(\ZZ^2)^{p-1}}
	\widehat v(\varepsilon t_{1:p-1})\Theta(|t_1|+CN),
\end{align}
where \(C>0\) depends only on \(p\).
For \(l=p\), since
\[
\mathcal F(F_p^{\varepsilon,N})(k_{1:p})
=
p^{-1}\1_{k_1+\cdots+k_p=0}\,
\widehat v(\varepsilon k_{2:p}),
\]
by the definition of $\Theta$ and the symmetry of \(\widehat v\), we have
\[
\sum_{k_{1:p}}
|\mathcal F(F_p^{\varepsilon,N})(k_{1:p})|\Theta(|k_1|)
\lesssim
\sum_{j=2}^p\sum_{k_{2:p}}\widehat v(\varepsilon k_{2:p})\Theta(|k_j|)\lesssim\sum_{k_{2:p}}
\widehat v(\varepsilon k_{2:p})\Theta(|k_2|),\]
which is bounded by the right-hand side of \eqref{1.4'}.
For \(2\le l<p\), by \eqref{f:positive Fourier}, we have
\begin{align}\label{F_l,m}
	|\fF(F_l^{\varepsilon,N})(k_{1:l})|
	\lesssim\sum_{\substack{\fm\in\mathcal M_{p,l}\\C_\fm\neq0}}\fF(F_{l,\fm}^{\varepsilon,N})(k_{1:l}).
\end{align}
Fix such an \(\fm\). Let \(E_p:=\{(a,b):1\le a<b\le p\}\) and write \(m_\me=m_{ab}\) for \(\me=(a,b)\). We use the convention \(G_N^0\equiv1\). By \eqref{B5}, we have
\begin{align}\label{F:fourier}
	\fF(F^{\varepsilon,N}_{l,\fm})(k_{1:l})=
	\1_{k_1+\cdots+k_l=0}
	\sum_{(q_\me)_{\me\in E_p}}\left(\prod_{\me\in E_p}\fF(G_N^{m_\me})(q_\me)\right)\widehat{v^\varepsilon}(k_2-\mathbf q_1,\cdots,k_l-\mathbf q_{l-1},-\mathbf q_{l:p-1})
\end{align}
where 
$$\mathbf q_r=\sum_{\me=(a,b)\in E_p}(\1_{a=r+1}-\1_{b=r+1})q_\me,\quad 1\leq r\leq p-1.$$
Set $E_*:=\{\me_j=(1,j+1): l\le j\le p-1\}.$
For \(j=l,\ldots,p-1\), the edge \(\me_j\) contributes only to \(\mathbf q_j\), and its contribution is \(-q_{\me_j}\). More precisely, for fixed \((q_\me)_{\me\in E_p\setminus E_*}\), we write $\mathbf q_j=-q_{\me_j}+\widetilde{\mathbf q}_j,$ where
\[
\widetilde{\mathbf q}_j
:=
\sum_{\me=(a,b)\in E_p\setminus E_*}
(\1_{a=j+1}-\1_{b=j+1})q_\me .
\]
Then, by the change of variable $t_i=k_{i+1}-\mathbf q_i$ for $1\leq i\leq l-1$ and $t_j=q_{\me_j}-\widetilde{\mathbf q_j}$ for $l\leq j\leq p-1$, it follows that
\begin{align*}
	&\sum_{k_{1:l}\in(\mathbb Z^2)^l}
	\big|\mathcal F(F_{l,\fm}^{\varepsilon,N})(k_{1:l})\big|
	\Theta(|k_1|)\\
	=&
	\sum_{(q_\me)_{e\me\in E_p\setminus E_*}}
	\left(\prod_{\me\in E_p\setminus E_*}
	\mathcal F(G_N^{m_\me})(q_\me)\right)
	\sum_{t_{1:p-1}}
	\widehat v(\varepsilon t_{1:p-1})
	\left(\prod_{j=l}^{p-1}\fF(G_N^{m_{\me_j}})
	(t_j+\widetilde{\mathbf q_j})\right)
	\Theta\left(\left|\sum_{i=1}^{l-1}(t_i+\mathbf q_i)\right|\right).
\end{align*}
By \eqref{GN:bound}, we bound $\fF(G_N^{m_{\me_j}})$ by a constant, and on the
support of the remaining product, \(|\mathbf q_i|\lesssim N\) for \(1\le i\le l-1\). Therefore,
$$\Theta\left(\left|\sum_{i=1}^{l-1}(t_i+\mathbf q_i)\right|\right)\lesssim\sum_{i=1}^{l-1}\Theta(|t_i+\mathbf q_i|)\lesssim\sum_{i=1}^{l-1}\Theta(|t_i|+CN),$$
for some $C>0$.
Then by the symmetry of $v$, we have
\begin{align*}
	&\sum_{k_{1:l}\in(\mathbb Z^2)^l}
	\big|\mathcal F(F_{l,\fm}^{\varepsilon,N})(k_{1:l})\big|
	\Theta(|k_1|)
	\lesssim\left(\prod_{\me\in E_p\setminus E_*}G_N(0)^{m_\me}\right)
	\sum_{t_{1:p-1}}\widehat v(\varepsilon t_{1:p-1})\Theta(|t_1|+CN).
\end{align*}
Since $\sum_{\me\in E_p}m_\me=p-l,$ by \eqref{C3_2}, \eqref{1.4'} holds for $2\leq l<p$.
Substituting into \eqref{5.17}, by \eqref{5.15} and \eqref{f_0}, we obtain
\begin{align}
	|\textnormal{II}|\lesssim |\log\lambda|^{p(p-1)}((\log N)^{p-2}+1)
	\sum_{t_{1:p-1}\in(\mathbb Z^2)^{p-1}}
	\widehat v(\varepsilon t_{1:p-1})\Theta(|t_1|+CN)+N^{-1}|\log\varepsilon|^{p-1}|\log\lambda|^{p}\,.\label{5.21}
\end{align}

Since $N\leq\lambda^{-1/8}$ and $0<\theta_0<1$, by the definition of \(\Theta\), we have
\begin{align*}
	\Theta(|t_1|+CN)
	\lesssim&N^{-1/4}\varepsilon^{1-p}|\log\lambda|^{p/2}+(\lambda N^{1/2}|t_1|)^{\theta_0}|\log\lambda|^p,
\end{align*}
and by \eqref{v-fourier},
\[
\sum_{t_{1:p-1}\in(\mathbb Z^2)^{p-1}}
\widehat v(\varepsilon t_{1:p-1}) |t_1|^{\theta_0}
\lesssim\sum_{t_{1:p-1}\in(\mathbb Z^2)^{p-1}}
(1+\varepsilon |t_{1:p-1}|)^{-2p+2-\delta_0} |t_1|^{\theta_0}\lesssim
\varepsilon^{2-2p-\theta_0}.
\]
Therefore,
\begin{align*}
	\sum_{t_{1:p-1}}
	\widehat v(\varepsilon t_{1:p-1})
	\Theta(|t_1|+CN)
	\lesssim&
	N^{-1/4}\varepsilon^{3-3p}|\log\lambda|^{p/2}
	+\lambda^{\theta_0} N^{\theta_0/2}\varepsilon^{2-2p-\theta_0}
	|\log\lambda|^p.
\end{align*}
Substituting into \eqref{5.21}, since $\varepsilon\geq\lambda^\eta$, we obtain
\begin{align*}
	|\textnormal{II}|
	\lesssim&|\log\lambda|^{p^2}((\log N)^{p-2}+1)
	(N^{-1/4}\varepsilon^{3-3p}+\lambda^\delta)\,,
\end{align*}
where $\delta=15\theta_0/16-(2p-2+\theta_0)\eta$.
Combining with \eqref{4.1}, we obtain that
\begin{align}
	|\Tr((\WW-\WW_{N})\Gamma_\lambda)|
	\lesssim |\log\lambda|^{p^2}((\log N)^{p-2}+1)
	(N^{-1/4}\varepsilon^{3-3p}+\lambda^\delta).\label{4.6}
\end{align}

\noindent {\bf Step 2: Semiclassical analysis of the low-frequency interaction.}  
Now we turn to the low-momentum
part of the interaction and use the de Finetti theorem to connect with the classical field measure.
By \eqref{0.2} and the same argument as in \eqref{1.72} and \eqref{1.71},
for every $1\leq l\leq p$ and $k_{1:l}\in\ZZ^{2l}$, we have
\begin{align}
	\lambda^l\left\langle\prod_{j=1}^l\Wick{\ud\Gamma(e_{k_j}^-)}\right\rangle_{\lambda}=\lambda^l\sum_{Q\subseteq [l]}\prod_{i\notin Q}(-\langle\ud\Gamma(e_{k_i}^-)\rangle_0)
	\,{q!}\,\Tr\left(\mathop{\bigotimes^{\text{sym}}}\limits_{j\in Q}e_{k_j}^-\Gamma_\lambda^{(q)}\right)+O(\lambda |\log\lambda|^{p(l-1)})\,,\label{5.26}
\end{align}
where $q=|Q|$; for $Q=\varnothing$, we use the convention $q=0$.
By \cite[(9.9)]{LNR21}, we have uniformly in \(k\),
\begin{align}
	\left|\lambda \Tr(e_k^-\Gamma_0^{(1)}) - \int\langle u,e_k^-u\rangle\ud\mu_0(u)\right|\lesssim \lambda N^2.\label{9.9}
\end{align}
Since $\int|u\rangle\langle u|\ud\mu_0(u)=h^{-1}$ and \(N\le \lambda^{-1/8}\), it follows that
\begin{align}
	|\lambda \Tr(e_k^-\Gamma_0^{(1)})|+\left|\int\langle u,e_k^-u\rangle\ud\mu_0(u)\right|\leq\lambda N_0+\Tr(P_Nh^{-1})
	\lesssim |\log\lambda|+\log N+1\lesssim|\log\lambda|.\label{5.20}
\end{align}
Then, for every $q=|Q|<l$, we have
\begin{align}
	\left|\prod_{i\notin Q}(\lambda\langle\ud\Gamma(e_{k_i}^-)\rangle_0)-\prod_{i\notin Q}\int\langle u,e_{k-i}^-u\rangle\ud\mu_0(u)\right|\lesssim&|\log\lambda|^{l-q-1}\sum_{j\notin Q}\left|\lambda \Tr(e_{k_j}^-\Gamma_0^{(1)}) - \int\langle u,e_{k_j}^-u\rangle\ud\mu_0(u)\right|\nonumber\\
	\lesssim&\lambda N^2|\log\lambda|^{l-q-1}.\label{difference}
\end{align}
Moreover, since \(\|e_k^-\|\le1\), by \eqref{0.2}, we have
\[\lambda^q
\left|\Tr\left(
\mathop{\bigotimes^{\mathrm{sym}}}_{j\in Q}
e_{k_j}^-\Gamma_\lambda^{(q)}
\right)
\right|
\lesssim \lambda^q\Tr(\Gamma_\lambda^{(q)})
\leq\lambda^q\Tr(\mathcal N^q\Gamma_\lambda)
\lesssim
|\log\lambda|^{pq}.
\]
Since $\sum_{q=0}^{l-1}|\log\lambda|^{l-q-1}|\log\lambda|^{pq}
\le|\log\lambda|^{p(l-1)}$, \eqref{5.26} becomes
\begin{align}
	\lambda^l\left\langle\prod_{j=1}^l\Wick{\ud\Gamma(e_{k_j}^-)}\right\rangle_{\lambda}=\sum_{Q\subseteq [l]}\prod_{i\notin Q}\left(-\int\langle u,e_{k_i}^-u\rangle\ud\mu_0\right)
	\,{q!}\,\lambda^q\,\Tr\left(\mathop{\bigotimes^{\text{sym}}}\limits_{j\in Q}e_{k_j}^-\Gamma_{\lambda}^{(q)}\right)+O(\lambda N^2|\log\lambda|^{p(l-1)})\,.\label{4.7}
\end{align}

Let $P:=\1_{\{h\leq(2\pi)^2N^2+1\}}$. 
Since $P_NP=PP_N=P_N$, by \eqref{eq:GammaV-k}, we have
\begin{align}
	\Tr\left(\mathop{\bigotimes^{\text{sym}}}\limits_{j\in Q}e_{k_j}^-\Gamma_{\lambda}^{(q)}\right)=\Tr\left(\mathop{\bigotimes^{\text{sym}}}\limits_{j\in Q}e_{k_j}^-(\Gamma_{\lambda})_P^{(q)}\right)\,.\label{5.32}
\end{align}
Let $\mu_{P,\lambda}^\lambda$ be the lower symbol of $\Gamma_{\lambda}$ associated with the projection $P$ at scale $\lambda$. 
Using Theorem \ref{thm:quant deF}, for $1\leq k\leq l$, we have
\begin{align*}
	\lambda^k(\Gamma_{\lambda})_P^{(k)}& = \frac{1}{k!}\int_{P\cH}|u^{\otimes k}\rangle\langle u^{\otimes k}|\ud\mu_{P,\lambda}^\lambda(u)-\lambda^k\sum_{m=0}^{k-1}\binom{k}{m}(\Gamma_{\lambda})_P^{(m)}\otimes_sP^{\otimes_s(k-m)}\,.
\end{align*} 
Again, by \eqref{0.2}, for $0\leq m<k\leq l$, $$\lambda^k\Tr((\Gamma_{\lambda})_P^{(m)}\otimes_sP^{\otimes_s(k-m)})\lesssim_k\lambda^k\Tr(\Gamma_{\lambda}^{(m)})(\Tr P)^{k-m}\lesssim(\lambda N^2)^{k-m}|\log\lambda|^{pm}\,.$$
Therefore, by \eqref{5.32} and $\lambda N^2\leq1$, we have
\begin{align}\label{5.33}
	q!\,\lambda^q\,\Tr\left(\mathop{\bigotimes^{\text{sym}}}\limits_{j\in Q}e_{k_j}^-\Gamma_{\lambda}^{(q)}\right)=\int_{P\cH}
	\prod_{j\in Q}\langle u,e_{k_j}^-u\rangle
	\,\ud\mu_{P,\lambda}^{\lambda}(u) + O(\lambda N^2|\log\lambda|^{p(q-1)})\,.
\end{align}
Substituting this into \eqref{4.7}, by \eqref{5.20} and $\sum_{q=0}^{l-1}|\log\lambda|^{l-q}|\log\lambda|^{p(q-1)}
\le|\log\lambda|^{p(l-1)}$, we have
\begin{align*}
	\lambda^l\left\langle\prod_{j=1}^l\Wick{\ud\Gamma(e_{k_j}^-)}\right\rangle_{\lambda}
	=&\sum_{Q\subseteq [l]}\prod_{i\notin Q}\left(-\int\langle u,e_{k_i}^-u\rangle\ud\mu_0\right)\int_{P\cH}\prod_{j\in Q}\langle u,e_{k_j}^- u\rangle\ud\mu_{P,\lambda}^\lambda(u)
	+O(\lambda N^2|\log\lambda|^{p(l-1)})\\
	=&\int_{P\cH}\prod_{j=1}^l\Wick{\langle u,e_{k_j}^-u\rangle}\ud\mu_{P,\lambda}^\lambda(u)+O(\lambda N^2|\log\lambda|^{p(l-1)})\,.
\end{align*}
Moreover, for $2\leq l<p$, by \eqref{F_l,m} and \eqref{F:bound1}, we have
\begin{align}
	\sum_{k_{1:l}}
	\left|\mathcal F(F_l^{\varepsilon,N})(k_{1:l})\right|
	&\lesssim
	\sum_{\substack{\mathbf m\in\mathcal M_{p,l}\\ C_{\mathbf m}\ne0}}
	\sum_{k_{1:l}}
	\mathcal F(F_{l,\mathbf m}^{\varepsilon,N})(k_{1:l})
	=
	\sum_{\substack{\mathbf m\in\mathcal M_{p,l}\\ C_{\mathbf m}\ne0}}
	F_{l,\mathbf m}^{\varepsilon,N}(0)
	\lesssim
	\varepsilon^{2-2l}|\log\lambda|^{p-l}.\label{5.35}
\end{align}
For \(l=p\), since
\(F_p^{\varepsilon,N}=p^{-1}v^\varepsilon\), the same bound follows from $\widehat v\geq0$ and \eqref{periodiz}.
Hence for every $2\leq l\leq p$, it follows that
\begin{align}
	\lambda^l
	\sum_{k_{1:l}}
	\mathcal F(F_l^{\varepsilon,N})(k_{1:l})
	\left\langle
	\prod_{j=1}^l\Wick{\ud\Gamma(e_{k_j}^-)}
	\right\rangle_\lambda
	\geq&
	\sum_{k_{1:l}}
	\mathcal F(F_l^{\varepsilon,N})(k_{1:l})
	\int_{P\cH}\prod_{j=1}^l
	\fF(\Wick{|u_N|^2})(-k_j)\ud\mu_{P,\lambda}^{\lambda}(u)\label{5.34}\\
	&-C\lambda N^2\varepsilon^{2-2l}|\log\lambda|^{(p-1)l},\nonumber
\end{align}
where we used that $\Wick{\langle u,e_{k_j}^-u\rangle}=\fF(\Wick{|u_N|^2})(-k_j)$. The first term in the right-hand side is exactly the $l$-th component in \eqref{W1} by the Parseval equality.
The one-body term is treated in the same way. 
Taking \(q=1\) and \(e_{k_j}^-=P_N^2\) in \eqref{5.33}, we get
\[
\lambda\left\langle\Wick{\ud\Gamma(P_N^2)}\right\rangle_\lambda
=
\int_{P\cH}
\Wick{\langle u,P_N^2u\rangle}
\,\ud\mu_{P,\lambda}^{\lambda}(u)
+
O(\lambda N^2).
\]
Since \(|f_1^{\varepsilon,N}|\lesssim|\log\varepsilon|^{p-1}\) by \eqref{f_0}, 
it follows that
$$f_1^{\varepsilon,N}\lambda\left\langle\Wick{\ud\Gamma(P_N^2)}\right\rangle_\lambda\geq f_1^{\varepsilon,N}\int_{P\cH}
\Wick{|u_N|^2}\ud\mu_{P,\lambda}^{\lambda}(u)-C\lambda N^2|\log\varepsilon|^{p-1}.$$
Combining this with \eqref{5.34} for $2\leq l\leq p$, we obtain
\[
\Tr(\WW_N\Gamma_\lambda)
\geq
\int_{P\mathcal H}W_p^{\varepsilon,N}(u)\,
\ud\mu_{P,\lambda}^{\lambda}(u)
-C
\lambda N^2\varepsilon^{2-2p}|\log\lambda|^{p(p-1)}.
\]
Finally, since $N\leq \lambda^{-\frac18}$, combining with \eqref{4.6}, we obtain that
\begin{align}
	\Tr(\WW\Gamma_\lambda)
	\geq&\int_{P\cH} W_p^{\varepsilon,N}\ud\mu_{P,\lambda}^\lambda -C |\log\lambda|^{p^2}((\log N)^{p-2}+1)
	(N^{-1/4}\varepsilon^{3-3p}+\lambda^\delta)\label{5.22}
\end{align}

\noindent {\bf Step 3: Variational conclusion.} 
We now combine the preceding estimate with the variational principle and the Gaussian lower-symbol comparison from \cite{NZZ25}.
Let $\mu_{P,0}^\lambda$ be the lower symbol of $\Gamma_0$ associated with the projection $P$ at scale $\lambda$, and let $\mu_{0,P}$ be the cylindrical projection of $\mu_0$ on $P\cH$.
Then by \eqref{5.22}, the variational formula \eqref{variational problem}, and the Berezin-Lieb inequality \eqref{eq:Berezin-Lieb}, we have
\begin{align*}
	-\log\frac{Z_{\lambda}}{Z_0}\geq& \fH_{\text{cl}}(\mu_{P,\lambda}^\lambda,\mu_{P,0}^\lambda)+\int_{P\cH} W_p^{\varepsilon,N}\ud\mu_{P,\lambda}^\lambda-C|\log\lambda|^{p^2}((\log N)^{p-2}+1)
	(N^{-1/4}\varepsilon^{3-3p}+\lambda^\delta)\,.
\end{align*}
The pointwise estimate \cite[(8.25)]{NZZ25}, applied to the present two-dimensional cutoff, gives
\begin{align}\label{5.23}
	\fH_{\text{cl}}(\mu_{P,\lambda}^\lambda,\mu_{P,0}^\lambda)\geq \fH_{\text{cl}}(\mu_{P,\lambda}^\lambda,\mu_{0,P})-C\lambda N^4\int_{P\cH}\|u\|^2\ud\mu_{P,\lambda}^\lambda(u)
	\geq\fH_{\text{cl}}(\mu_{P,\lambda}^\lambda,\mu_{0,P})-C\lambda N^4 |\log\lambda|^{p}\,,
\end{align} 
where in the second inequality we used $\lambda N^2\leq1,$ and
\begin{align*}
	\int_{P\cH}\|u\|^2\ud\mu_{P,\lambda}^\lambda(u)=\lambda \Tr(\Gamma_{\lambda})_P^{(1)}+\lambda\, \Tr P
	\lesssim |\log\lambda|^{p}+\lambda N^2\,.
\end{align*}
Since $\lambda N^4\leq N^{-\frac12}$, we have
\begin{align*}
	-\log\frac{Z_{\lambda}}{Z_0}
	\geq&\fH_{\text{cl}}(\mu_{P,\lambda}^\lambda,\mu_{0,P})+\int_{P\cH} W_p^{\varepsilon,N}\ud\mu_{P,\lambda}^\lambda -C|\log\lambda|^{p^2}((\log N)^{p-2}+1)
	(N^{-1/2}\varepsilon^{3-3p}+\lambda^\delta)\nonumber\\
	=&\fH_{\text{cl}}(\mu_{P,\lambda}^\lambda,\widetilde\mu_p^{\varepsilon,N})-\log\int_{P\cH}e^{-W_p^{\varepsilon,N}(u)}\ud\mu_{0,P}(u)-C|\log\lambda|^{p^2}((\log N)^{p-2}+1)
	(N^{-1/4}\varepsilon^{3-3p}+\lambda^\delta)\,,
\end{align*}
where $\widetilde\mu_p^{\varepsilon,N}=\mu_p^{\varepsilon,N}\circ P^{-1}$ with $\mu_p^{\varepsilon,N}$ defined in \eqref{def:mu_p^N}.
Since $P_NP=PP_N=P_N$, for any $u\in\cH$, we have $W_p^{\varepsilon,N}(u)=W_p^{\varepsilon,N}(Pu)$ and $$\int_{P\cH}e^{-W_p^{\varepsilon,N}(u)}\ud\mu_{0,P}(u)=\int_{\cH} e^{-W_p^{\varepsilon,N}(u)}\ud\mu_0(u)\,,$$
which concludes the proof.

\end{proof}

\subsection{Free energy upper bound}
Now we follow the strategy in \cite[Proposition 10.1]{LNR21} to derive the relative free energy upper bound. Together with \eqref{partition function conv} and \eqref{2.0}, this completes the proof of the free energy convergence \eqref{thm-eq1}.

As in the quartic case, we use a trial state consisting of an interacting Gibbs state on the low-energy Fock space and the free Gibbs state on its orthogonal complement. For \(p\ge3\), the main new difficulty is to establish a lower bound for the projected interaction uniformly in the cutoff \(N\). The low-energy interaction contains effective kernels
\(f_l^{\varepsilon,N}\) of every order \(2\le l\le p-1\). Although the corresponding lower-order terms are controlled by the positive \(p\)-body contribution in Lemma~\ref{lem 4.1}, this argument cannot be applied directly since \(P_N\) does not commute with multiplication by the interaction kernels.
We therefore separate the contributions with pairwise distinct particle labels from those with repeated labels, and use a partition of unity to obtain the uniform lower bound in Lemma~\ref{lem 5.5}.

\begin{lem}\label{lem:upper bound}
	In the setting of Lemma \ref{Lem:lower bound}, when $\lambda$ is sufficiently small, there exists a positive constant $C$ depending only on $p$, such that 
	\begin{align}
		-\log\frac{Z_{\lambda}}{Z_0}\leq&-\log\int e^{-W_p^{\varepsilon,N}(u)}\ud\mu_{0}(u)
		+C|\log\lambda|^{p^2}((\log N)^{p-2}+1)
		(N^{-1/4}\varepsilon^{3-3p}+\lambda^\delta)\,.\label{2.1}
	\end{align}
\end{lem}

For $N\geq1$, define the interaction $\widetilde\WW_N$ on $\cF$ as
\begin{align}
	\widetilde\WW_N:=&\WW_N-\lambda^{\frac14}(\lambda\ud\Gamma(P_N^2))^{p-1}\,,\label{def:quantum-W2}
\end{align}
where $\WW_N$ is given in \eqref{def:W_N}, and $P_N=\fF^{-1}(\widehat\rho(N^{-1}\cdot)\fF)$.
We first establish a uniform lower bound for \(\widetilde\WW_N\) with respect to $\varepsilon$. Besides being used in the estimates below, this bound ensures that the low-energy interacting Gibbs state in our trial state, which we introduce later, is well-defined.

\begin{lem}\label{lem 5.5}
	Let $\widetilde\WW_N$ be defined in \eqref{def:quantum-W2} and assume $0<\eta<1/(2p-2)$. There exists a constant $C>0$, such that for every $0<\varepsilon\leq 1/(4R)$, $\varepsilon\geq\lambda^\eta$ and $N\in\NN$,
	\begin{align}\label{lower bound:WN}
		\widetilde\WW_N\geq-C|\log\lambda|^p.
	\end{align}
	Here $R$ is given in \eqref{def:R}.
\end{lem}

\begin{proof}
	We prove the bound on each \(n\)-particle sector $\cH^n$. Throughout the proof, the constants are independent of \(n,N,\varepsilon\), and \(\lambda\).
	
	\medskip
	\noindent\textbf{Step 1. Reduction to operators with pairwise distinct labels.}
	Recall the definition of $\ud\Gamma_n(T)$ in \eqref{def:Gamma_n} for one-body operator $T$.
	Then on $\cH^n$, $n\geq1$, the truncated interaction $\WW_N$ defined in \eqref{def:W_N} can be written as
	\begin{align}\label{WN}
		\WW_N|_{\cH^n}=\sum_{l=2}^p\lambda^l\sum_{k_{1:l}}\fF(F^{\varepsilon,N}_l)(k_{1:l})
		\prod_{j=1}^{l}\Wick{\ud\Gamma_n(e_{k_j}^-)} + f_1^{\varepsilon,N}\lambda\Wick{\ud\Gamma_n(P_N^2)} + f_0^{\varepsilon,N}.
	\end{align} 
	For simplicity, we denote by $\boldsymbol i$ an ordered
	$k$-tuple $(i_1,\ldots,i_k)$, and for a set $S$, $S_{\neq}^k$ denotes the set of ordered $k$-tuples of pairwise distinct elements of $S$.
	For $2\leq l\leq p$ and every real-valued function $F\in \cH^l$ with $\sum_{k\in\ZZ^{2l}}|\widehat F(k)|<\infty$, set
	\begin{align*}
		B_l(F):=\sum_{k_{1:l}\in(\ZZ^2)^l}\widehat F(k_{1:l})\prod_{j=1}^l\ud\Gamma_n(e_{k_j}^-)-\sum_{\ii\in[n]^l_{\neq}}(P_N)_{\ii}F(x_{i_{1:l}})(P_N)_{\ii},
	\end{align*}
	where $(P_N)_{\ii}:=\prod_{j=1}^l(P_N)_{i_j}$. 
	Since \(F\) is real-valued and symmetric, \(B_l(F)\) is self-adjoint. It holds that
	\begin{align*}
		\sum_{\boldsymbol i\in[n]^l_{\neq}}(P_N)_{\boldsymbol i}
		F(x_{i_{1:l}})(P_N)_{\boldsymbol i} =
		\sum_{k_{1:l}\in(\mathbb Z^2)^l}
		\widehat F(k_{1:l})\sum_{\boldsymbol i\in[n]_{\neq}^l}
		(P_N)_{\boldsymbol i}\prod_{j=1}^l(e_{k_j})_{i_j}
		(P_N)_{\boldsymbol i}.
	\end{align*}
Since the operators in the
different particle variables commute, all contributions with pairwise distinct particle labels cancel, and
\[B_l(F)
=\sum_{k_{1:l}\in(\mathbb Z^2)^l}\widehat F(k_{1:l})
\sum_{\boldsymbol i\in[n]^l\backslash[n]^l_{\neq}}\prod_{j=1}^l(e_{k_j}^-)_{i_j}.\]
Every tuple in \([n]^l\setminus[n]^l_{\neq}\) contains at most \(l-1\) distinct labels. Hence \(B_l(F)\) is an operator of order at most \(l-1\).
Then, using \(\|e_k\|_{\cH}=1\) and $\|P_N\|\leq1$, where $\|\cdot\|$ denotes the operator norm on $\cH$, we obtain
	\begin{align}\label{B_l(F)}
		\pm B_l(F)\lesssim\sum_{r=1}^{l-1}\ud\Gamma_n(P_N^2)^r\sum_{k\in\ZZ^{2l}}|\widehat F(k)|\lesssim(\ud\Gamma_n(P_N^2)^{l-1}+1)\sum_{k\in\ZZ^{2l}}|\widehat F(k)|.
	\end{align}
	Since $P_N$ is a Fourier multiplier, by translation invariance,
	\begin{align}\label{N_0}
		\Tr(e_k^-\Gamma_0^{(1)})=\1_{k=0}\Tr(P_N^2\Gamma_0^{(1)})=:\1_{k=0}\widetilde N_0\leq\1_{k=0}\lambda^{-1}|\log\lambda|.
	\end{align} 
Then, by \eqref{4.2}, we have
	\begin{align}
		\sum_{k_{1:l}}\fF(F_l^{\varepsilon,N})(k_{1:l})\prod_{j=1}^l\Wick{\ud\Gamma_n(e_{k_j}^-)}
		=&\sum_{r=0}^l\binom lr(-\widetilde N_0)^{l-r}\sum_{\ii\in[n]^r_{\neq}}(P_N)_{\ii}(\Pi_{r+1}F_l^{\varepsilon,N})(x_{i_{1:r}})(P_N)_{\ii} + B,\label{5.30}
	\end{align}
where $$B=\sum_{r=2}^l\binom lr(-\widetilde N_0)^{l-r}B_r(\Pi_{r+1}F_l^{\varepsilon,N}).$$
Recall that $F_l^{\varepsilon,N}=\sum_{\fm\in\fM_{p,l}}C_\fm F_{l,\fm}^{\varepsilon,N}$, where $F_{l,\fm}^{\varepsilon,N}$ is defined in \eqref{def:Fl,m}.
Since $\fF(F_{l,\fm}^{\varepsilon,N})\geq0$ by \eqref{f:positive Fourier}, for each $2\leq r\leq l$, we have
$$\sum_{k\in\ZZ^{2r}}|\fF(\Pi_{r+1}F^{\varepsilon,N}_l)(k)|\lesssim \sum_{\fm\in\fM_{p,l}}\sum_{k\in\ZZ^{2r}}\fF(\Pi_{r+1}F^{\varepsilon,N}_{l,\fm})(k)=\sum_{\fm\in\fM_{p,l}}\Pi_{r+1}F_{l,\fm}^{\varepsilon,N}(0).$$
Then, combining with \eqref{B_l(F)} and \eqref{N_0}, it follows that
$$\pm B\lesssim \sum_{r=2}^l(\lambda^{-1}|\log\lambda|)^{l-r}(\ud\Gamma_n(P_N^2)^{r-1}+1)\left(\sum_{\substack{\fm\in\fM_{p,l}\\C_\fm\neq0}}\Pi_{r+1}F_{l,\fm}^{\varepsilon,N}(0)\right).$$
For $2\leq r\leq l$, by \eqref{F:bound} and \eqref{v:bound},
$$\Pi_{r+1}F_{l,\fm}^{\varepsilon,N}(0)=\Pi_rf_{l,\fm}^{\varepsilon,N}(0)\lesssim\varepsilon^{2-2l}|\log\lambda|^{p-l}\int_{\Lambda^{l-r}}\1_{\{d(y_i)\leq R\varepsilon,\forall i\}}\ud y_{r:l-1}\lesssim\varepsilon^{2-2r}|\log\lambda|^{p-l}.$$
Therefore, in the regime $\varepsilon\geq\lambda^\eta$ with $0<\eta<1/(2p-2)$, we have
$$\pm \,\lambda^l\,B\lesssim \lambda\sum_{r=2}^l|\log\lambda|^{p-r}(\lambda\ud\Gamma_n(P_N^2))^{r-1}\varepsilon^{2-2r} + 1 \lesssim (\lambda\ud\Gamma_n(P_N^2))^{p-1} +1.$$
Substituting into \eqref{5.30}, by \eqref{WN}, it follows that
\begin{align*}
	\WW_N|_{\cH^n}\geq&\sum_{l=2}^p\sum_{r=0}^l\binom lr(-\lambda\widetilde N_0)^{l-r}\lambda^r\sum_{\ii\in[n]^r_{\neq}}(P_N)_{\ii}(\Pi_{r+1}F_l^{\varepsilon,N})(x_{i_{1:r}})(P_N)_{\ii} + f_1^{\varepsilon,N}\lambda\Wick{\ud\Gamma_n(P_N^2)} + f_0^{\varepsilon,N}\\
	&-C((\lambda\ud\Gamma_n(P_N^2))^{p-1} +1).
\end{align*}

\noindent\textbf{Step 2. Expansion of the projected second-quantized interaction.}
Let $\vartheta^N$, $E_0^N$ and $R_l^N(v^\varepsilon)$ be in \eqref{def:R_l} and \eqref{def:nu} with $N_0$ and $f_l^\varepsilon$ replaced by $\widetilde N_0$ and $f_l^{\varepsilon,N}$ respectively. For convenience, we set $R_p^N(v^\varepsilon):=\frac1p v^\varepsilon$. Since $\Pi_{m+1}F_l^{\varepsilon,N}=\Pi_mf_l^{\varepsilon,N}$ for every $0\leq m\leq l$ (with $\Pi_0 f_l^{\varepsilon,N}$ understood as the full integral of $f_l^{\varepsilon,N}$), the interaction in \eqref{def:quantum-W2} has the lower bound
	\begin{align}
		\WW_N|_{\cH^n}\geq& \sum_{l=2}^{p}\lambda^l\sum_{\ii\in[n]^l_{\neq}}(P_N)_{\ii} R_l^N(v^\varepsilon)(x_{i_2}-x_{i_1},\cdots,x_{i_l}-x_{i_1})(P_N)_{\ii}-\lambda(\vartheta^N+1)\ud\Gamma_n(P_N^2)+ E_0^N\nonumber\\
		&-C((\lambda\ud\Gamma_n(P_N^2))^{p-1}+1)\,.\label{5.31}
	\end{align}
By \eqref{C3_2}, the estimates \eqref{F:bound} and \eqref{f_0} remain valid, uniformly in \(N\), with \(f_l^\varepsilon\) replaced by \(f_l^{\varepsilon,N}\). Consequently, using \eqref{N_0}, we obtain the same bound as \eqref{R} and \eqref{R1}:
	\begin{align}\label{6.28}
		|E_0^N|\lesssim(|\log\lambda|+|\log\varepsilon|)^p,\quad |\vartheta^N|\lesssim(|\log\lambda|+|\log\varepsilon|)^{p-1},\quad |R_l^N(v^\varepsilon)|\lesssim v_l^\varepsilon, \, 2\leq l\leq p-1.
	\end{align}
For $2\leq l\leq p$, we define
\begin{align*}
	\fV_{\varepsilon,l}^{(n)}[P_N]
	:=\sum_{\boldsymbol i\in[n]^l_{\neq}}
	(P_N)_{\boldsymbol i}\, v_l^\varepsilon(x_{i_2}-x_{i_1},\cdots,x_{i_l}-x_{i_1})\, (P_N)_{\boldsymbol i}.
\end{align*}
Then, in the regime $\varepsilon\geq\lambda^\eta$ with $0<\eta<1/(2p-2)$, \eqref{5.31} becomes
	\begin{align}
		\WW_N|_{\cH^n}\geq&\frac{\lambda^p}{p}\fV^{(n)}_{\varepsilon,p}[P_N]- C\sum_{l=2}^{p-1}\lambda^l\fV^{(n)}_{\varepsilon,l}[P_N]
		-C\left((\lambda\ud\Gamma_n(P_N^2))^{p-1}+
		|\log\lambda|^{p-1}\lambda \ud\Gamma_n(P_N^2)+|\log\lambda|^p\right).\label{4.8_1}
	\end{align}
We claim that analogues of \eqref{lwb1} and \eqref{lwb2} hold after insertion of $P_N$.
More precisely, for every \(2\le l\le p-1\), $n\geq 1$ and $C_1>0$, in the regime
$\varepsilon\geq\lambda^\eta$ with $0<\eta<1/(2p-2)$, there exists
\(C_2>0\), independent of $\varepsilon$, $n$ and $\lambda$, such that
\begin{align}
	\lambda^p\fV_{\varepsilon,p}^{(n)}[P_N]
	-C_1\lambda^l\fV_{\varepsilon,l}^{(n)}[P_N]
	\ge-C_2|\log\lambda|^p.\label{lwb3}
\end{align}
Also, there exists $c>0$ and $C>0$, independent of $\varepsilon$, $n$ and $\lambda$, such that
\begin{align}
	\lambda^p\fV_{\varepsilon,p}^{(n)}[P_N]\geq c\lambda^p\ud\Gamma_n(P_N^2)^p-C.\label{lwb4}
\end{align}
For $p=2$, the lower-order sum in \eqref{4.8_1} is empty. For $p\geq3$, apply \eqref{lwb3} to each of its $p-2$ terms. In either case,

\begin{align*}
	\WW_N|_{\cH^n}\ge\frac{1}{2p}\lambda^p\fV_{\varepsilon,p}^{(n)}[P_N]
	-C\left((\lambda\ud\Gamma_n(P_N^2))^{p-1}+
	|\log\lambda|^{p-1}\lambda\ud\Gamma_n(P_N^2)+|\log\lambda|^p\right).
\end{align*}
By \eqref{lwb4} and Young's inequality, we obtain \eqref{lower bound:WN}.

\noindent\textbf{Step 3. Proof of the claim.}
Now we prove \eqref{lwb3} and \eqref{lwb4} for $n\geq 1$.
The difficulty is that the operator $P_N$ does not commute with multiplication by the interaction potential, so the pointwise lower bounds \eqref{lwb1} and \eqref{lwb2} cannot be applied directly to the projected interaction. To overcome this, we localize the variables in each $n$-particle space using a quadratic partition of unity $P_N^2+Q_N^2=1$. 
More precisely, since $0\leq P_N\leq \1_{\cH}$, we define $Q_N:=(\1_{\cH}-P_N^2)^{1/2}$.
For $S\subset[n]$, define the operator on $\cH^n$: 
\begin{align*}
	K_S^{[n]}:=\prod_{j\in S}(P_N)_j\prod_{j\in[n]\backslash S}(Q_N)_j.
\end{align*}
Since \(Q_N\) is a Borel function of \(P_N\), all factors occurring in \(K_S^{[n]}\) commute and \(K_S^{[n]}\) is self-adjoint. Moreover, we have the exact partition of the identity
\begin{align}
	\sum_{S\subset[n]}(K_S^{[n]})^2
	=\sum_{S\subset[n]}\prod_{j\in S}(P_N)_j^2\prod_{j\notin S}(Q_N)_j^2
	&=\prod_{j=1}^n((P_N)_j^2+(Q_N)_j^2)=\1_{\cH^n}.\label{eq:partition-identity}
\end{align}
More generally, for every $T\subset[n]$, if we set $(P_N^2)_T:=\prod_{j\in T}(P_N^2)_j$, then
\begin{align}
	\sum_{S\supset T}(K_S^{[n]})^2
	&=(P_N^2)_T\sum_{U\subset [n]\backslash T}
	\prod_{j\in U}(P_N)_j^2\prod_{j\in T^c\setminus U}(Q_N)_j^2=(P_N^2)_T\prod_{j\in [n]\backslash T}((P_N)_j^2+(Q_N)_j^2)
	=(P_N^2)_T. \label{eq:subset-identity}
\end{align}
We claim that, for every $2\leq l\leq p$,
\begin{equation}
	\fV_{\varepsilon,l}^{(n)}[P_N]
	=\sum_{S\subset[n]}K_S^{[n]}\,\fV_{\varepsilon,l,\neq}^{(S)}\,K_S^{[n]}, \label{eq:exact-distinct-localization}
\end{equation}
where 
\begin{align}\label{5.36}
	\fV_{\varepsilon,l,\neq}^{(S)}=\sum_{\ii\in S^l_{\neq}}v_l^\varepsilon(x_{i_2}-x_{i_1},\cdots,x_{i_l}-x_{i_1})=:\sum_{\ii\in S^l_{\neq}}v_{l,\ii}^\varepsilon.
\end{align}
To prove this identity, expand the right-hand side of \eqref{eq:exact-distinct-localization} and interchange the finite sums:
\begin{align}
	\sum_{S\subset[n]}K_S^{[n]}\,\fV_{\varepsilon,l,\neq}^{(S)}\,K_S^{[n]}
	&=\sum_{S\subset[n]}
	\sum_{\boldsymbol i\in S_{\neq}^l}
	K_S^{[n]}\,v_{l,\ii}^\varepsilon\,K_S^{[n]}
	=\sum_{\boldsymbol i\in[n]_{\neq}^l}
	\sum_{S\supset \underline{\ii}}
	K_S^{[n]}\,v_{l,\ii}^\varepsilon\,K_S^{[n]}.\label{5.37}
\end{align}
where \(\underline{\boldsymbol i}:=\{i_1,\ldots,i_l\}\).
Since the labels in $\ii$ are pairwise distinct, for every set $S\supset\underline{\ii}$,
$$K_S^{[n]}=(P_N)_{\ii}\, K_{S\backslash\underline{\ii}}^{[n]\backslash\underline{\ii}}\,.$$
The second factor acts only on the particle variables outside
\(\underline{\boldsymbol i}\), and hence commutes with both
$v_{l,\ii}^\varepsilon$ and $(P_N)_{\ii}$.
Therefore, by \eqref{eq:partition-identity} applied on
\([n]\setminus\underline{\boldsymbol i}\), we have
\begin{align*}
	\sum_{S\supset \underline{\ii}}
	K_S^{[n]}\,v_{l,\ii}^\varepsilon\,K_S^{[n]} = (P_N)_{\ii}\,v_{l,\ii}^\varepsilon\,(P_N)_{\ii} \sum_{U\subseteq [n]\backslash\underline{\ii}}(K_U^{[n]\backslash\underline{\ii}})^2=(P_N)_{\ii}\,v_{l,\ii}^\varepsilon\,(P_N)_{\ii}.
\end{align*}
Substituting this identity into \eqref{5.37} proves
\eqref{eq:exact-distinct-localization}.

Let $\fV_{\varepsilon,l}^{(S)}$ be defined as in \eqref{5.36} with $S_{\neq}^l$ replaced by $S^l$.
The proofs of \eqref{lwb1} and \eqref{lwb2} apply to every subset $S\subseteq[n]$ and show that, for $2\leq l\leq p-1$ and $C_1>0$, there are constants $C_2,c>0$ such that

\begin{align}
	\lambda^p\fV_{\varepsilon,p}^{(S)}
	-C_1\lambda^l\fV_{\varepsilon,l}^{(S)}
	\ge-C_2|\log\lambda|^p,\quad \text{and}
\quad \fV_{\varepsilon,p}^{(S)}\geq c|S|^p.\label{lwb2'}
\end{align}
Since $v_l^\varepsilon\geq0$, we have $\fV_{\varepsilon,l,\neq}^{(S)}\leq \fV_{\varepsilon,l}^{(S)}$, and by \eqref{periodiz}, under the regime $\varepsilon\geq\lambda^\eta$ with $0<\eta<1/(2p-2)$, 
$$\lambda^p(\fV_{\varepsilon,p}^{(S)}-\fV_{\varepsilon,p,\neq}^{(S)})\leq \lambda^p\|v^\varepsilon\|_{L^\infty}\sum_{\ii\in S^p\backslash S^p_{\neq}}\1_{\cH^n} \lesssim\lambda\varepsilon^{2-2p}(\lambda|S|)^{p-1}\lesssim(\lambda|S|)^{p-1}.$$
Therefore, by \eqref{lwb2'} and Young's inequality, we have
$$\lambda^p\fV_{\varepsilon,p,\neq}^{(S)}
-C_1\lambda^l\fV_{\varepsilon,l,\neq}^{(S)}\geq\frac12\left(\lambda^p\fV_{\varepsilon,p}^{(S)}
-2C_1\lambda^l\fV_{\varepsilon,l}^{(S)}\right)+\frac12\lambda^p\fV_{\varepsilon,p}^{(S)}-C(\lambda|S|)^{p-1}
\gtrsim-|\log\lambda|^p,$$
and $$\lambda^p\fV_{\varepsilon,p,\neq}^{(S)}\geq\lambda^p\fV_{\varepsilon,p}^{(S)}-C(\lambda|S|)^{p-1}\geq\frac c2(\lambda|S|)^p-C.$$
Then, by \eqref{eq:exact-distinct-localization}, it follows that
\begin{align*}
	\lambda^p\fV_{\varepsilon,p}^{(n)}[P_N]
	-C_1\lambda^l\fV_{\varepsilon,l}^{(n)}[P_N] = \sum_{S\subset[n]}K_S^{[n]}\,(\lambda^p\fV_{\varepsilon,p,\neq}^{(S)}-C_1\lambda^l\fV_{\varepsilon,l,\neq}^{(S)})\,K_S^{[n]}\gtrsim-|\log\lambda|^p,
\end{align*}
Thus, \eqref{lwb3} holds.
By \eqref{eq:partition-identity}, we have
\begin{align}
	\lambda^p\fV_{\varepsilon,p}^{(n)}[P_N]=\sum_{S\subset[n]}K_S^{[n]}\,(\lambda^p\fV_{\varepsilon,p,\neq}^{(S)})\,K_S^{[n]}\geq \frac c2\lambda^p\sum_{S\subset[n]}|S|^p(K_S^{[n]})^2 - C.\label{5.42}
\end{align}
Moreover, since $0\leq P_N\leq1$,
\eqref{eq:subset-identity} gives
$$\sum_{S\subset[n]}|S|^p(K_S^{[n]})^2=\sum_{\ii\in[n]^p}\sum_{S\supset\{i_1,\cdots,i_p\}}(K_S^{[n]})^2=\sum_{\ii\in[n]^p}\prod_{j\in\{i_1,\ldots,i_p\}}(P_N^2)_j\geq(\ud\Gamma_n(P_N^2))^p.$$
In the last inequality, repeated indices in the expansion of
$(\ud\Gamma_n(P_N^2))^p$ produce higher powers of $(P_N^2)_j$, which are bounded above by $(P_N^2)_j$.
Substituting into \eqref{5.42}, we obtain \eqref{lwb4}.
\end{proof}

We now return to the proof of Lemma~\ref{lem:upper bound}.
Let $P:=\1_{\{h\leq(2\pi)^2N^2+1\}}$ and $Q=\1_{\cH}-P$. Since \(P_NP=PP_N=P_N\), the restriction
\[\mathbb W_P
:=\widetilde\WW_N\big|_{\mathfrak F(P\mathfrak H)}\]
satisfies the same lower bound as in \eqref{lower bound:WN}. Hence the quadratic form
\[\mathbb H_{\lambda,P}:=\lambda\ud\Gamma(Ph)+\mathbb W_P\]
is bounded from below and admits a Friedrichs extension.
Using the unitary $\fU$ in \eqref{eq:Fock factor}, we define a trial state on $\cF$:  
$$\widehat\Gamma_\lambda = \fU^*(\Gamma_{\lambda,P}\otimes(\Gamma_0)_Q)\fU,$$ 
where $(\Gamma_0)_Q$ is the $Q$-localization of the Gaussian state and $\Gamma_{\lambda,P}$ is the interacting Gibbs state on $\cF(P\cH)$:
$$\Gamma_{\lambda,P} =\frac{e^{-\mathbb H_{\lambda,P}}}
{\Tr_{\mathfrak F(P\mathfrak H)}(e^{-\mathbb H_{\lambda,P}})}.$$

\begin{proof}[Proof of Lemma \ref{lem:upper bound}]
Under the factorization induced by
\(\mathfrak H=P\mathfrak H\oplus Q\mathfrak H\), the Gaussian state satisfies
\[
\Gamma_0
=
\mathcal U^*
\bigl(
(\Gamma_0)_P\otimes(\Gamma_0)_Q
\bigr)
\mathcal U.
\]
Then, by \cite[Lemma 10.3]{LNR21}, we have
\[\mathcal H(\widehat\Gamma_\lambda,\Gamma_0)
=\mathcal H(\Gamma_{\lambda,P},(\Gamma_0)_P).\]
Applying the variational principle \eqref{variational problem} with
\(\Gamma=\widehat\Gamma_\lambda\), we obtain
\begin{align}\label{4.9}
	-\log\frac{Z_{\lambda}}{Z_0}\leq \fH(\Gamma_{\lambda,P},(\Gamma_0)_P)+\Tr(\WW\widehat\Gamma_\lambda).
\end{align}

First, we estimate the expectation of $\WW-\WW_N$ with respect to $\widehat\Gamma_\lambda$ similarly as Step 1 in the proof of Lemma \ref{Lem:lower bound}. 
By the variational formula for $\widehat\Gamma_\lambda$, we have
\begin{align*}
	&\fH(\widehat\Gamma_\lambda,\Gamma_0)
	+\Tr\!\left(\fU^*(\WW_P\otimes\1)\fU\,\widehat\Gamma_\lambda\right)\leq
	\Tr\!\left(\fU^*(\WW_P\otimes\1)\fU\,\Gamma_0\right)
	=\Tr(\WW_P(\Gamma_0)_P).
\end{align*}
Using \eqref{lower bound:WN} and $\WW_P\leq\WW_N|_{\cF(P\cH)}$, we infer
\begin{align}\label{eq:trial-entropy}
	\fH(\widehat\Gamma_\lambda,\Gamma_0)
	\leq C|\log\lambda|^p+
	\Tr(\WW_N(\Gamma_0)_P).
\end{align}
As in \eqref{3.7}, we have
 \begin{align*}
 	\Tr(\WW_N(\Gamma_{0})_P)=&\sum_{l=2}^p\sum_{k_{1:l}}\fF(F_l^{\varepsilon,N})(k_{1:l})\lambda^l \,\Tr\left(\prod_{j=1}^l\Wick{\ud\Gamma(e_{k_j}^-)}(\Gamma_0)_P\right) + f_0^{\varepsilon,N}\,.
 \end{align*}
Since $P_NP=PP_N=P_N$ and $|\widehat\rho|\leq1$, the estimate proceeds as before, with \eqref{3.8} replaced by
\begin{align*}
	\lambda^{|c|}\Tr_c(T^N)
	=&\lambda^{|c|}\1_{\sum_{i\in c}k_i=0}\sum_{q\in\ZZ^2}\prod_{j=1}^{|c|}\left(\lambda_{q+k_{c_1}+\cdots+k_{c_j}}\widehat\rho(N^{-1}(q+k_{c_1}+\cdots+k_{c_j}))^2\right)\\
	\leq&\1_{\sum_{i\in c}k_i=0}\sum_{q\in\ZZ^2}\prod_{j=1}^{|c|}\widehat G(q+k_{c_1}+\cdots+k_{c_j}),
\end{align*}
where $(T^N)_j = e_{k_j}^-\Gamma_0^{(1)}$ for every $1\leq j\leq l$. Thus all estimates are uniform in
$N$. By \eqref{eq:trial-entropy}, we get the same result as \eqref{entropy}: $$\fH(\widehat\Gamma_\lambda,\Gamma_0)\lesssim |\log\lambda|^p \,.$$
Moreover, the particle-number argument of Lemma~\ref{Lemma4.3} applies to the decoupled
Hamiltonian of $\widehat\Gamma_\lambda$. More precisely, replacing $h$ by $h_t=h-t|\log\lambda|^{-p}$ leaves the projected interaction unchanged;
\eqref{lower bound:WN} and the estimate remain uniform for
$|t|\leq1$. The higher-order correlation inequality used in
Lemma~\ref{Lemma4.3} therefore yields, for every fixed $m\in\NN$,
\begin{align}\label{eq:trial-moments}
	\lambda^m\Tr(\fN^m\widehat\Gamma_\lambda)
	\lesssim_m |\log\lambda|^{pm}.
\end{align}
The proof of the high-momentum estimate \eqref{1.4} applies as well with $F_l^{\varepsilon}$ replaced by $F_l^{\varepsilon,N}$, since the latter kernels
satisfy the same Fourier bounds uniformly in $N$. The only new term is the double commutator with
$-\lambda^{1/4}(\lambda\ud\Gamma(P_N^2))^{p-1}$, comparing with \eqref{4.36}.  Since
$$\|[[\ud\Gamma_n(f_k^+),\ud\Gamma_n(P_N^2)],\ud\Gamma_n(f_k^+)]\|=\|\ud\Gamma_n([[f_k^+,P_N^2],f_k^+])\|\lesssim n,$$
similarly as in \eqref{eq:double-commutator-sector-bound}, we have
\begin{align*}
\lambda^{\frac94}\|[[\ud\Gamma_n(f_k^+),(\lambda\ud\Gamma(P_N^2))^{p-1}],\ud\Gamma_n(f_k^+)]\|
	\lesssim \lambda^{\frac94}(\lambda n)^{p-1}.
\end{align*}
Therefore, the expectation of its second quantization with respect to $\widehat\Gamma_\lambda$ is dominated by the right-hand side of \eqref{1.4}, in view of \eqref{eq:trial-moments}. Consequently, Step~1 of the lower-bound proof
can be repeated for $\widehat\Gamma_\lambda$, and gives
\begin{align}\label{5.25}
	\left|\Tr\bigl((\WW-\WW_N)\widehat\Gamma_\lambda\bigr)\right|
	\lesssim |\log\lambda|^{p^2}((\log N)^{p-2}+1)
	(N^{-1/4}\varepsilon^{3-3p}+\lambda^\delta),
\end{align}
where $\delta=15\theta_0/16-(2p-2+\theta_0)\eta>0$.

Now we can reduce the problem to a finite dimensional estimate. Since $\Tr(\widetilde \WW_N\widehat\Gamma_\lambda)=\Tr(\WW_P\Gamma_{\lambda,P})$, and
\begin{align*}
	\lambda^{\frac14}\Tr((\lambda\ud\Gamma(P_N^2))^{p-1}\Gamma_{\lambda,P})=&\lambda^{\frac14}\Tr(\fU^*((\lambda\ud\Gamma(P_N^2))^{p-1}\otimes\1_{\cF(Q\cH)})\fU\widehat\Gamma_\lambda)
	\leq\lambda^{\frac14}\Tr((\lambda\fN)^{p-1}\widehat\Gamma_\lambda),
\end{align*}
by \eqref{eq:trial-moments}, \eqref{4.9} and \eqref{5.25}, we have
\begin{align}
	-\log\frac{Z_{\lambda}}{Z_0}\leq&\fH(\Gamma_{\lambda,P},(\Gamma_0)_P)+\Tr(\WW_P\Gamma_{\lambda,P})+C|\log\lambda|^{p^2}((\log N)^{p-2}+1)
	(N^{-1/4}\varepsilon^{3-3p}+\lambda^\delta)\nonumber\\
=&-\log\frac{\Tr(e^{-\lambda\ud\Gamma(Ph)-\WW_P})}{\Tr(e^{-\lambda\ud\Gamma(Ph)})} + C|\log\lambda|^{p^2}((\log N)^{p-2}+1)(N^{-1/4}\varepsilon^{3-3p}+\lambda^\delta)\,.\label{4.10} 
\end{align}
To estimate the relative partition function in the right-hand side in \eqref{4.10}, as \cite[(10.13)]{LNR21}, we use the resolution of identity \eqref{resolution of identity} and the Peierls–Bogoliubov inequality to get
\begin{align}
	\Tr(e^{-\lambda\ud\Gamma(Ph)-\WW_P})\geq\frac{1}{(\lambda\pi)^K}\int_{P\cH}\exp[-\langle W(u/\sqrt\lambda),(\lambda\ud\Gamma(Ph)+\WW_P)W(u/\sqrt\lambda)\rangle]\ud u \,,\label{5.27}
\end{align}
where $K:=\Tr(P)$.
Now we upper bound $\langle W(u/\sqrt\lambda),\WW_P W(u/\sqrt\lambda)\rangle$ for $u\in P\cH$. 
For $1\leq l\leq p$ and $k_{1:l}\in\ZZ^{2l}$, expanding $\Wick{\ud\Gamma(e_{k_j}^-)}$ for $1\leq j\leq l$ gives
\begin{align}\label{5.28}
	\Tr\left(\prod_{j=1}^{l}\Wick{\ud\Gamma(e_{k_j}^-)}|W(u/\sqrt\lambda)\rangle\langle W(u/\sqrt\lambda)|\right)
	=\sum_{Q\subseteq[l]}\prod_{i\notin Q}(-\langle\ud\Gamma(e_{k_i}^-)\rangle_0)\left\langle W(u/\sqrt\lambda),\prod_{j\in Q}\ud\Gamma(e_{k_j}^-)W(u/\sqrt\lambda)\right\rangle\,.
\end{align}
Fix $Q\subseteq[l]$ and write $q=|Q|$. When $q=0$, the coherent-state expectation in \eqref{5.28} equals $1$. 
If \(q\ge1\), by \eqref{1.8}, we have
\begin{align}
	\prod_{j\in Q}\ud\Gamma(e_{k_j}^-)
	&=\bigoplus_{n\ge q}
	\sum_{\ii\in[n]^q_{\neq}}
	\prod_{j\in Q}(e_{k_j}^-)_{i_j}
	+B_Q
	=:A_Q+B_Q ,
	\label{eq:coherent-collision}
\end{align}
where \(\ii=(i_j)_{j\in Q}\), and \([n]^q_{\neq}\) denotes the set of ordered \(q\)-tuples of pairwise distinct particle labels. 
For $q=1$, \(B_Q=0\). For \(q\ge2\), \(B_Q\) collects all terms in which at least two of the one-body operators act on the same particle. Grouping these terms according to the number of distinct particle labels, we write
\[B_Q=\sum_{r=1}^{q-1}\bigoplus_{n\ge r}\left(\sum_{\mathbf i\in[n]^r_{\neq}}
(K_r)_{i_1,\cdots,i_r}\right),\]
where \(K_r\) is a finite sum of \(r\)-body operators, and each tensor factor of these operators is a product of
\(e_{k_j}^-\) for some $j\in Q$. 
Then, for $1\leq r\leq q-1$, we have
\[\left|\langle u^{\otimes r},K_{r}u^{\otimes r}\rangle\right|\lesssim_q\|P_Nu\|_{\cH}^{2r}.\]
Consequently, by \eqref{density matrix of coherent state}, we obtain
\begin{align*}
	\lambda^q
	\left\langle
	W(u/\sqrt\lambda),A_QW(u/\sqrt\lambda)
	\right\rangle
	=\prod_{j\in Q}\langle u,e_{k_j}^-u\rangle.
\end{align*}
and
\[
\begin{aligned}
	\lambda^q\left|\left\langle
	W(u/\sqrt\lambda),B_QW(u/\sqrt\lambda)\right\rangle\right|\lesssim\sum_{r=1}^{q-1}
	\lambda^{q-r}\|P_Nu\|_{\cH}^{2r}
	\lesssim
	\lambda\left(1+\|P_Nu\|_{\cH}^{2(q-1)}\right).
\end{aligned}
\]
Combining with \eqref{eq:coherent-collision}, we conclude that

\begin{align}\label{eq:coherent-product}
	&\lambda^q\left\langle W(u/\sqrt\lambda),
	\prod_{j\in Q}\ud\Gamma(e_{k_j}^-)
	W(u/\sqrt\lambda)\right\rangle
	=\prod_{j\in Q}\langle u,e_{k_j}^-u\rangle
	+O\!\left(\lambda
	\bigl(1+\|P_Nu\|^{2(q-1)}_\cH\bigr)\right).
\end{align}

Since $\lambda|\langle\ud\Gamma(e_{k_i}^-)\rangle_0|\leq\lambda N_0\lesssim|\log\lambda|$, by \eqref{5.28} and \eqref{eq:coherent-product}, we have
\begin{align*}
	\lambda^l\left\langle W(u/\sqrt\lambda),
	\prod_{j=1}^{l}\Wick{\ud\Gamma(e_{k_j}^-)}
	W(u/\sqrt\lambda)\right\rangle
	=&\sum_{Q\subseteq[l]}\prod_{i\notin Q}(-\lambda\langle\ud\Gamma(e_{k_i}^-)\rangle_0)\prod_{j\in Q}\langle u,e_{k_j}^-u\rangle\\
	&+O\!\left(\lambda|\log\lambda|^{l-1}
	\bigl(1+\|P_Nu\|^{2(l-1)}_\cH\bigr)\right)\,.
\end{align*}
Then, by \eqref{difference} and $|\langle u,e_{k_j}^-u\rangle|\leq\|P_Nu\|^2$, the above becomes
\begin{align}\label{eq:coherent-wick-comparison}
	&\lambda^l\left\langle W(u/\sqrt\lambda),
	\prod_{j=1}^{l}\Wick{\ud\Gamma(e_{k_j}^-)}
	W(u/\sqrt\lambda)\right\rangle
	=\prod_{j=1}^l\Wick{\langle u,e_{k_j}^-u\rangle}
	+O\!\left(\lambda N^2|\log\lambda|^{l-1}
	\bigl(1+\|P_Nu\|^{2(l-1)}_\cH\bigr)\right).
\end{align}
We now sum \eqref{eq:coherent-wick-comparison} against the Fourier coefficients of the truncated interaction. For \(2\le l\leq p\), by \eqref{5.35}, the error generated by the \(l\)-th interaction term is bounded by
\[
\begin{aligned}
	&
	C\lambda N^2|\log\lambda|^{l-1}
	\left(
	1+\|P_Nu\|_{\cH}^{2(l-1)}
	\right)
	\sum_{k_{1:l}}
	\left|
	\mathcal F(F_l^{\varepsilon,N})(k_{1:l})
	\right|\lesssim
	\lambda N^2\varepsilon^{2-2p}|\log\lambda|^{p-1}
	\left(
	1+\|P_Nu\|_{\cH}^{2(p-1)}
	\right).
\end{aligned}
\]
For the one-body term, \eqref{eq:coherent-wick-comparison} with
\(l=1\) and \(k_1=0\) gives
\[
\lambda
\left\langle
W(u/\sqrt\lambda),
\Wick{\ud\Gamma(P_N^2)}
W(u/\sqrt\lambda)
\right\rangle
=
\Wick{\langle u,P_N^2u\rangle}
+
O(\lambda N^2).
\]
By \eqref{f_0}, this error is also absorbed by the above bound.
By the Parseval equality, the main terms obtained after summing over \(k_{1:l}\) are exactly the corresponding terms in \(W_p^{\varepsilon,N}(u)\). We therefore obtain
\begin{align}\label{eq:coherent-WPN}
	\left\langle
	W(u/\sqrt\lambda),
	\WW_N|_{\cF(P\cH)}
	W(u/\sqrt\lambda)
	\right\rangle
	\le
	W_p^{\varepsilon,N}(u)
	+
	C\lambda N^2\varepsilon^{2-2p}
	|\log\lambda|^{p-1}
	\bigl(1+\|P_Nu\|_\cH^{2(p-1)}\bigr).
\end{align}
Moreover, Jensen's inequality for the spectral measure of the positive operator $\lambda\ud\Gamma(P_N^2)$ gives the exact lower bound
\begin{align*}
	&\left\langle W(u/\sqrt\lambda),
	(\lambda\ud\Gamma(P_N^2))^{p-1}W(u/\sqrt\lambda)\right\rangle\geq
	\left\langle W(u/\sqrt\lambda),
	\lambda\ud\Gamma(P_N^2)W(u/\sqrt\lambda)\right\rangle^{p-1}
	=\|P_Nu\|^{2(p-1)}_\cH.
\end{align*}
Combining with \eqref{eq:coherent-WPN}, we have
\begin{align*}
	\langle W(u/\sqrt\lambda),\WW_PW(u/\sqrt\lambda)\rangle
	\leq& W_p^{\varepsilon,N}(u)+C\lambda N^2\varepsilon^{2-2p}|\log\lambda|^{p-1}
	-(\lambda^{\frac14}-C\lambda|\log\lambda|^{p-1}N^2\varepsilon^{2-2p})\|P_Nu\|^{2(p-1)}_\cH\\
	\leq& W_p^{\varepsilon,N}(u)+C\lambda N^2\varepsilon^{2-2p}|\log\lambda|^{p-1}.
\end{align*}
Indeed, $N\leq\lambda^{-1/8}$ and
$\varepsilon\geq\lambda^\eta$ imply
	\[
	\lambda^{3/4}N^2\varepsilon^{2-2p}|\log\lambda|^{p-1}
	\lesssim \lambda^{1/2-2\eta(p-1)}|\log\lambda|^{p-1}=o(1),
	\]
so the coefficient of $\|P_Nu\|_{\cH}^{2(p-1)}$ is non-negative for small $\lambda$.
Substituting into \eqref{5.27} and using \cite[(10.12) and (9.15)]{LNR21}, we obtain
\begin{align*}
	\frac{\Tr(e^{-\lambda\ud\Gamma(Ph)-\WW_P})}{\Tr(e^{-\lambda\ud\Gamma(Ph)})}
	\geq&(1-C\lambda N^4)e^{-C\lambda N^2\varepsilon^{2-2p}|\log\lambda|^{p-1}}\int_{P\cH}e^{-W_p^{\varepsilon,N}(u)}\ud\mu_{0,P}(u)\,.
\end{align*}
The two additional errors are absorbed by the error in \eqref{2.1} since $\lambda N^4\leq\lambda^\frac12$ and $\lambda N^2\leq N^{-\frac14}$.
Combining the above estimate with \eqref{4.10}, we obtain \eqref{2.1}.

\end{proof}

\section{Convergence of density matrices}\label{sec6}

In this section, we prove the convergence of the reduced density matrices in \eqref{thm-eq2}-\eqref{thm-eq3}. 
As in the previous section, we set
$P=\1_{\{h\leq(2\pi)^2N^2+1\}}$ and $Q=\1_{\cH}-P$.
Throughout this section, the limit is taken along
\(\lambda,\varepsilon\to0\) with \(\varepsilon\geq\lambda^\eta\), and
\(N=\lfloor\lambda^{-1/8}\rfloor\).  We assume that
$0<\eta<\min\left\{\frac{1}{96p},
\frac{15\theta_0}{16(2p-2+\theta_0)}\right\}$, where $\theta_0:=\frac12\min\{1,\delta_0\}$ with $\delta_0$ defined in \eqref{v-fourier}.

\subsection{Hilbert--Schmidt convergence}\label{sec:6.1}

The proof of \eqref{thm-eq2} follows closely from \cite[Section 11]{LNR21}.
However, due to the lower bound of the interaction in \eqref{lower bound:W}, the bound of $\lambda^k\|\Gamma_{\lambda}^{(k)}\|_{\text{HS}}$ for any $k\in\NN$ derived as in the proof of \cite[Lemma 11.3]{LNR21} is not sufficient. Hence, we instead use \eqref{0.2} to obtain
\begin{align}
	\lambda^k\|\Gamma_{\lambda}^{(k)}\|_{\text{HS}}\leq\lambda^k\Tr(\Gamma_{\lambda}^{(k)})\leq\lambda^k \Tr(\fN^k\Gamma_{\lambda})\lesssim_{k} |\log\lambda|^{pk}.\label{HS-bound}
\end{align}
Now we are ready to give the proof of \eqref{thm-eq2}.

\begin{proof}[Proof of \eqref{thm-eq2}]
	Let 
	$$\widetilde\Gamma_{\lambda} = \fU^*((\Gamma_{\lambda})_P\otimes (\Gamma_{\lambda})_Q)\fU \,.$$
	By \cite[(11.20)]{LNR21} and \eqref{HS-bound}, the reduced density matrices of \(\widetilde\Gamma_\lambda\) satisfy the same bound as in \eqref{HS-bound}.
	Consequently, \cite[Lemma~11.4]{LNR21} and Pinsker's inequality (e.g. \cite[(6.1)]{LNR21}) imply
	\begin{align}\label{2.2}
		\lambda^k\|\Gamma_{\lambda}^{(k)}-\widetilde\Gamma_{\lambda}^{(k)}\|_{\text{HS}}\lesssim_{k}|\log\lambda|^{pk}(\Tr|\Gamma_{\lambda}-\widetilde\Gamma_{\lambda}|)^{\frac12}\lesssim |\log\lambda|^{pk}\fH(\Gamma_{\lambda},\widetilde\Gamma_{\lambda})^{\frac14}\,.
	\end{align}
	For the relative entropy $\fH(\Gamma_\lambda,\widetilde\Gamma_\lambda)$,
	by \cite[(11.6)]{LNR21} and \eqref{5.22}, we have
	\begin{align*}
		-\log\frac{Z_{\lambda}}{Z_0}\geq&
		\fH(\Gamma_\lambda,\widetilde\Gamma_\lambda)+\fH((\Gamma_\lambda)_P,(\Gamma_0)_P)+\fH((\Gamma_\lambda)_Q,(\Gamma_0)_Q)+\int_{P\cH} W_p^{\varepsilon,N}\ud\mu_{P,\lambda}^\lambda
		-CR(\lambda,\varepsilon,N),
	\end{align*}
where 
\begin{align}\label{R(lambda,N)}
	R(\lambda,\varepsilon,N)=|\log\lambda|^{p^2}((\log N)^{p-2}+1)(N^{-1/4}\varepsilon^{3-3p}+\lambda^\delta)=o(|\log\lambda|^{-m})
\end{align}
for any $m\geq0$ due to \(N=\lfloor\lambda^{-1/8}\rfloor\), $\varepsilon\geq\lambda^\eta$ with $0<\eta<\frac{1}{96p}$ and $\delta>0$.
Then, by the Berezin-Lieb inequality \eqref{eq:Berezin-Lieb} and \eqref{5.23}, the above is lower bounded by
\begin{align*}
	\fH(\Gamma_\lambda,\widetilde\Gamma_\lambda)+\fH((\Gamma_\lambda)_Q,(\Gamma_0)_Q)+\fH_{\text{cl}}(\mu_{P,\lambda}^\lambda,\widetilde\mu_p^{\varepsilon,N})-\log\int_{P\cH}e^{-W_p^{\varepsilon,N}(u)}\ud\mu_{0,P}(u)-CR(\lambda,\varepsilon,N)\,.
\end{align*}
Here we also used the equality $$\fH_{\text{cl}}(\mu_{P,\lambda}^\lambda,\mu_{0,P})+\int_{P\cH} W_p^{\varepsilon,N}\ud\mu_{P,\lambda}^\lambda
=\fH_{\text{cl}}(\mu_{P,\lambda}^\lambda,\widetilde\mu_p^{\varepsilon,N})-\log\int_{P\cH}e^{-W_p^{\varepsilon,N}(u)}\ud\mu_{0,P}(u),$$
where $\widetilde\mu_p^{\varepsilon,N}=\mu_p^{\varepsilon,N}\circ P^{-1}$ with $\mu_p^{\varepsilon,N}$ defined in \eqref{def:mu_p^N}.
Together with the upper bound \eqref{2.1}, this gives
	\begin{align}
		\fH(\Gamma_{\lambda},\widetilde\Gamma_{\lambda})+\fH((\Gamma_{\lambda})_Q,(\Gamma_0)_Q)+\fH_{\text{cl}}(\mu_{P,\lambda}^\lambda,\widetilde\mu_p^{\varepsilon,N})\lesssim R(\lambda,\varepsilon,N)\,.\label{2.3}
	\end{align} 
The estimate \eqref{2.3} is the key link between free-energy convergence and convergence of the correlation functions.
	In particular, \eqref{2.2} becomes
	\begin{align}
		\lambda^{4k}\|\Gamma_\lambda^{(k)}-\widetilde\Gamma_\lambda^{(k)}\|^4_{\text{HS}}\lesssim&_k|\log\lambda|^{4pk}R(\lambda,\varepsilon,N)\,.\label{6.3}
	\end{align}

It therefore remains to analyze $\widetilde\Gamma_\lambda^{(k)}$. 
By \cite[(11.20)]{LNR21}, it suffices to prove that the P-localized term $(\Gamma_\lambda)_P^{(k)}$ converges to the desired limit in the Hilbert-Schmidt norm, while the Q-localized term $(\Gamma_\lambda)_Q^{(l)}$ converges to 0 for any $l\leq k$.
	We first consider the $P$-localized term. By \eqref{0.2}, we have from \eqref{eq:quantitative}:
	\begin{align}
		\left\|k!\,\lambda^k(\Gamma_\lambda)_P^{(k)}-\int_{P\cH}|u^{\otimes k}\rangle \langle u^{\otimes k}|\ud\mu_{P,\lambda}^\lambda(u)\right\|_{\text{HS}}
		\lesssim_k \lambda^k\sum_{l=0}^{k-1}d^{k-l}\Tr(\fN^l\Gamma_{\lambda})
		\lesssim_k\lambda N^2|\log\lambda|^{pk},\label{2.4}
	\end{align}
where $d=\Tr(P)\lesssim N^2.$
This formula, together with \eqref{HS-bound}, gives
\begin{align}\label{6.12}
	\left\|\int_{P\cH}|u^{\otimes k}\rangle\langle u^{\otimes k}|\ud\mu_{P,\lambda}^\lambda(u)\right\|_{\text{HS}}\lesssim_{k}\lambda N^2|\log\lambda|^{pk}+\lambda^k\Tr(\Gamma_\lambda^{(k)})
	\lesssim |\log\lambda|^{pk}.
\end{align} 
Set $$\fM_k:=\int|u^{\otimes k}\rangle\langle u^{\otimes k}|
\ud\mu_p(u).$$
By \eqref{eq:moment-map-bound} for $g=\cZ_p^{-1}e^{-V_p}$, we have $\fM_k\in\mathfrak S^2$. Since \(P\to\1\) strongly, combining with \eqref{matrices conv}, we have
\begin{align}
	&\left\|\int_{P\cH}|u^{\otimes k}\rangle\langle u^{\otimes k}|
	\ud\widetilde\mu_p^{\varepsilon,N}(u)
	-\fM_k\right\|_{\mathfrak S^2}
	\leq\left\|\int|u^{\otimes k}\rangle\langle u^{\otimes k}|\ud\,(\mu_p^{\varepsilon,N}(u)-\mu_p(u))\right\|_{\mathfrak S^2}
	+\left\|P^{\otimes k}\mathcal M_k P^{\otimes k}
	-\mathcal M_k\right\|_{\mathfrak S^2},\label{eq:classical-moment-convergence}
\end{align}
which vanishes as $\lambda\to0$.
In particular, for sufficiently small \(\lambda\),
\begin{align}\label{6.12'}
	\left\|\int_{P\cH}|u^{\otimes k}\rangle\langle u^{\otimes k}|
	\,\ud\widetilde\mu_p^{\varepsilon,N}(u)\right\|_{\mathfrak S^2}\leq 1+\|\fM_k\|_{\mathfrak S^2}\lesssim_k1.
\end{align}

Let \(\nu_\lambda:=|\mu_{P,\lambda}^\lambda-
	\widetilde\mu_p^{\varepsilon,N}|\).
Since \begin{align*}
		&\left\|\int_{P\cH}|u^{\otimes k}\rangle\langle u^{\otimes k}|(\ud\mu_{P,\lambda}^\lambda(u)-\ud\widetilde\mu_p^{\varepsilon,N}(u))\right\|_{\text{HS}}^2
		\leq\int|\langle u_1,u_2\rangle|^{2k}\ud\nu_\lambda(u_1)\ud\nu_\lambda(u_2),
	\end{align*}
we use the Cauchy-Schwarz inequality on
\(P\cH\times P\cH\) with respect to
\(\nu_\lambda\otimes\nu_\lambda\) to obtain
\begin{align*}
	\left\|\int_{P\cH}|u^{\otimes k}\rangle\langle u^{\otimes k}|(\ud\mu_{P,\lambda}^\lambda(u)-\ud\widetilde\mu_p^{\varepsilon,N}(u))\right\|_{\text{HS}}^2
	\leq
	\nu_\lambda(P\cH)
	\left\|
	\int_{P\cH}
	|u^{\otimes 2k}\rangle\langle u^{\otimes 2k}|\,
	\ud\rho_\lambda(u)
	\right\|_{\mathfrak S^2},
\end{align*}
where $\rho_\lambda=\mu_{P,\lambda}^\lambda+
\widetilde\mu_p^{\varepsilon,N}$.
The classical Pinsker inequality and \eqref{2.3} imply
\begin{align}\label{6.13}
	\nu_\lambda(P\cH)^2
	&\lesssim
	\fH_{\mathrm{cl}}
	\bigl(\mu_{P,\lambda}^{\lambda},
	\widetilde\mu_p^{\varepsilon,N}\bigr)
	\lesssim R(\lambda,\varepsilon,N).
\end{align}
Combining this estimate with the preceding moment estimates in \eqref{6.12} and \eqref{6.12'} with \(k\) replaced by \(2k\), we conclude that
\begin{align}\label{6.14}
	\left\|\int_{P\cH}|u^{\otimes k}\rangle\langle u^{\otimes k}|(\ud\mu_{P,\lambda}^\lambda(u)-\ud\widetilde\mu_p^{\varepsilon,N}(u))\right\|_{\text{HS}}^2
	&\lesssim_k
	|\log\lambda|^{2pk}R(\lambda,\varepsilon,N)^{1/2},
\end{align}
which vanishes by \eqref{R(lambda,N)}.
Combining with \eqref{2.4} and \eqref{eq:classical-moment-convergence}, and using
\(\lambda N^2|\log\lambda|^{pk}\to0\), we obtain
\begin{align}\label{eq:P-localized-convergence}
	\left\|k!\lambda^k(\Gamma_\lambda)_P^{(k)}
	-\int|u^{\otimes k}\rangle\langle u^{\otimes k}|
	\,\ud\mu_p(u)\right\|_{\mathfrak S^2}
	\longrightarrow0.
\end{align}

For the $Q$-localized term, apply \cite[Lemma~11.4]{LNR21} to the states
\((\Gamma_\lambda)_Q\) and \((\Gamma_0)_Q\).  By \eqref{2.3} and Pinsker's inequality, for every fixed
\(1\leq j\leq k\),
\begin{align}\label{eq:Q-comparison}
	\lambda^j
	\|(\Gamma_\lambda)_Q^{(j)}-(\Gamma_0)_Q^{(j)}\|_{\mathfrak S^2}
	\lesssim_j |\log\lambda|^{pj}R(\lambda,\varepsilon,N)^{1/4}.
\end{align}
Moreover, the quasi-free state satisfies
\begin{align}\label{eq:free-Q-tail}
	\lambda^j\|(\Gamma_0)_Q^{(j)}\|_{\mathfrak S^2}
	\leq \|Qh^{-1}\|_{\mathfrak S^2}^j
	\lesssim_j N^{-j},
\end{align}
which concludes the proof.
	
\end{proof}

\subsection{Trace class convergence of the relative one-body density matrix}

Now we prove convergence of the relative one-body density matrices in trace norm. We use a spectral cutoff to separate the low- and high-momentum contributions. For the low-momentum part,
we combine the quantitative Hilbert--Schmidt estimate from Section~\ref{sec:6.1} with the Cauchy--Schwarz inequality and a bound on the rank of the spectral projection.
For the quantum high-momentum tail, we use the weighted
Hilbert--Schmidt estimate \eqref{eqa priori bound on RDM}.
The cutoff is chosen so that both contributions tend to zero in trace norm.

Here our argument follows the strategy of \cite[Lemma~11.5]{LNR21}, with a few modifications. Specifically, due to the divergent upper bound of the relative entropy in \eqref{entropy}, we cannot take $P_L$ for $L$ fixed as in \cite{LNR21}, and simply applying the original truncation $P$ in the proof of the relative partition function is insufficient. Instead, we introduce a smaller truncation $P_{L_\lambda}$ to control $\fH(\Gamma_\lambda,\Gamma_0)$, and use the bound obtained in the proof of Hilbert-Schmidt convergence to conclude the result.

\begin{proof}[Proof of \eqref{thm-eq3}]
Let
	\begin{align*}
		X_\lambda:=\lambda(\Gamma_\lambda^{(1)}-\Gamma_0^{(1)}),\quad
		X_p^{\varepsilon,N}:=\int_{\cH}|u\rangle\langle u|\,
		(\ud\mu_p^{\varepsilon,N}-\ud\mu_0)(u),\quad
		X:=\int_{\cH}|u\rangle\langle u|\,
		(\ud\mu_p-\ud\mu_0)(u).
	\end{align*}
	Lemma~\ref{lem:Gaussian-moment-maps} and
	\eqref{1.2} for $m=1,2$ imply that $X_p^{\varepsilon,N}, X\in\mathfrak S^1(\cH)$ and
	\begin{align}\label{eq:classical-trace-convergence}
		\|X_p^{\varepsilon,N} - X\|_{\mathfrak S^1}
		\leq C\|(\cZ_p^{\varepsilon,N})^{-1}e^{-W_p^{\varepsilon,N}}-\cZ_p^{-1}e^{-V_p}\|_{L^2(\mu_0)}
		\longrightarrow0.
	\end{align}

We next obtain a quantitative Hilbert-Schmidt estimate under the cutoff \(P\). Taking \(k=1\) in \eqref{2.4} and
\eqref{6.14}, and using
\begin{align*}
	P X_p^{\varepsilon,N} P
	=\int_{P\cH}|u\rangle\langle u|\,
	(\ud\widetilde\mu_p^{\varepsilon,N}-\ud\mu_{0,P})(u),
\end{align*}
we obtain
\begin{align*}
	\left\|P\bigl(X_\lambda-X_p^{\varepsilon,N}\bigr)P\right\|_{\mathfrak S^2}
	\leq C\left(\lambda N^2|\log\lambda|^p
	+|\log\lambda|^pR(\lambda,\varepsilon,N)^{1/4}
	+\left\|\lambda(\Gamma_0)_P^{(1)}-Ph^{-1}\right\|_{\mathfrak S^2}\right).
\end{align*}
For the free state, we have
\begin{align*}
	\left\|\lambda(\Gamma_0)_P^{(1)}-Ph^{-1}P\right\|^2_{\mathfrak S^2}
	=&\left\|P\left(\frac{\lambda}{e^{\lambda h}-1}-h^{-1}\right)P
	\right\|^2_{\mathfrak S^2}\nonumber\\
	\lesssim&\int_1^{N^2+1}\left(\frac{\lambda}{e^{\lambda x}-1}-\frac1x\right)^2\ud x=\int_1^{N^2+1}\frac{(e^{\lambda x}-1-\lambda x)^2}{x^2(e^{\lambda x}-1)^2}\ud x\leq\lambda^2N^2,
\end{align*}
where we used $0\leq e^{\lambda x}-1-\lambda x\leq \lambda x(e^{\lambda x}-1)$ in the last step.
Therefore,
\begin{align}\label{eq:low-HS-rate}
	\left\|P\bigl(X_\lambda-X_p^{\varepsilon,N}\bigr)P\right\|_{\mathfrak S^2}\leq C(\lambda N^2|\log\lambda|^p+|\log\lambda|^p
	R(\lambda,\varepsilon,N)^{1/4}+\lambda N)=:a_\lambda.
\end{align}
Since $N\simeq \lambda^{-\frac18}$, by \eqref{R(lambda,N)}, we have $a_\lambda^{1/4}|\log\lambda|^p\to0.$
On the other hand, $0<\lambda\leq\frac12$ and
\(\varepsilon\leq(4R)^{-1}\) implies
\begin{align}\label{eq:a-lambda-vs-N}
	a_\lambda^{-1}\lesssim R(\lambda,\varepsilon,N)^{-\frac14}\lesssim N^{1/16}\varepsilon^{(3p-3)/4}=o(N^2).
\end{align}

Define
\begin{align*}
	L_\lambda:=a_\lambda^{-1},\qquad
	P_{L_\lambda}:=\1_{\{h\leq L_\lambda\}},\qquad
	Q_{L_\lambda}:=\1_\cH-P_{L_\lambda}.
\end{align*}
By \eqref{eq:a-lambda-vs-N}, \(PP_{L_\lambda}=P_{L_\lambda}P=P_{L_\lambda}\) for all
sufficiently small \(\lambda\). 
Moreover, \(X_\lambda\) and \(X_p^{\varepsilon,N}\) commute with the translation group on \(\cH\). 
Indeed, the many-body Hamiltonian $\HH_\lambda$ is invariant under the translations, hence so is \(\Gamma_\lambda\), and therefore
\(\Gamma_\lambda^{(1)}\) commutes with the one-body translation group. The same is true for \(\Gamma_0^{(1)}\). On the classical side, \(W_p^{\varepsilon,N}\) and $V_p$ are translation invariant.
Hence \(X_\lambda\) and \(X_p^{\varepsilon,N}\) are Fourier multipliers and commute with \(P\).
Consequently, by the Cauchy-Schwarz inequality, \eqref{eq:low-HS-rate} implies
\begin{align}\label{eq:low-trace-rate}
	\|P_{L_\lambda}(X_\lambda-X_p^{\varepsilon,N})\|_{\mathfrak S^1}=\|P_{L_\lambda}P(X_\lambda-X_p^{\varepsilon,N})P\|_{\mathfrak S^1}
	\leq (\Tr P_{L_\lambda})^{1/2}a_\lambda
	\lesssim a_\lambda^{1/2}
	\to0.
\end{align}
It remains to control the quantum high-energy tail. By \eqref{eqa priori bound on RDM} and H\"older's inequality, we have
\begin{align}\label{eq:high-quantum-tail}
	\|Q_{L_\lambda}X_\lambda\|_{\mathfrak S^1}
	&\leq \|h^{-\frac12}\|_{\mathfrak S^4}\|Q_{L_\lambda}h^{-1/2}\|_{\mathfrak S^4}
	\|h^{1/2}X_\lambda h^{1/2}\|_{\mathfrak S^2}\lesssim L_\lambda^{-1/4}|\log\lambda|^p
	=a_\lambda^{1/4}|\log\lambda|^p
	\to0.
\end{align}
Finally, \(X\in\mathfrak S^1\) and \(Q_{L_\lambda}\to0\) strongly, so $\|Q_{L_\lambda}X\|_{\mathfrak S^1}\to0.$
Combining \eqref{eq:classical-trace-convergence},
\eqref{eq:low-trace-rate}, and \eqref{eq:high-quantum-tail}, and using the exact decomposition
\[
X_\lambda-X
=P_{L_\lambda}(X_\lambda-X_p^{\varepsilon,N})
+P_{L_\lambda}(X_p^{\varepsilon,N}-X)
+Q_{L_\lambda}X_\lambda-Q_{L_\lambda}X,
\]
we conclude that
\begin{align*}
	\|X_\lambda-X\|_{\mathfrak S^1}
	\leq{}&\|P_{L_\lambda}(X_\lambda-X_p^{\varepsilon,N})\|_{\mathfrak S^1}
	+\|X_p^{\varepsilon,N}-X\|_{\mathfrak S^1}
	+\|Q_{L_\lambda}X_\lambda\|_{\mathfrak S^1}
	+\|Q_{L_\lambda}X\|_{\mathfrak S^1}
	\to 0.
\end{align*}
This proves \eqref{thm-eq3}.

\end{proof}

\appendix
\renewcommand{\appendixname}{Appendix}
\renewcommand{\theequation}{A.\arabic{equation}}
\section{Complex Wick expansion}\label{Appendix}
Fix $N\in\NN$, let \(u_N=P_Nu\) be the regularized complex Gaussian field. We recall that
\[
\mathbb E\big[u_N(x)\overline{u_N(y)}\big]=G_N(x-y),
\qquad
\mathbb E\big[u_N(x)u_N(y)\big]=0 .
\]
All Wick products in the following are taken with respect to this Gaussian field. 
In this appendix, we decompose the product $\prod_{k=1}^p \Wick{|u_N(x_k)|^2}$ into Wick-ordered monomials
corresponding to all possible contraction patterns. Then we  express these monomials by the red graphs in \(\fH^N(x_{1:p})\), which is introduced in Definition~\ref{def:graph}.

\begin{lem}\label{lem:complex-wick}
	Let $p\geq2$, \(N\in\NN\) and \(x_1,\ldots,x_p\in\Lambda\). Then,
	\begin{equation}\label{0.0}
		\prod_{k=1}^p \Wick{|u_N(x_k)|^2}
		=\sum_{\substack{I,J\subseteq[p]\\|I|=|J|}}\sum_{\sigma\in S_{I^c,J^c}}
		\Wick{\prod_{i\in I}u_N(x_i)\prod_{j\in J}\overline{u_N(x_j)}}\prod_{i\in I^c}G_N(x_i-x_{\sigma(i)}),
	\end{equation}
where $I^c=[p]\backslash I$, $J^c=[p]\backslash J$ and
\[S_{I^c,J^c}
:=\Bigl\{\sigma:I^c\to J^c; \sigma \text{ is a bijection and}
\;\sigma(i)\neq i \text{ for all } i\in I^c\Bigr\}.\]
If \(I^c=J^c=\varnothing\), the inner sum consists of the empty bijection and the empty product is understood to be \(1\). If no such bijection exists, the corresponding inner sum is \(0\).
Equivalently, in the graph notation of Definition~\ref{def:graph}, we have
\begin{align}\label{0.1}
	\prod_{k=1}^p \Wick{|u_N(x_k)|^2}
	=
	\sum_{\fg\in\fH^N(x_{1:p})}2^{m_\fg}\fg .
\end{align}
\end{lem}

\begin{proof}
	For real parameters $\lambda_{1:p}$ and $\mu_{1:p}$, set $$X_k:=\lambda_ku_N(x_k)+\mu_k\overline{u_N(x_k)},\quad 1\leq k\leq p.$$
	Define the normalized Wick exponentials by 
	$$\Phi(\lambda,\mu):=\frac{\exp\bigl(\sum_{k=1}^p X_k\bigr)}
	{\EE\bigl[\exp\bigl(\sum_{k=1}^p X_k\bigr)\bigr]},\quad \Phi_k(\lambda_k,\mu_k):=\frac{\exp(X_k)}{\EE[\exp(X_k)]},\, 1\leq k\leq p.$$ 
	By applying the real Wick ordering \cite[(A.1)]{FKSS25} to the real and imaginary parts of the Gaussian vector \((u_N(x_1),\ldots,u_N(x_p))\), and then rewriting the result in terms of \(u_N\) and \(\overline{u_N}\), for any subsets \(I,J\subseteq[p]\), we have
\begin{align}\label{A0}
	\left.\Wick{\prod_{i\in I}u_N(x_i)\prod_{j\in J}\overline{u_N(x_j)}}=
	\prod_{i\in I}\partial_{\lambda_i}\prod_{j\in J}\partial_{\mu_j}\Phi(\lambda,\mu)\right|_{\lambda=\mu=0}.
\end{align}
In particular, for every $1\leq k\leq p$, we have
\begin{align*}
\left.\partial_{\lambda_k}\partial_{\mu_k}\Phi_k(\lambda_k,\mu_k)\right|_{\lambda_k=\mu_k=0}=\Wick{|u_N(x_k)|^2}.
\end{align*}
Therefore,
\begin{equation}\label{A1} 
	\left.\prod_{k=1}^p\Wick{|u_N(x_k)|^2} = \prod_{k=1}^p\partial_{\lambda_k}\partial_{\mu_k}\left(\prod_{k=1}^p\Phi_k(\lambda_k,\mu_k)\right)\right|_{\lambda=\mu=0}.
\end{equation} 

Using \eqref{Green}, we have
\[\EE\left[\exp\left(\sum_{k=1}^p X_k\right)\right]
=\exp\left(\sum_{i,j=1}^p\lambda_i\mu_jG_N(x_i-x_j)\right),
\quad \prod_{k=1}^p\EE[e^{X_k}]
=\exp\left(\sum_{k=1}^p\lambda_k\mu_kG_N(0)\right).\]
Then,
\begin{align}\label{A3}
	\prod_{k=1}^p\Phi_k(\lambda_k,\mu_k)
	=\Phi(\lambda,\mu)\,\Psi(\lambda,\mu),\quad \Psi(\lambda,\mu)
	:=\exp\left(\sum_{1\le i\neq j\le p}
	\lambda_i\mu_jG_N(x_i-x_j)\right).
\end{align}
We apply $\prod_{k=1}^p\partial_{\lambda_k}\partial_{\mu_k}$ to the product $\Phi(\lambda,\mu)\Psi(\lambda,\mu)$. 
By Leibniz's rule, we choose a subset \(I\subseteq[p]\) of the $\lambda$-derivatives and a subset
\(J\subseteq[p]\) of the $\mu$-derivatives to hit \(\Phi\); the remaining derivatives hit \(\Psi\). 
Then, by \eqref{A1} and \eqref{A0}, 
\begin{align}\label{A2}
	\left.\prod_{k=1}^p\Wick{|u_N(x_k)|^2} =
	\sum_{I,J\subseteq[p]}\Wick{\prod_{i\in I}u_N(x_i)\prod_{j\in J}\overline{u_N(x_j)}}\prod_{i\in I^c}\partial_{\lambda_i}\prod_{j\in J^c}\partial_{\mu_j}\Psi(\lambda,\mu)\right|_{\lambda=\mu=0}.
\end{align}
The exponential \(\Psi\) contains only the bilinear monomials \(\lambda_i\mu_jG_N(x_i-x_j)\) with \(i\neq j\). Hence a non-zero derivative at the origin must pair every remaining \(\lambda_i\)-derivative with exactly one remaining \(\mu_j\)-derivative. This is possible only when \(|I^c|=|J^c|\), or equivalently \(|I|=|J|\). The condition \(i\neq j\) precisely removes the local self-contractions already subtracted by the local Wick ordering. Thus the admissible pairings are the bijections \(\sigma\in S_{I^c,J^c}\), and
\begin{align*}
	\left.
	\prod_{i\in I^c}\partial_{\lambda_i}
	\prod_{j\in J^c}\partial_{\mu_j}
	\Psi(\lambda,\mu)
	\right|_{\lambda=\mu=0}
	=\sum_{\sigma\in S_{I^c,J^c}}\prod_{i\in I^c}G_N(x_i-x_{\sigma(i)}),
\end{align*}
which, inserted into \eqref{A2}, proves \eqref{0.0}.

It remains to identify \eqref{0.0} with the graphical
expansion \eqref{0.1}. 
Fix a summand indexed by \(I,J\subset[p]\) and
\(\sigma\in\mathfrak S_{I^c,J^c}\) in the right-hand side of \eqref{0.0}. For \(i\in I\) (resp. $J$), we draw the uncontracted field \(u_N(x_i)\) (resp. \(\overline{u_N(x_j)}\)) as a red wavy half-edge at \(x_i\). For each \(i\in I^c\), the factor \(G_N(x_i-x_{\sigma(i)})\) is represented by a dashed edge joining \(x_i\) and \(x_{\sigma(i)}\). The condition \(\sigma(i)\ne i\) excludes dashed self-loops.
By construction, each of the two local field legs at every vertex is used exactly once, either as an uncontracted red wavy half-edge or as an endpoint of a dashed contraction edge. Hence the graph obtained from this summand is precisely one of the red contraction graphs in
\(\fH^N(x_{1:p})\).

Conversely, every graph \(\fg\in\fH^N(x_{1:p})\) arises from this construction.
Its red wavy lines determine the uncontracted sets \(I\) and \(J\). The remaining field legs must then be paired along the dashed edges. On each dashed-connected component containing red wavy lines, these pairings are uniquely determined by the orientations of the red wavy lines. On a component with no red wavy lines, every vertex has dashed degree two, so the component is a closed dashed loop.
If such a loop contains at least three vertices, it has two possible cyclic orientations, and both give the same unoriented graph because \(G_N(x-y)=G_N(y-x)\). Hence each such loop contributes a factor \(2\). A loop with two vertices does not give an additional multiplicity, since reversing the orientation gives the same contractions.
Therefore, exactly \(2^{m_\fg}\) summands in
\eqref{0.0} give the same graph \(\fg\), where \(m_\fg\) is
the number of dashed loops containing at least three vertices. Regrouping the terms according to the resulting red graph gives \eqref{0.1}.

\end{proof}

\renewcommand{\theequation}{B.\arabic{equation}}
\section{Green function estimates and properties of the effective kernels}\label{appB}

In this appendix, we collect the basic properties of the periodic Green function $G$, of its truncated version $G_N$, and of the kernels $f_l^\varepsilon$ defined in \eqref{def:F_l}.  
Throughout the appendix, we identify $G$ with its periodic extension to $\RR^2$. In particular, $G(x+n)=G(x)$ for every $x\in\RR^2$ and $n\in\ZZ^2$.  Since \(\Lambda=[-1/2,1/2)^2\), we have \(d(x)=|x|\) for \(x\in\Lambda\).

First, we recall from \cite[Lemma B.1]{FKSS25} that 
there exists a periodic function $\widetilde G$ such that, for a.e. $x\in\mathbb R^2$,
\begin{equation}
	G(x)=\frac{1}{2\pi}\log\frac{1}{d(x)} + \widetilde G(x), \label{C0}
\end{equation}
and \begin{align}\label{C0_1}
|\widetilde G(x)|+|\nabla\widetilde G(x)|\lesssim 1,\quad |\nabla^2 \widetilde G(x)|\lesssim 1+\log\frac{1}{d(x)},
\end{align}
where $d(x):=\min_{n\in \mathbb Z^2}|x+n|.$
Let $\widetilde \rho:=\rho*\rho$. For each $N\in\mathbb N$, the truncated Green function $G_N=\widetilde\rho_N*G$ is defined in \eqref{Green}, where
\begin{align*}
	\widetilde\rho_N = \sum_{n\in\ZZ^2}N^2\widetilde\rho(N(\cdot+n)).
\end{align*}
Since $\widehat \rho\in C_c^\infty(\mathbb R^2)$, we also have $\widehat{\widetilde \rho}\in C_c^\infty(\mathbb R^2)$, and hence $\widetilde \rho$ is rapidly decaying:
\begin{align}
	|\widetilde\rho(x)|\lesssim_m(1+|x|^2)^{-m},\qquad x\in\mathbb R^2,\ m\ge 1.\label{C2_1}
\end{align}
We shall use the following consequences for the periodization. If \(x\in\Lambda\) and \(|n|\geq1\), then \(|x+n|\geq |n|/4\). Therefore, for every \(m\geq2\),
\begin{align}
	|\widetilde\rho_N(x)|\leq N^2|\widetilde\rho(Nx)| +\sum_{|n|\geq1}N^2|\widetilde\rho(N(x+n))|\lesssim_m \frac{N^2}{(1+N^2|x|^2)^{m}} + \sum_{|n|\geq1}\frac{N^{2-2m}}{|n|^{2m}}
	\lesssim_m \frac{N^2}{(1+N^2|x|^2)^m}.\label{C2}
\end{align}
Moreover, for $m>2$, by the change of variable $y=N^{-1}z$, we have
\begin{align}
	\int_{\Lambda}|y|^2|\widetilde\rho_N(y)|\ud y\lesssim_m N^2\int_{\RR^2}\frac{|y|^2}{(1+N^2|y|^2)^m}\ud y=N^{-2}\int_{\RR^2}\frac{|z|^2}{(1+|z|^2)^m}\ud z\lesssim_m\frac{1}{N^2}.\label{C7}
\end{align}

The next lemma estimates the approximation error \(G_N-G\). It is the periodic analogue of \cite[Lemma B.3]{FKSS25}; we include the proof to keep track of the dependence on \(N\).

\begin{lem}\label{lem:G}
	For every $x\in\Lambda\backslash\{0\}$, it holds that
	\begin{align}
		|G_N(x)-G(x)|\lesssim&(1+|\log (N|x|)|)\wedge\frac{1}{N^2|x|^2}.\label{C3}
	\end{align} 
As a result, for any $M>N$, 
\begin{align}
		|G_N(x)-G_M(x)|\lesssim&(1+|\log (N|x|)|)\wedge\frac{1}{N^2|x|^2},\label{C3_1}
\end{align}
and \begin{align}
	|G_N(x)|\lesssim&(1+\log N)\wedge(1+|\log|x||).\label{C3_2}
\end{align}
\end{lem}

	\begin{proof}	
	To prove \eqref{C3}, we write
	$$G_N(x)-G(x) = \int_\Lambda (G(x-y)-G(x))\widetilde\rho_N(y)\ud y.$$
	We first consider the case $|x|\leq1/6$ and split $\Lambda$ into four regions: $$\Omega_1=\{|y|\leq|x|/2\},\quad \Omega_2=\{|x|/2<|y|<2|x|\},\quad \Omega_3=\{2|x|\leq|y|\leq 1/3\},\quad \Omega_4=\{|y|>1/3\}.$$	
	For $y\in\Omega_4$, \(d(x-y)\geq d(y)-d(x)\geq1/6\). Hence, by \eqref{C0},
	$$|G(x-y)-G(x)|\leq|G(x-y)|+|G(x)|\lesssim1+|\log|x||.$$
	Using \eqref{C2} with $m=2$,
	\begin{align*}
		\int_{\Omega_4}|G(x-y)-G(x)||\widetilde\rho_N(y)|\ud y\lesssim&(1+|\log|x||)\int_{\frac13\leq|y|\leq1}\frac{N^2}{(1+N^2|y|^2)^2}\ud y	\lesssim
		\frac{1+|\log |x||}{N^2}.
	\end{align*}
This is bounded by the right side of \eqref{C3} since $1+|\log t|\lesssim t^{-2}$ for $0<t\leq1$.
	
	Next, since $|x-y|\leq\frac12$ on $\Lambda\backslash\Omega_4$, by \eqref{C0} and \eqref{C2} with $m=2$, for $i=1,2,3$,
	\begin{align}\label{C6}
		\int_{\Omega_i} |G(x-y)-G(x)||\widetilde\rho_N(y)|\ud y\lesssim N^2\int_{\Omega_i}\frac{1+|\log(|x-y|/|x|)|}{(1+N^2|y|^2)^2}\ud y=\int_{N\Omega_i}\frac{1+\left|\log\left|\frac{z}{N|x|}-\frac{x}{|x|}\right|\right|}{(1+|z|^2)^2}\ud z.
	\end{align}
where we changed the variable $y=N^{-1}z$.
Consider first \(i=3\). On \(N\Omega_3\),
\[
1\leq\left|\frac{z}{N|x|}-\frac{x}{|x|}\right|
\leq \frac{3|z|}{2N|x|}.
\]
Thus
\begin{align}\label{C6_1}
	\int_{\Omega_3}|G(x-y)-G(x)|\,|\widetilde\rho_N(y)|\,\ud y
	&\lesssim
	\int_{|z|\geq2N|x|}
	\frac{1+\log\frac1{N|x|}+\log|z|}{(1+|z|^2)^2}\,\ud z
	\lesssim
	1+|\log(N|x|)|.
\end{align}
The same integral also satisfies the complementary bound. Indeed, setting \(z=N|x|w\),
\begin{align}\label{C6_2}
	&\int_{|z|\geq2N|x|}
	\frac{1+\log\frac1{N|x|}+\log|z|}{(1+|z|^2)^2}\,\ud z
	=
	\int_{|w|\geq2}
	\frac{N^2|x|^2(1+\log|w|)}{(1+N^2|x|^2|w|^2)^2}\,\ud w
	\lesssim
	\frac{1}{N^2|x|^2}.
\end{align}
For $i=2$, since $|y|\geq|x|/2$ on $\Omega_2$, we have 
$$N^2\int_{\Omega_2}\frac{1+|\log(|x-y|/|x|)|}{(1+N^2|y|^2)^2}\ud y\lesssim\frac{N^2}{(1+N^2|x|^2)^2}\int_{\Omega_2}(1+|\log(|x-y|/|x|)|)\ud y$$
Then, by the change of variable $y=|x|z'+x$, \eqref{C6} becomes
\begin{align*}
\int_{\Omega_2}|G(x-y)-G(x)||\widetilde\rho_N(y)|\ud y
	\lesssim\frac{N^2|x|^2}{(1+N^2|x|^2)^2}\int_{|z'|\leq3}(1+|\log|z'||)\ud z'\lesssim1\wedge\frac{1}{N^2|x|^2}.
\end{align*}
Combining this with \eqref{C6_1} and \eqref{C6_2} gives
\begin{align}\label{C8}
	\int_{\Omega_2\cup\Omega_3}|G(x-y)-G(x)||\widetilde\rho_N(y)|\ud y\lesssim(1+|\log (N|x|)|)\wedge\frac{1}{N^2|x|^2}.
\end{align}

Finally, on $\Omega_1$, since $\frac12|x|\leq|x-y|\leq\frac32|x|$, it follows that
$$N^2\int_{\Omega_1}\frac{1+|\log(|x-y|/|x|)|}{(1+N^2|y|^2)^2}\ud y\lesssim N^2\int_{|y|\leq1}(1+N^2|y|^2)^{-2}\ud y\leq\int_{\RR^2}(1+|z|^2)^{-2}\ud z\lesssim1.$$
Thus, by \eqref{C6},
$$\int_{\Omega_1}|G(x-y)-G(x)||\widetilde\rho_N(y)|\ud y\lesssim1.$$
Moreover, since $\widetilde\rho_N$ is even, we have
$$\int_{\Omega_1}(G(x-y)-G(x))\widetilde\rho_N(y)\ud y = \int_{\Omega_1}(G(x-y)-G(x)+y\cdot\nabla G(x))\widetilde\rho_N(y)\ud y.$$
By \eqref{C0_1}, for any $z\in\Lambda$, $$|\nabla^2G(z)|\lesssim1+|\log|z||+\frac{1}{|z|^2}\lesssim\frac{1}{|z|^2}.$$
Then, by the Taylor's formula, for any $y\in\Omega_1$,
\begin{align}
	|G(x-y)-G(x)+y\cdot\nabla G(x)|\leq|y|^2\int_0^1(1-t)|\nabla^2G(x-ty)|\ud t\lesssim|y|^2\int_0^1\frac{1-t}{|x-ty|^2}\ud t\lesssim\frac{|y|^2}{|x|^2},\label{C5}
\end{align}
where in the last inequality, we used that $|x-ty|\geq|x|/2$ for any $|t|<1$ and $y\in\Omega_1$.
Therefore, by \eqref{C7},
$$\left|\int_{\Omega_1}(G(x-y)-G(x))\widetilde\rho_N(y)\ud y\right| \lesssim\frac{1}{|x|^2}\int_{\Omega_1}|y|^2|\widetilde\rho_N(y)|\ud y\lesssim\frac{1}{N^2|x|^2}.$$
Together with \eqref{C8}, this proves \eqref{C3} when \(|x|\leq1/6\).

If $|x|>\frac16$, split $\Lambda$ into $\Omega_1'=\{y\in\Lambda:|y|\leq1/12\}$, and $\Omega_2'=\Lambda\backslash\Omega_1'$. 
On $\Omega_1'$, for any $0<t<1$, $d(x-ty)\geq d(x)-d(ty)>1/12$ and thus,
$$|\nabla^2G(x-ty)|\lesssim 1+|\log d(x-ty)|+\frac{1}{d(x-ty)^2}\lesssim1.$$
Arguing as above, by \eqref{C5} and \eqref{C7},
\begin{align*}
	\left|\int_{\Omega_1'}(G(x-y)-G(x))\widetilde\rho_N(y)\ud y\right|\leq\int_{\Omega_1'}\int_0^1(1-t)|\nabla^2G(x-ty)||y|^2|\widetilde\rho_N(y)|\ud t\ud y \lesssim\int_{\Omega_1'}|y|^2|\widetilde\rho_N(y)|\ud y\lesssim\frac{1}{N^2}.
\end{align*}
On $\Omega_2'$, by \eqref{C2} with $m=2$, $|\widetilde\rho_N|\lesssim N^{-2}$, while $|G(x)|\lesssim1$ by $d(x)=|x|>1/6$ and \eqref{C0}. Hence
\begin{align}\label{C7_2}
	\int_{\Omega_2'}|G(x-y)-G(x)||\widetilde\rho_N(y)|\ud y\lesssim N^{-2}\int_\Lambda(|G(x-y)|+1)\ud y\lesssim N^{-2}\int_\Lambda(|G(y)|+1)\ud y\lesssim N^{-2}.
\end{align}
Since $1/6\leq|x|\leq1$, 
this is bounded by the right side of \eqref{C3}. Hence \eqref{C3} holds for all \(x\in\Lambda\setminus\{0\}\).

For \eqref{C3_1}, set 
$$f(s) = (1+|\log s|)\wedge s^{-2}=\begin{cases}
	1+\log\frac1s, & 0< s\le 1,\\
	s^{-2}, & s>1.
\end{cases}$$ 
Since $f$ decreases on $(0,\infty)$, for any $M>N$, by \eqref{C3},
$$|G_N(x)-G_M(x)|\lesssim f(N|x|) + f(M|x|)\leq 2f(N|x|),$$
which is \eqref{C3_1}.

Finally, for \eqref{C3_2}, since $\supp\widehat\rho\subseteq\{|k|\leq1\}$ and $0\leq\widehat\rho\leq1$,
$$|G_N(x)|\leq\sum_{k\in\ZZ^2}\widehat\rho(N^{-1}k)^2\langle k\rangle^{-2}\leq\sum_{|k|\leq N}(1+|k|^2)^{-1}\lesssim1+\log N.$$
Moreover, if $|x|>1/N$, then by \eqref{C0} and \eqref{C3},
$$|G_N(x)|\leq|G_N(x)-G(x)|+|G(x)|\lesssim \frac{1}{N^2|x|^2} + 1+|\log|x||\lesssim1+|\log|x||,$$
which proves \eqref{C3_2}.
\end{proof}

The pointwise estimates in Lemma~\ref{lem:G} allow us to compare the singular kernel $G$ with its truncated approximation $G_N$. 
We now use these bounds to study the kernels $f_l^\varepsilon$ defined in \eqref{def:F_l} for $2\leq l\leq p-1$ and their truncated versions
\(f_l^{\varepsilon,N}\), obtained by replacing \(G\) with \(G_N\).
First, we claim that for each fixed \(\varepsilon>0\), as $N\to\infty$,
\[\|f_{l}^{\varepsilon,N}-f_{l}^{\varepsilon}\|_{L^1(\Lambda^{l-1})}\to0.
\]
Indeed, expanding $\prod_{1\le a<b\le p} G_N(x_a-x_b)^{m_{ab}}-\prod_{1\le a<b\le p} G(x_a-x_b)^{m_{ab}}$ for every $\fm\in\fM_{p,l}$, each summand contains exactly one factor $G_N-G$, while the remaining \(p-l-1\) Green factors are bounded by \(1+|\log d(\cdot)|\) thanks to Lemma~\ref{lem:G}. Since $\|v^\varepsilon\|_{L^\infty}\lesssim\varepsilon^{2-2p}$,by H\"older's inequality and Young's inequality, we obtain
\begin{align}
	\|f_{l}^{\varepsilon,N}-f_{l}^{\varepsilon}\|_{L^1}^2
	\lesssim&\left(\int_{\Lambda}\Pi_2v^\varepsilon(x)|G_N(x)-G(x)|^2\ud x\right)\left(\int_{\Lambda}\Pi_2v^\varepsilon(x)(1+|\log|x||^{2(p-l)})\ud x\right)\label{limit}\\
	\lesssim&\varepsilon^{4-4p}\int_{\Lambda}|G_N(x)-G(x)|^2\ud x\lesssim\varepsilon^{4-4p}N^{-2},\nonumber
\end{align} 
where in the last inequality, we used \eqref{C3} and \eqref{F_{N,M}}.

The next lemma proves the symmetry of \(f_l^\varepsilon\) and
\(f_l^{\varepsilon,N}\), together with the positivity of the Fourier transform of each component. For convenience, we
write $f_l^{\varepsilon,\infty}:=f_l^{\varepsilon}$ and $G_\infty:=G$.
For $1\leq N\leq\infty$, define the associated translation-invariant \(l\)-variable kernel by
\begin{align}
	F_l^{\varepsilon,N}(x_{1:l})
	:=f_l^{\varepsilon,N}(x_2-x_1,\ldots,x_l-x_1).\label{F}
\end{align}
Recall $\fM_{p,l}$ in \eqref{def:M}. For every \(\fm\in\fM_{p,l}\) with \(C_\fm\neq0\), set
\begin{align}
	F_{l,\fm}^{\varepsilon,N}(x_{1:l}):=&f_{l,\fm}^{\varepsilon,N}(x_2-x_1,\dots,x_l-x_1)\nonumber\\
	:=&
	\int_{\Lambda^{p-l}} v^\varepsilon(x_2-x_1,\dots,x_p-x_1)
	\prod_{1\leq a<b\leq p} G_N(x_a-x_b)^{m_{ab}}\,\ud x_{l+1:p}.\label{def:Fl,m}
\end{align}

\begin{lem}\label{lem:fl}
	For $2\leq l\leq p-1$, $\varepsilon>0$ and $1\leq N\leq\infty$, \(F_l^{\varepsilon,N}\) is symmetric in the variables \(x_1,\ldots,x_l\).
	Equivalently, for every \(\sigma\in S_l\),
\begin{align}\label{sym:f}
	F_l^{\varepsilon,N}(x_{1:l})
	=
	F_l^{\varepsilon,N}(x_{\sigma(1)},\ldots,x_{\sigma(l)}).
\end{align}
Moreover, for every \(k_{1:l}\in(\ZZ^2)^l\), $N\geq1$ and $\fm\in\fM_{p,l}$ with \(C_\fm\neq0\), it holds that
\begin{align}\label{f:positive Fourier}
	\fF(F_{l,\fm}^{\varepsilon})(k_{1:l})\geq\fF(F_{l,\fm}^{\varepsilon,N})(k_{1:l})\ge 0.
\end{align}
\end{lem}

\begin{proof}
	We begin from \eqref{sym:f}. Fix $1\leq N\leq\infty$.
	Define
	\[\Phi_l(x_{1:p}):=
	\sum_{\fm\in\fM_{p,l}}C_\fm\prod_{1\leq a<b\leq p} G_N(x_a-x_b)^{m_{ab}}.\]
Since $v$ satisfies \eqref{con:sym}, it is enough to prove that \(\Phi_l\) is symmetric in the first \(l\) variables.
	Fix \(\sigma\in S_l\), and extend it to \(\overline\sigma\in S_p\) by setting \(\overline\sigma(r)=\sigma(r)\) for \(1\leq r\leq l\), and \(\overline\sigma(r)=r\) for \(l<r\leq p\).
For an unordered pair \(\{a,b\}\), write \(m_{ab}:=m_{\min(a,b),\max(a,b)}\). 
For every $\fm\in\fM_{p,l}$, define \(\sigma\cdot\fm\in\fM_{p,l}\) by
\[
(\sigma\cdot\fm)_{ab}
:=
m_{\{\overline\sigma^{-1}(a),\overline\sigma^{-1}(b)\}},
\qquad 1\leq a<b\leq p.
\]
	Then \begin{align}
		\Phi_l(x_{\sigma(1)},\dots,x_{\sigma(l)},x_{l+1},\dots,x_p) =& \sum_{\fm\in\fM_{p,l}}C_\fm\prod_{1\leq a<b\leq p} G_N(x_{\overline\sigma(a)}-x_{\overline\sigma(b)})^{m_{ab}}\nonumber\\
		=&\sum_{\fm\in\fM_{p,l}}C_{\fm}\prod_{1\leq a<b\leq p} G_N(x_a-x_b)^{(\sigma\cdot\fm)_{ab}},\label{B1}
	\end{align}
	and the graph defined in \eqref{graph1} satisfies
	$$\fg_\fm(x_{\sigma(1)},\dots,x_{\sigma(l)},x_{l+1},\dots,x_p) = \fg_{\sigma\cdot\fm}(x_{1:p}).$$
	The coefficient \(C_{\fg_m(x_{1:p})}\) are invariant under relabelling of the first \(l\)
	vertices. Hence $C_\fm=C_{\sigma\cdot \fm},$
	and	\eqref{B1} becomes
	\begin{align*}
		\Phi_l(x_{\sigma(1)},\dots,x_{\sigma(l)},x_{l+1},\dots,x_p) =& \sum_{\fm\in\fM_{p,l}}C_{\sigma\cdot\fm}\prod_{1\leq a<b\leq p} G_N(x_a-x_b)^{(\sigma\cdot\fm)_{ab}}
		=\Phi_l(x_1,\dots,x_p).
	\end{align*}
	This completes the proof of \eqref{sym:f}.

	We now prove \eqref{f:positive Fourier}. 
Let $E_p:=\{(a,b):1\le a<b\le p\}$. Set \(y_0=0\) and \(y_j=x_{j+1}-x_1\), \(1\leq j\leq p-1\).
Then
\begin{align*}
	f_{l,\fm}^{\varepsilon,N}(y_{1:l-1})
	&=
	\int_{\Lambda^{p-l}}v^\varepsilon(y_{1:p-1})
	\prod_{\me=(a,b)\in E_p}
	G_N(y_{a-1}-y_{b-1})^{m_\me}\,\ud y_{l:p-1},
\end{align*}
where $m_\me=m_{ab}$ if $\me=(a,b)$.
For \(m_\me=0\), we interpret \(G_N^0=1\). 
Expanding all these factors into Fourier series gives
\begin{align*}
	f_{l,\fm}^{\varepsilon,N}(y_{1:l-1})
	&=
	\sum_{(q_\me)_{\me\in E_p}}
	\left(\prod_{\me\in E_p}\fF(G_N^{m_\me})(q_\me)\right)
	\int_{\Lambda^{p-l}}
	v^\varepsilon(y_{1:p-1})
	\prod_{r=1}^{p-1}e_{\mathbf q_r}(y_r)\,\ud y_{l:p-1},
\end{align*}
where $$\mathbf q_r=\sum_{\me=(a,b)\in E_p}(\1_{a=r+1}-\1_{b=r+1})q_\me,\quad 1\leq r\leq p-1.$$
Therefore, for $k_{1:l-1}\in(\ZZ^2)^{l-1}$, 
	\begin{align}\label{B5}
	\fF(f_{l,\fm}^{\varepsilon,N})(k_{1:l-1})
	&=
	\sum_{(q_\me)_{\me\in E_p}}
	\left(\prod_{\me\in E_p}\fF(G_N^{m_\me})(q_\me)\right)
	\widehat{v^\varepsilon}
	(k_1-\mathbf q_1,\ldots,k_{l-1}-\mathbf q_{l-1},-\mathbf q_l,
	\ldots,-\mathbf q_{p-1}).
\end{align}
For every $m\geq1$ and $q\in\ZZ^2$, we have
$$\mathcal F(G_N^m)(q)=\sum_{k_1+\cdots+k_m=q}\widehat{G_N}(k_1)\cdots \widehat{G_N}(k_m)=\sum_{k_1+\cdots+k_m=q}\prod_{i=1}^m(\widehat\rho(N^{-1}k_i)^2\langle k_i\rangle^{-2}).$$
Since $|\widehat\rho|\leq1$ and $\supp\widehat\rho\subseteq\{q\in\RR^2:|q|\leq1\}$, it holds that
\begin{align}
	0\leq \mathcal F(G_N^m)(q)\leq \1_{\{|q|\leq mN\}}\sum_{k_1+\cdots+k_m=q}\prod_{i=1}^m\langle k_i\rangle^{-2}=\1_{\{|q|\leq mN\}}\fF(G^m)(q)\lesssim\1_{\{|q|\leq mN\}},\label{GN:bound}
\end{align}
where we used $\fF(G^m)(q)\leq\int_\Lambda(1+|\log|x||)^m\ud x\lesssim1$ in the last inequality.
Since $\widehat v\geq0$, it follows from \eqref{B5} that
$0\leq\fF(f_{l,\fm}^{\varepsilon,N})\leq\fF(f_{l,\fm}^{\varepsilon})$.
This concludes the proof since $\fF(F_{l,\fm}^{\varepsilon,N})(k_{1:l})=\1_{k_1+\cdots+k_l=0}\fF(f_{l,\fm}^{\varepsilon,N})(k_{2:l})$.

\end{proof}

\renewcommand{\theequation}{C.\arabic{equation}}
\section{Estimates of the correlation functions}\label{appC}

In this appendix we recall two elementary estimates for the correlation functions associated with the Gaussian free field \(\mu_0\). For \(n\in\mathbb N\) and \(1\le r<\infty\), we denote by \(\mathfrak S^r(\mathfrak H^n)\) the
Schatten class
\[\mathfrak S^r(\mathfrak H^n)
:=
\{T\in\mathcal K(\mathfrak H^n): \Tr |T|^r<\infty\},
\quad
\|T\|_{\mathfrak S^r(\mathfrak H^n)}
:=(\Tr |T|^r)^{1/r},\]
where \(\cK(\mathfrak H^n)\) is the space of compact operators on \(\mathfrak H^n\), and \(|T|=(T^*T)^{1/2}\). Note that \(\|\cdot\|_{\mathrm{HS}}=\|\cdot\|_{\mathfrak S^2}\).
For \(m\ge1\) and \(g\in L^2(\mu_0)\), define the correlation function weakly by

	\begin{align}\label{eq:moment-map-definition}
		\mathcal M_m(g)
		:=\int|u^{\otimes m}\rangle\langle u^{\otimes m}|\,
		g(u)\,\ud\mu_0(u).
	\end{align}

	\begin{lem}\label{lem:Gaussian-moment-maps}
		For every \(m\geq1\) and \(g\in L^2(\mu_0)\), we have
		\begin{align}\label{eq:moment-map-bound}
			\|\mathcal M_m(g)\|_{\mathfrak S^2(\cH^m)}
			\leq C_m\|g\|_{L^2(\mu_0)}.
		\end{align}
		Moreover, if \(g,g'\in L^2(\mu_0)\) are probability densities, then
		\begin{align}\label{eq:relative-moment-trace-bound}
			\fM(g,g')
			:=\int|u\rangle\langle u|\,(g(u)-g'(u))\,\ud\mu_0(u)
			\in\mathfrak S^1(\cH)
		\end{align}
		and
		\begin{align}\label{eq:relative-moment-trace-estimate}
			\|\fM(g,g')\|_{\mathfrak S^1(\cH)}
			\leq \bigl(\Tr \,h^{-2}\bigr)^{1/2}
			\|g-g'\|_{L^2(\mu_0)}.
		\end{align}
\end{lem}

\begin{proof}
	
	Recall from \cite[Lemma 3.3]{LNR18} that for $m\geq1$,
	$$\fM_m(1)=\int|u^{\otimes m}\rangle\langle u^{\otimes m}|\,\ud\mu_0(u) = m!\,(h^{-1})^{\otimes m}.$$
	Then, since $h^{-1}$ is in $\mathfrak S^2(\cH)$, by the Cauchy-Schwarz inequality in $\mathfrak S^2$, for any self-adjoint finite-rank operator $A$ on $\cH^m$, we have
	$$\int|\langle u^{\otimes m},Au^{\otimes m}\rangle|^2\ud\mu_0(u)=\Tr((A\otimes A)\fM_{2m}(1))\leq (2m)!\|(h^{-1})^{\otimes(2m)}\|_{\mathfrak S^2(\cH^{2m})}\|A\|_{\mathfrak S^2(\cH^m)}^2\lesssim_m\|A\|_{\mathfrak S^2(\cH^m)}^2.$$
	Consequently, by Cauchy-Schwarz inequality in \(L^2(\mu_0)\), we get
	\begin{align}
		|\Tr(A\fM_m(g))|^2\leq \|g\|_{L^2(\mu_0)}^2\left(\int|\langle u^{\otimes m},Au^{\otimes m}\rangle|^2\ud\mu_0(u)\right)\lesssim_m\|g\|_{L^2(\mu_0)}^2\|A\|_{\mathfrak S^2(\cH^m)}^2.\label{M_m}
	\end{align}
	Now, let \(A\) be an arbitrary finite-rank operator and write $A=A_1+\iota A_2$ with
	$$A_1:=\frac{A+A^*}{2},
	\quad A_2:=\frac{A-A^*}{2\iota}.$$
	Both \(A_1\) and \(A_2\) are self-adjoint, and 
	$\|A\|_{\mathfrak S^2}^2=\|A_1\|_{\mathfrak S^2}^2
	+\|A_2\|_{\mathfrak S^2}^2.$ Thus \eqref{M_m} holds for every finite-rank operator.
	Since these operators are dense in
	\(\mathfrak S^2\), the self-duality of
	\(\mathfrak S^2\) yields a unique operator
	\(\mathcal M_m(g)\in\mathfrak S^2(\cH^m)\) satisfying
	\eqref{eq:moment-map-bound}.
	
	For \eqref{eq:relative-moment-trace-estimate}, let \(A\) be finite rank on \(\cH\).  Since
	\(\int(g-g')\,\ud\mu_0=0\),
	\begin{align*}
		\Tr\bigl(A\mathcal M(g,g')\bigr)
		=\int \Wick{\langle u,Au\rangle}\,(g(u)-g'(u))\,\ud\mu_0(u).
	\end{align*}
	Then, by the H\"older inequality, and \cite[(5.6)]{LNR21},
	$$|\Tr\bigl(A\mathcal M(g,g')\bigr)|^2\leq\|g-g'\|_{L^2(\mu_0)}^2\int\bigl|\Wick{\langle u,Au\rangle}\bigr|^2\ud\mu_0\leq \|g-g'\|_{L^2(\mu_0)}^2\Tr(h^{-2})\|A\|^2,$$
	where $\|\cdot\|$ is the operator norm.
	Since the finite-rank operators are dense in \(\mathfrak K(\cH)\), and 
	$\mathfrak K(\cH)^*=\mathfrak S^1(\cH),$ we obtain \eqref{eq:relative-moment-trace-estimate}.
\end{proof}

\textbf{Acknowledgements:} We thank Hao Liang for informing us about his work on nonlocal Hartree measures with three-body interactions. 
Part of this work was completed while the second author was visiting LMU Munich, and she would like to thank the members of the Department of Mathematics at LMU for the warm hospitality and support. 
P.T.N. acknowledges support from the European Research Council through the ERC Consolidator Grant RAMBAS (Project No. 10104424).
Z.Y. and X.Z. are grateful to the financial supports by National Key R \& D Program of China (No. 2022YFA1006300) and the financial supports of the NSFC (No. 12426205), and the financial supports in part by the NSFC (No. 12595281, 12288201) and the support by key Lab of Random Complex Structures and Data Science, Chinese Academy of Science.

\def\cprime{$'$} \def\ocirc#1{\ifmmode\setbox0=\hbox{$#1$}\dimen0=\ht0
	\advance\dimen0 by1pt\rlap{\hbox to\wd0{\hss\raise\dimen0
			\hbox{\hskip.2em$\scriptscriptstyle\circ$}\hss}}#1\else {\accent"17 #1}\fi}


\begin{thebibliography}{CKRG26}
	
	\bibitem[Bou96]{Bou96}
	J. Bourgain.
	Invariant measures for the 2D-defocusing nonlinear Schr\"odinger equation.
	{\em Comm. Math. Phys.} 176 (1996), no. 2, 421--445.
	
	\bibitem[CKRG26]{CKRG26}
	C. Caraci, A. Knowles, A. Ranallo, and P. Torres Giesteira.
	The Euclidean $\phi^4_2$ theory as a limit of an inhomogeneous Bose gas.
	arXiv:2603.12241, 2026.
	
	\bibitem[DNN25]{DNN25}
	A. Deuchert, P. T. Nam, and M. Napi\'orkowski.
	The Gibbs state of the mean-field Bose gas.
	arXiv:2501.19396, 2025.
	
	\bibitem[DNY24]{DNY24}
	Y. Deng, A. R. Nahmod, and H. Yue.
	Invariant Gibbs measures and global strong solutions for nonlinear Schr\"odinger equations in dimension two.
	{\em Ann. of Math.} (2) 200 (2024), no. 2, 399--486.
	
	\bibitem[FKSS17]{FKSS17}
	J. Fr\"ohlich, A. Knowles, B. Schlein, and V. Sohinger.
	Gibbs measures of nonlinear Schr\"odinger equations as limits of many-body quantum states in dimensions $d\leq 3$.
	{\em Comm. Math. Phys.} 356 (2017), no. 3, 883--980.
	
	\bibitem[FKSS22]{FKSS23}
	J. Fr\"ohlich, A. Knowles, B. Schlein, and V. Sohinger.
	The mean-field limit of quantum Bose gases at positive temperature.
	{\em J. Amer. Math. Soc.} 35 (2022), no. 4, 955--1030.
	
	\bibitem[FKSS25]{FKSS25}
	J. Fr\"ohlich, A. Knowles, B. Schlein, and V. Sohinger.
	The Euclidean $\phi^4_2$ theory as a limit of an interacting Bose gas.
	{\em J. Eur. Math. Soc.} 27 (2025), no. 11, 4399--4468.
	
	\bibitem[GIP15]{GIP15} M. Gubinelli, P. Imkeller, N. Perkowski. Paracontrolled distributions and singular PDEs, {\em Forum Math.Pi 3} no. 6, 2015.
	
	\bibitem[GJ87]{GJ87}
	J. Glimm and A. Jaffe.
	{\em Quantum Physics. A Functional Integral Point of View}.
	Second edition.
	Springer-Verlag, New York, 1987.
	
	\bibitem[Jan97]{Jan97}
	S. Janson.
	{\em Gaussian Hilbert Spaces}.
	Cambridge Tracts in Mathematics, vol. 129.
	Cambridge University Press, Cambridge, 1997.
	
	\bibitem[JR25]{JR25}
	L. Jougla and N. Rougerie.
	$\Phi^4_2$ theory limit of a many-body bosonic free energy.
	arXiv:2512.10704, 2025.
	
	
	\bibitem[LNR15]{LNR15}
	M. Lewin, P. T. Nam, and N. Rougerie.
	Derivation of nonlinear Gibbs measures from many-body quantum mechanics.
	{\em J. \'Ec. polytech. Math.} 2 (2015), 65--115.
	
	\bibitem[LNR18]{LNR18}
	M. Lewin, P. T. Nam, and N. Rougerie.
	Gibbs measures based on 1D (an)harmonic oscillators as mean-field limits.
	{\em J. Math. Phys.} 59 (2018), no. 4, 041901.
	
	\bibitem[LNR21]{LNR21}
	M. Lewin, P. T. Nam, and N. Rougerie.
	Classical field theory limit of many-body quantum Gibbs states in 2D and 3D.
	{\em Invent. Math.} 224 (2021), no. 2, 315--444.
	
	\bibitem[Lia26]{Lia26} H. Liang. Nonlocal cubic density Gibbs measures from bosonic Gibbs states with three-body interactions. Preprint 2026.
	
	\bibitem[LNZ26]{LNZ26}
	L. L\"u, P. T. Nam, and R. Zhu.
	Derivation of the focusing $\Phi^6_1$ measure in the optimal mass regime from many-body quantum Gibbs states.
	arXiv:2605.25755, 2026.
	
	\bibitem[Nel66]{Nel66}
	E. Nelson.
	A quartic interaction in two dimensions.
	In {\em Mathematical Theory of Elementary Particles}
	(Proc. Conf., Dedham, Mass., 1965), pp. 69--73.
	M.I.T. Press, Cambridge, Mass., 1966.
	
	\bibitem[Nel73]{Nel73}
	E. Nelson.
	Construction of quantum fields from Markoff fields.
	{\em J. Funct. Anal.} 12 (1973), 97--112.
	
	\bibitem[NZZ25]{NZZ25}
	P. T. Nam, R. Zhu, and X. Zhu.
	$\Phi^4_3$ theory from many-body quantum Gibbs states.
	arXiv:2502.04884, 2025.
	
	\bibitem[NZZ26]{NZZ26}
	P. T. Nam, R. Zhu, and X. Zhu.
	Derivation of Gibbs measure from Gibbs state with the fractional-Bessel interaction in two dimensions.
	arXiv:2604.21583, 2026.

	
	\bibitem[OT18]{OT18}
	T. Oh and L. Thomann.
	A pedestrian approach to the invariant Gibbs measures for the 2D defocusing nonlinear Schr\"odinger equations.
	{\em Stoch. Partial Differ. Equ. Anal. Comput.} 6 (2018), 397--445.
	
	\bibitem[RS25]{RS25}
	A. Rout and V. Sohinger.
	A microscopic derivation of Gibbs measures for the 1D focusing quintic nonlinear Schr\"odinger equation.
	{\em SIAM J. Math. Anal.} 57 (2025), no. 5, 4680--4755.
	
	\bibitem[Sim74]{Sim74}
	B. Simon.
	The $P(\phi)_2$ Euclidean (Quantum) Field Theory.
	Princeton Series in Physics.
	Princeton University Press, Princeton, NJ, 1974.
	
	
\end{thebibliography}
\end{document}